\newif\ifsubmit    
\newif\ifllncs      
\newif\ifexabs      
\newif\ifblind 
\newif\ifacm

\submittrue

\def\qip{1}

\ifllncs
\documentclass[runningheads,a4paper]{llncs}

\else \documentclass[letterpaper,11pt,pdfa]{article}
  \usepackage[in]{fullpage}
\fi

\usepackage{iftex}
\ifPDFTeX
  \usepackage[utf8]{inputenc}
  \usepackage[noTeX]{mmap}
  \usepackage[T1]{fontenc}
\fi
\ifLuaTeX
\usepackage{luatex85}
\usepackage[noTeX]{mmap}
\fi

\usepackage{mdframed}
\usepackage{amsmath}
\usepackage{amsfonts}
\usepackage{amssymb}
\usepackage{amsthm}
\usepackage{color}
\usepackage[dvipsnames]{xcolor}
\usepackage{mathtools}
\usepackage{bbm}

\usepackage{mdframed}
\usepackage{amsmath}
\usepackage{amsfonts}
\usepackage[nottoc,numbib]{tocbibind}
\usepackage{amssymb}
\usepackage{amsthm}
\usepackage{color}      
\usepackage[dvipsnames]{xcolor}
\definecolor{DarkBlue}{RGB}{0,0,150}
\definecolor{NotSoDarkBlue}{RGB}{15,15,210}
\definecolor{DarkRed}{RGB}{150,0,0}
\definecolor{DarkGreen}{RGB}{0,100,0}
\usepackage[pdfstartview=FitH,pdfpagemode=UseNone,colorlinks,linkcolor=DarkRed,filecolor=blue,citecolor=DarkRed,urlcolor=DarkRed,pagebackref,breaklinks]{hyperref}
\usepackage[nameinlink]{cleveref}
\usepackage{tikz}
\usetikzlibrary{arrows.meta, positioning}

\ifllncs

  \spnewtheorem{claim}{Claim}{\bfseries}{\rmfamily}
  \crefname{claim}{claim}{claims}
  \Crefname{claim}{Claim}{Claims}
\else
  \newtheorem{theorem}{Theorem}[section]
  \newtheorem{definition}[theorem]{Definition}
  \newtheorem{remark}[theorem]{Remark}
  \newtheorem{lemma}[theorem]{Lemma}
  \newtheorem{corollary}[theorem]{Corollary}
  \newtheorem{proposition}[theorem]{Proposition}
  \newtheorem{claim}[theorem]{Claim}
  \newtheorem*{remark*}{Remark}

  \newtheorem*{theorem*}{Theorem}
  \newtheorem*{lemma*}{Lemma}

\fi

\usepackage{appendix}
\usepackage{algorithm}
\usepackage{algpseudocode}
\usepackage{braket}
\usepackage{comment}
\usepackage{url}
\usepackage{multicol}
\usepackage{tikz}
\usetikzlibrary{arrows,automata,positioning}
\usepackage{caption}
\usepackage[caption=false]{subfig}
\usepackage[font=small,labelfont=bf]{caption}
\usepackage{comment}
\usepackage{pifont}
\usetikzlibrary{patterns}

\usepackage{enumitem}
\setlist[description]{noitemsep}
\setlist[enumerate]{noitemsep}
\setlist[itemize]{noitemsep}

\usepackage{braket}

\definecolor{RoyalBlue}{RGB}{0,0,150}
\definecolor{NotSoDarkBlue}{RGB}{15,15,210}
\definecolor{DarkRed}{RGB}{150,0,0}
\definecolor{DarkGreen}{RGB}{0,100,0}

\ifsubmit
\newcommand{\authnote}[3]{}
\newcommand{\bhaskar}[1]{}
\newcommand{\justin}[1]{}
\newcommand{\jiahui}[1]{}
\newcommand{\vnote}[1]{}
\newcommand{\anote}[1]{}
\newcommand{\note}[1]{}
\else
\newcommand{\authnote}[3]{\textcolor{#3}{[{\footnotesize {\bf #1:} { {#2}}}]}}
\newcommand{\bhaskar}[1]{\authnote{Bhaskar}{#1}{OliveGreen}}
\newcommand{\justin}[1]{\authnote{Justin}{#1}{RoyalBlue}}
\newcommand{\jiahui}[1]{\authnote{Jiahui}{#1}{ForestGreen}}
\newcommand{\vnote}[1]{\authnote{Vinod}{#1}{blue}}
\newcommand{\anote}[1]{\authnote{Aparna}{#1}{Magenta}}
\newcommand{\note}[1]{\authnote{Note}{#1}{brown}}
\fi

\newcommand{\cA}{\mathcal{A}}
\newcommand{\cB}{\mathcal{B}}

\newcommand{\cD}{\mathcal{D}}
\newcommand{\cF}{\mathcal{F}}

\newcommand{\cH}{\mathcal{H}}

\newcommand{\cO}{\mathcal{O}}
\newcommand{\cQ}{\mathcal{Q}}
\newcommand{\cR}{\mathcal{R}}

\newcommand{\cV}{\mathcal{V}}
\newcommand{\cX}{\mathcal{X}}
\newcommand{\cY}{\mathcal{Y}}

\newcommand{\vk}{\mathsf{vk}}
\newcommand{\sigk}{\mathsf{sigk}}

\newcommand{\calA}{\mathcal{A}}
\newcommand{\calB}{\mathcal{B}}
\newcommand{\calC}{\mathcal{C}}
\newcommand{\calD}{\mathcal{D}}

\newcommand{\calF}{\mathcal{F}}

\newcommand{\calH}{\mathcal{H}}

\newcommand{\calM}{\mathcal{M}}

\newcommand{\calO}{\mathcal{O}}
\newcommand{\calP}{\mathcal{P}}
\newcommand{\calQ}{\mathcal{Q}}
\newcommand{\calR}{\mathcal{R}}

\newcommand{\calV}{\mathcal{V}}

\newcommand{\calX}{\mathcal{X}}
\newcommand{\calY}{\mathcal{Y}}
\newcommand{\calZ}{\mathcal{Z}}

\newcommand{\bbF}{\mathbb{F}}

\newcommand{\bbN}{\mathbb{N}}

\newcommand{\ketbra}[1]{\ket{#1}\bra{#1}}
\newcommand{\Tr}{\mathsf{Tr}}
\newcommand{\secp}{\lambda}
\newcommand{\secpar}{\secp}
\newcommand{\negl}{\mathsf{negl}}

\newcommand{\Eval}{\mathsf{Eval}}
\newcommand{\eval}{\mathsf{Eval}}
\newcommand{\getsr}{\overset{\$}{\leftarrow}}

\newcommand{\Cs}{\mathcal{C}}

\newcommand{\setup}{\mathsf{Setup}}

\newcommand{\prove}{\mathsf{Prove}}

\newcommand{\delegate}{\mathsf{Delegate}}

\newcommand{\gen}{\mathsf{Generate}}

\newcommand{\crs}{\mathsf{CRS}}

\newcommand{\Enc}{\mathsf{Enc}}

\newcommand{\Dec}{\mathsf{Dec}}

\newcommand{\sk}{\mathsf{sk}}

\newcommand{\msk}{\mathsf{msk}}

\newcommand{\ek}{\mathsf{ek}}

\newcommand{\pk}{\mathsf{pk}}

\newcommand{\LG}{\mathsf{LearningGame}}

\newcommand{\GLG}{\mathsf{GenLearningGame}}

\newcommand{\GOTP}{\mathsf{GenOTP}}

\newcommand{\ct}{\mathsf{ct}}

\newcommand{\chk}{\mathsf{check}}

\newcommand{\generate}{\mathsf{Generate}}

\newcommand{\ans}{\mathsf{ans}}

\newcommand{\out}{\mathsf{out}}

\newcommand{\SKE}{\mathsf{SKE}}

\newcommand{\prf}{\mathsf{PRF}}
\newcommand{\cprf}{\mathsf{cPRF}}

\newcommand{\constrain}{\mathsf{Constrain}}
\newcommand{\puncture}{\mathsf{Puncture}}
\newcommand{\constraineval}{\mathsf{ConstrainEval}}
\newcommand{\invert}{\mathsf{Invert}}

\newcommand{\keygen}{\mathsf{KeyGen}}
\newcommand{\kgen}{\mathsf{KeyGen}}

\newcommand{\samp}{\mathsf{Samp}}

\newcommand{\qhe}{\mathsf{QHE}}

\newcommand{\iO}{\mathsf{iO}}

\newcommand{\shO}{\mathsf{shO}}

\newcommand{\CC}{\mathsf{CC}}

\newcommand{\obf}{\mathsf{Obf}}

\newcommand{\aux}{\mathsf{aux}}

\newcommand{\ccobf}{{\sf CC.Obf}}

\newcommand{\tildeP}{\tilde{P}}

\newcommand{\Hyb}{\mathsf{Hyb}}

\newcommand{\Accept}{\mathsf{Acc}}

\newcommand{\concat}{\|}
\newcommand{\adv}{\mathcal{A}}

\newcommand{\OTP}{\mathsf{OTP}}
\newcommand{\Sim}{\mathsf{Sim}}

\renewcommand{\vec}[1]{\mathbf{#1}} 

\newcommand{\N}{\mathbb{N}}

\newcommand{\F}{\mathbb{F}}
\newcommand{\C}{\mathbb{C}}

\newcommand{\bit}{\{0,1\}}

\newcommand{\poly}{\mathsf{poly}}

\newcommand{\inp}{\mathsf{inp}}

\newcommand{\bfA}{\mathbf{A}}

\def\vecv{\mathbf{v}}

\def\vecx{\mathbf{x}}

\def\vecy{\mathbf{y}}

\newcommand{\Decomp}{\mathsf{Decomp}}
\newcommand{\CO}{\mathsf{CO}}

\newcommand{\Sign}{\mathsf{Sign}}
\newcommand{\Verify}{\mathsf{Verify}}
\newcommand{\Gen}{\mathsf{Gen}}

\newcommand{\PRF}{\mathsf{PRF}}

\newcommand{\norm}[1]{\left\|#1\right\|}

\newcommand\QM{\mathsf{QM}}
\newcommand\Generate{\mathsf{Generate}}
\newcommand\Evaluate{\mathsf{Evaluate}}

\newcommand\Ver{\mathsf{Ver}}

\newcommand{\Distr}{\mathsf{Distr}}
\newcommand{\SEQ}{\mathsf{SEQ}}

\title{Quantum One-Time Programs, Revisited}
\ifblind
\author{}
\else
\author{Aparna Gupte\\MIT \and 
Jiahui Liu\thanks{Part of this work done while at Fujitsu Research.}\\MIT \and 
Justin Raizes\\CMU \and 
Bhaskar Roberts\\UC Berkeley \and
Vinod Vaikuntanathan\\MIT}
\fi
\begin{document}

\ifexabs
\else
\maketitle
\thispagestyle{empty}
\begin{abstract}
    One-time programs (Goldwasser, Kalai and Rothblum, \ifacm 
    \\ \else\fi CRYPTO 2008) are programs that can be run on any single input of a user's choice, but not on a second input. Classically, they are unachievable without trusted hardware, but the destructive nature of quantum measurements seems to provide an alternate path to constructing them. Unfortunately, Broadbent, Gutoski and Stebila (CRYPTO 2013) showed that even with quantum techniques, 
    a strong notion of one-time programs, similar to ideal obfuscation, cannot be achieved for any non-trivial quantum function. On the positive side, Ben-David and Sattath (Quantum, 2023) showed how to construct a quantum one-time program for a certain (probabilistic) digital signature scheme, under a weaker notion of one-time program security. There is a vast gap between achievable and provably impossible notions of one-time program security, and it is unclear what functionalities are one-time program\ifacm 
    -\else\fi mable and which are not, under the achievable notions of security.
     
     In this work, we present new, meaningful, yet achievable definitions of one-time program security for {\em probabilistic} classical functions. We show how to construct one time programs satisfying these definitions for all functions in the classical oracle model and for constrained pseudorandom functions in the plain model. Finally, we examine the limits of these notions: we show a class of functions which cannot be one-time programmed in the plain model, as well as a class of functions which appears to be highly random given a single query, but whose quantum one-time program  leaks the entire function even in the oracle model.
\end{abstract}
\newpage 
\pagenumbering{roman}
\setcounter{tocdepth}{2} 
\tableofcontents
\newpage 

\pagenumbering{arabic}
\section{\ifexabs Extended Abstract \else Introduction \fi}
\label{sec:intro}
\sloppy The notion of one-time programs, first proposed by Goldwasser, Kalai and Rothblum~\cite{goldwasser2008one}, allows us to compile a program into one that can be run on a single input of a user's choice, but only one. If realizable, one-time programs would have wide-ranging applications in software protection, digital rights management, electronic tokens and electronic cash. Unfortunately, one-time programs immediately run into a fundamental barrier: software can be copied multiple times at will, and therefore, if it can be run on a single input of a user's choice, it can also be run on as many inputs as desired.

To circumvent this barrier, \cite{goldwasser2008one} designed a one-time program with the assistance of a specialized stateful hardware device that they called a \textit{one-time memory}. A one-time memory is a device instantiated with two strings $(s_0, s_1)$; it takes as input a choice bit $b \in \{0,1\}$, outputs $s_b$ and then {\em self-destructs}. Using one-time memory devices, Goldwasser et al.~showed how to compile any program into a one-time program, assuming one-way functions exist. Goyal et al.~\cite{DBLP:conf/tcc/GoyalISVW10} extended these results by achieving unconditional security against malicious parties and using a weaker type of one-time memories that store single bits. Notwithstanding these developments, the security of these schemes rests on shaky grounds: security relies on how much one is willing to trust the impenetrability of these hardware devices in the hands of a motivated and resourceful adversary who may be willing to mount sophisticated side-channel attacks. Which brings up the motivating question of our paper: {\em is there any other way to construct one-time programs}?

One might hope that the quantum no-cloning theorem~\cite{WoottersZurek} might give us a solution. The no-cloning theorem states that quantum information cannot be generically copied, so if one can encode the given program into an appropriate quantum state, one might expect to circumvent the barrier. However, there is a simple impossibility result by Broadbent, Gutoski and Stebila~\cite{broadbent2013quantum} that rules out {quantum} one-time versions of any \textit{deterministic} program. Indeed, given a candidate quantum one-time program state $\ket{\psi_f}$, an adversary can evaluate $f$ many times on different inputs as follows: it first evaluates the program on some input $x$, measures the output register to obtain $f(x)$. Since $f$ is deterministic, the measurement does not disturb the state of the program at all (if the computation is perfectly correct). The adversary then uncomputes the first evaluation, restoring the initial program state. She can repeat this process on as many inputs as she wishes.

While this impossibility result rules out one-time programs for deterministic functionalities, it raises the following natural question:
\begin{center}
    \textit{Can we obtain one-time programs for randomized functionalities?}
\end{center}
More concretely, can we construct quantum one-time programs for randomized functions $f:\calX \times \calR \rightarrow \calY$ that lets the user choose the input $x \in \calX$ but not the randomness $r \in \calR$? One might hope that by forcing the evaluation procedure to utilize the {\em inherent randomness} of quantum information in sampling $r \leftarrow \calR$, measuring the output would collapse the program state in a way that does not allow further evaluations. However, once again, \cite{broadbent2013quantum} showed that it is impossible to compile \textit{any quantum channel} into a one-time program, unless it is (essentially) learnable with just one query. This is a much more general impossibility; in fact, it rules out non-trivial one-time programs for classical randomized functions. (We refer the reader to Section~\ref{sec:tech_overview} for a description of this impossibility result.) 

On the other hand, more recently, Ben-David and Sattath~\cite{ben2023quantum} demonstrated the first instance of a one-time program for a certain randomized function. In particular, they construct a digital signature scheme where the (randomized) signing procedure can be compiled into a one-time program that the user can use to generate a single signature for a message of her choice.

At a first glance, this positive result might seem like a contradiction to the \cite{broadbent2013quantum} impossibility; however, that is not so, and the difference lies in which definition of one-time programs one achieves. Ben-David and Sattath~\cite{ben2023quantum} achieve a much weaker notion of one-time security than what was proven to be impossible by \cite{broadbent2013quantum}. On the one hand, \cite{broadbent2013quantum} demanded that an adversarial user should not be able to do \textit{anything} other than evaluate the one-time program on a single input, an ideal obfuscation-like guarantee~\cite{AC:Hada00,barak2001possibility}. On the other hand, the positive result of \cite{ben2023quantum} only claimed security in the sense that an adversarial user cannot output two different valid signatures. 

The starting point of this paper is that there is a vast gap between these two security notions. Within the gap, one could imagine several meaningful and useful intermediate notions of quantum one-time programs for classical randomized functions. For example, strengthening the \cite{ben2023quantum} definition, one could imagine requiring that the user should not even be able to verify the correctness of two input-output pairs (and not just be unable to produce them). Such a definition is a meaningful strengthening in the context of indistinguishability games (such as in pseudorandom functions) rather than unpredictability games (such as in digital signatures).
One could also imagine realizing one-time programs for a wider class of functions than the signature tokens of \cite{ben2023quantum}.

In this work, we revisit notions of quantum one-time programs and make progress on these questions. We propose a number of security notions of quantum one-time programs for randomized functions; give constructions both in the plain model and a classical oracle model; and examine the limits of these notions by showing negative results. 
We next describe our contributions in more detail.

\subsection{Our Results}
\label{sec:our_results}

\paragraph{Definitions.} Our first contribution is definitional. We give correctness and security definitions for one-time programs of classical randomized circuits, which we call one-time sampling programs. 

For correctness of a one-time sampling program for a classical $f$, any honest user can choose its own input $x$ and the evaluation gives $f(x;r)$ for some random $r$. For security, we lay out a list of different notions of security that we might desire from the one-time sampling program.

We make a few attempts on a simulation-based definition: the desired one-time sampling program functionality should be indistinguishable from an idealized functionality, where we are allowed to make a single quantum query to a ''randomized oracle'' for the target functionality $f$. 
However, these definitions run into several strong impossibility results, unless assuming hardware assumptions. \ifnum\qip=1 (See our technical manuscript Figure 1 and discussions in technical overview Section 2.1). \else \fi

We therefore explore a possible weakening on the single quantum query access we allow in the ideal world. Inspired by the compressed oracle technique in \cite{zhandry19compressed} used to record queries for quantum random oracles, we re-define the single-query access oracle in the ideal world. Very informally, the randomized oracle would record queries so that it allows only one "informative" query to be made, but potentially many more dummy queries. We then give a new simulation-based definition based on this oracle we call single-\emph{effective}-query oracle\footnote{We refer to the traditional single query oracle which allows literally one query as single-\emph{physical}-query oracle.}, which allows us to bypass the above impossibility results.

We additionally introduce a weaker but highly useful security definition called operational definition, in which the adversary cannot "evaluate" twice given a one-time program\footnote{Throughout the work, we may use the terms "one-time programs" and "one-time sampling programs" interchangeably. But they both refer to one-time sampling programs unless otherwise specified.}.

\paragraph{Constructions.} 
We give a very generic construction for one-time sampling programs in the classical oracle model\footnote{A classical oracle is a classical circuit that can be accessed coherently by quantum users in a black-box manner.}, inspired by the one-time signature scheme in \cite{ben2023quantum}.
We allow an honest user to choose its own input and then generate a random string by measuring a ``signature token'' state. The evaluation is on the user's input together with this freshly generated randomness.

In particular, an honest evaluator does not need to run a classical circuit coherently on a quantum input, but only needs quantum memory and one measurement to evaluate the program in our construction. 
However, our construction has security even against adversaries who can evaluate large-depth quantum circuits coherently on quantum states. In other words, our construction can be honestly evaluated using a relatively weak quantum computer, but still has security against adversaries using a full quantum computer.

We prove its security under the single-effective-query \ifacm \\ \else\fi simulation-based definition.

\begin{theorem}(Informal)
  \label{thm:construction_informal}
  There exists a secure one-time sampling program for all functions in the classical oracle model, with respect to our simulation-based, single-effective-query model one-time sampling security.
\end{theorem}

We also instantiate the classical oracle using indistinguishability obfuscation, to get a compiler in the plain model, and prove its security for the class of pseudorandom functions under an operational security definition for cryptographic functionalities.

\begin{theorem}[Informal]
    Assuming post-quantum iO and LWE (or alternatively subexponentially secure iO and OWFs), there exists one-time sampling programs for constrained PRFs. 
\end{theorem}

\paragraph{Impossibilities.} To complement our constructions in the classical oracle model and the plain model, we also give two new negative results. The first negative result shows we cannot hope to one-time program \emph{all randomized}  functionalities in the plain model, even under the weakest possible operational security definitions. This impossibility is inspired by the work of \cite{ananth2020secure,alagic2020impossibility}. We tweak the idea to work with randomized circuits that can only be evaluated once. 

\begin{theorem}[Informal]
    Assuming LWE and quantum FHE, there exists a family of circuits with high min-entropy outputs but no secure one-time sampling programs exist for them. 
\end{theorem}

We also show that having  high min-entropy outputs is not a sufficient condition to have a secure one-time programs.
Our second impossibility result shows that there exists a family of randomized functions with high min-entropy and is unlearnable under a single physical query. But it cannot be one-time programmed even in the classical
oracle model, even under the weakest possible operational security definitions\footnote{This function is securely one-time programmable under the single-effective-query simulation-based definition, but in a "meaningless" sense since both the simulator and the real-world adversary can fully learn the functionality. This demonstrates the separations and relationships between several of our definitions.}. 
\ifexabs \else 

We demonstrate the definitions presented in this work and their corresponding impossibilities and/or constructions in \Cref{fig:definitions_impossibilities_figure}. We recommend the readers to come back to this figure after going through the technical overview.
\fi


\paragraph{Applications.} Using the techniques we developed for one-time programs, we construct the following one-time cryptographic primitives:
\begin{itemize}
    \item \textbf{One-Time Signatures.} We compile a wide class of existing signature schemes to add signature tokens, which allow a delegated party to sign exactly one message of their choice. Notably, our construction only changes the signing process  while leaving the verification almost unmodified, unlike \cite{ben2023quantum}'s construction. Thus, it enables signature tokens for \emph{existing} schemes with keys which are already distributed.
    \item \textbf{One-Time NIZK Proofs.} We show how a proving authority can delegate to a subsidiary the ability to non-interactively prove a single (true) statement in \ifacm \\ \else\fi
    zero-knowledge.
    \item \textbf{Public-Key Quantum Money.} We show that one-time programs satisfying a mild notion of security imply public-key quantum money.
\end{itemize}

\paragraph{Future Work}
For future work, some possible directions and applications may be one-time MPC, one-time encryptions with advanced security. As discussed in a concurrent work \cite{gunn2024quantum}, another interesting application is one-time tokens for large language models. 

\ifexabs
\else
\ifacm 
\begin{table*}
\caption{Definitions with Impossibilities and Constructions.
The exact impossibility results and positive results for operational definitions depend on which definition of single-query model we work with. 
\ifacm See the full version for details \else See \Cref{sec:unlearnability_defs}, \Cref{sec:operational_defs}, and \Cref{sec:impossibility} for details. \fi}
\centering
\begin{center}
\begin{tabular}{||c | c | c ||} 
 \hline
 Definition & Impossibilities & Construction   \\ [0.5ex] 
 \hline\hline
 Single physical query, & Strong impossibility  & For single physical   \\ 
 quantum-output, simulation-based  & in oracle model  & query learnable   \\
  \ifacm\else(\Cref{def:single-query-simulation-security})\fi & (\Cref{sec:tech_overview_definitional}) & (trivial) functions only \\
   & & \cite{broadbent2013quantum} \\
 \hline
 Single physical query, & Impossibility for generic & N/A  \\
classical-output, simulation-based  & construction in oracle model  &   \\
 \ifacm\else(\Cref{def:single-query-classical-output-simulation-security})\fi & (\Cref{par:impossibility}\ifacm\else, \Cref{sec:impossibility_oracle_model}\fi) &   \\
 \hline
 Single effective query & Impossibility for generic  & For all functions  \\
  quantum-output, & constructions in plain  & (with proper randomness   \\
simulation-based \ifacm\else(\Cref{def:simulation-style-otp-security})\fi &  model (\Cref{par:impossibility}\ifacm\else, \Cref{sec:impossibility}\fi)  & length) in classical oracle model \\ 
 \hline
 Operational definitions & Impossibility for generic & For random functions \\  
  \ifacm\else(\Cref{sec:operational_defs})\fi & constructions in plain   & in classical oracle model; \\  
 & model (\Cref{par:impossibility}\ifacm\else, \Cref{sec:impossibility}\fi) & For constrained PRF, NIZK, signatures \\
 & & in plain model \\
 \hline
\end{tabular}
\end{center}
\label{fig:definitions_impossibilities_figure}
\end{table*}
\else \begin{figure}[pt] 
\centering
\begin{center}
\begin{tabular}{||c | c | c ||} 
 \hline
 Definition & Impossibilities & Construction   \\ [0.5ex] 
 \hline\hline
 Single physical query, & Strong impossibility  & For single physical   \\ 
 quantum-output, simulation-based  & in oracle model  & query learnable   \\
  \ifacm\else(\Cref{def:single-query-simulation-security})\fi & (\Cref{sec:tech_overview_definitional}) & (trivial) functions only \\
   & & \cite{broadbent2013quantum} \\
 \hline
 Single physical query, & Impossibility for generic & N/A  \\
classical-output, simulation-based  & construction in oracle model  &   \\
 \ifacm\else(\Cref{def:single-query-classical-output-simulation-security})\fi & (\Cref{par:impossibility}\ifacm\else, \Cref{sec:impossibility_oracle_model}\fi) &   \\
 \hline
 Single effective query & Impossibility for generic  & For all functions  \\
  quantum-output, & constructions in plain  & (with proper randomness   \\
simulation-based \ifacm\else(\Cref{def:simulation-style-otp-security})\fi &  model (\Cref{par:impossibility}\ifacm\else, \Cref{sec:impossibility}\fi)  & length) in classical oracle model \\ 
 \hline
 Operational definitions & Impossibility for generic & For random functions \\  
  \ifacm\else(\Cref{sec:operational_defs})\fi & constructions in plain   & in classical oracle model; \\  
 & model (\Cref{par:impossibility}\ifacm\else, \Cref{sec:impossibility}\fi) & For constrained PRF, NIZK, signatures \\
 & & in plain model \\
 \hline
\end{tabular}
\end{center}
\caption{Definitions with Impossibilities and Constructions.
The exact impossibility results and positive results for operational definitions depend on which definition of single-query model we work with. 
\ifacm See the full version for details \else See \Cref{sec:unlearnability_defs}, \Cref{sec:operational_defs}, and \Cref{sec:impossibility} for details. \fi}
\label{fig:definitions_impossibilities_figure}
 \end{figure} 
 \fi
\fi


\section{Technical Overview}
\label{sec:tech_overview}

\subsection{Definitional Work}
\label{sec:tech_overview_definitional}
\paragraph{First Attempt at Defining One-Time Sampling Programs.}
As we discussed in the introduction, we cannot achieve one-time security for deterministic classical functions without hardware assumptions, even after encoding them into quantum states: by applying the gentle measurement lemma \cite{aaronson2004limitations,winter1999coding}, any adversary can repair the program state after a measurement on the program's output that gives a deterministic outcome.

We therefore resort to considering classical \emph{randomized} computation, which we model as the following procedure: the user (adversary) can pick its own input $x$; the program samples a random string $r$ for the user and outputs the evaluation $f(x,r)$ for the user, for some deterministic function
$f$. Note that it's essential that the user does \emph{not} get to pick their own randomness $r$ -- otherwise the evaluation is deterministic again and is subject to the above attack.

For correctness, we need to guarantee that after an honest evaluation, the user gets the outcome $f(x,r)$ for its own choice of $x$ and a uniformly random $r$. For security, the hope is that when the output of $f$ looks "random enough" (e.g. $f$ is a hash function or a pseudorandom function), the adversary should not be able to do more than evaluating the program honestly once.
We discuss several candidate definitions, the corresponding impossibility results as well as our solutions that circumvent the impossibilities.

Ideally, we would establish a \emph{simulation-based} security definition. This might require the existence of a QPT algorithm $\Sim$ which can produce the adversary's real-world view given a single query to $f$:
\[
    \OTP(f) \approx \Sim^{f_1}
\]
where $f_1$ denotes that $\Sim$ may query $f$ a single time. Indeed, such a definition is formalized, and subsequently ruled out, by Broadbent, Gutoski, and Stebila~\cite{broadbent2013quantum}.

This definition can be adapted to sampling programs by considering sampling $f$ from a function family $\calF$ at the start of the experiment. Additionally, to prevent a trivial definition which can be satisfied by $\Sim$ choosing its own $f$, the distinguisher gets access to the sampled $f$:
\[
    \{f, \OTP(f)\}_{f\gets \calF} \approx \{f, \Sim^{f_1}\}_{f\gets \calF}
\]
Unfortunately, this candidate definition suffers from impossibility results of its own.

\paragraph{Impossibility Results for the "Traditional" Simulation-Based Definitions.}
\label{par:discussions_sim_security_barrier}
As often applicable to simulation-based definitions, 
the first impossibility results from the separation between the simulated oracle world and a plain model where the one-time program can be accessed in a possibly non-black-box way. Our one-time program we give to $\calA$ consists of plain-model circuits (actual code) and quantum states, instead of oracle circuits.
Our simulator is given only oracle access to $f$.
Once $\calA$ has non-black-box access to the given one-time program, $\calA$ may be able to perform various attacks that the simulator cannot do: for instance, homomorphic evaluation on the one-time program. Then one can show that there exists a family of circuits, even though "unlearnable" when only given query access, can always be fully learned (i.e. the adversary can fully recover the functionality) once given non-black-box access.
As demonstrated in \cite{alagic2020impossibility, ananth2020secure}, this type of non-black-box attacks are applicable even if the obfuscation program is a quantum state.

One may wonder if we can simply use the above impossibility result above for quantum VBB directly as an impossibility result for the one-time program in the plain model. However, there are subtleties we need to deal with: the circuit in the quantum obfuscated program in the above results (\cite{alagic2020impossibility, ananth2020secure}) is deterministic, which will give a trivial impossibility result in the one-time program setting, irrelevant to non-black-box access. Moreover, the adversary receiving the one-time program is only able to evaluate once and the program may get destroyed. 
In our case, we need a sampling program with high min-entropy outputs where one can still apply a non-black-box attack with one single evaluation. We design a slightly contrived secret-key encryption circuit that leads to the impossibility result in the plain model -- we will elaborate its details in a later paragraph \ref{par:impossibility} and formally in \ifacm the full version.\else\Cref{sec:impossibility}\fi. For now, let us proceed with the discussion on the definitions.


\paragraph{Barriers for Stateless One-Time Programs.}
\label{par:stateless_barriers}
Even more problematic, the above definition encounters impossibilities even in the \emph{oracle} model, where we ensure that the program received by $\cA$ consists of oracle-aided circuits, preventing the non-black-box attack described earlier from applying.

This limitation primarily arises from the fact that $\Sim$ is given a stateful oracle, while $\cA$ is provided with a stateless one-time program (which includes a stateless oracle).
\jiahui{explain $\cO_f$}
To illustrate, consider the following $\cA$ and distinguisher $\cD$: $\cA$ receives a possibly oracle-aided program and simply passes the program itself to $\cD$. 
Let $\cO_f$ be a stateless oracle for $f$ that outputs $y = f(x,r)$ on any input $(x,r)$ and not restricted in the number of queries that it can answer.
$\cD$ is given arbitrary oracle access to $\cO_f$ so $\cD$ can perform the following attack using gentle measurement \ifacm\else(\Cref{lem:gentle_measure})\fi and un-computation:
\begin{enumerate}
    \item Evaluate the program given by $\cA$ on a $\ket{x_1}_\inp \ket{0}_\out \ket{0}_\chk$ where the input register $\inp$ contains $x_1, r_1$, some arbitrary 
$x_1$ of $\cD$'s choice and some randomness $r_1$ sampled by the program. $\out$ is an output register and $\chk$ is an additional register in $\cD$'s memory.

 \item Get outcome $\ket{x_1}_\inp \ket{r_1, y_1 = f(x_1, r_1)}_\out \ket{0}_\chk$.

\item But $\cD$ does not proceed to measure the register $\out$. Instead it performs a gentle measurement by checking if the $y_1$ value in $\out$ is equal to the correct $y_1 = f(x_1, r_1)$, writing the outcome in register $\chk$. It can do so because it has access to $\cO_f$. Then it measures the bit in register $\chk$.

\item Since the above measurement gives outcome $1$ with probability 1, $\cD$ can uncompute the above results and make sure that the program state is undisturbed. Then, it
can evaluate the program again on some different $x_2$ of its choice.

\item But when given a simulator's output, $\Sim$ cannot produce a program that contains more information than what's given in a single quantum oracle access to $f$. Therefore, unless $\Sim$ can "learn" $f$ in a single quantum query and produce a program that performs very closely to a real program on most inputs, $\cD$ may easily detect the difference between two worlds.
\end{enumerate}

The above argument is formalized in \cite{broadbent2013quantum}, which rules out stateless one-time programs for quantum channels even in the oracle model unless the function can be learned in a single query (for example, a constant function). The above simulation-based definition discussed can be viewed as a subcase of \cite{broadbent2013quantum}'s definition. 
\jiahui{fact check this}
Only in this trivial case, $\Sim$ can fully recover the functionality of $f$ and make up a program that looks like a real world program, since both $\Sim$ and $\calA$ can learn everything about $f$.

To get around the above oracle-model attack, we first consider the following weakening: what if we limit both the adversary and simulator to output only \emph{classical information}?
Intuitively, this requires both $\cA$ and $\Sim$ to dequantize and "compress" what they can learn from the program/oracle into a piece of classical information, so that $\cA$ cannot output the entire functionality unless it has "learned" the classical description of the functionality. However, we will show in \ref{par:impossibility_oracle} that there even exists a function with high min-entropy output  such that its classical description can be "learned" given any stateless, oracle-based one-time sampling program, but is unlearnable given only a single query. Thus we will need to explore other avenues

These impossibility results appear to stem more from a definitional limitation than a fundamental obstacle. The adversary is always given a stateless program, but the oracle given to the simulator is by definition strongly stateful: it shuts down after answering any single query (we call such an oracle single \emph{physical} query oracle). Therefore, $\Sim$ is more restricted than the $\cA$ in real world.  

\paragraph{The Single-Effective-Query Model.}
To avoid the above issue, we consider a weakening on the restriction of the "single query" which $\Sim$ can make. 
In the traditional one-time security, $\Sim$ can merely make one physical query, but $\cA$ and $\cD$ can actually make many queries, as long as the measurements on those queries they make are "gentle" (for example, a query where the outcome $f(x,r)$ is unmeasured and later uncomputed) or repeated (for example, two classical queries on the same $(x,r)$). 

In the single-\emph{effective}-query model, we relax $\Sim$'s single-physical-query restriction to also allow multiple queries, as long as they are "gentle" or repeated. 
We will define a stateful oracle $f_{\SEQ}$ which tracks at all times which evaluations $f(x;r)$ the adversary has knowledge about. If $f_{\SEQ}$ receives a query to some $x'$ while it knows the adversary has knowledge about an evaluation on $x\neq x'$, it will refuse to answer.
Using this oracle, we may define single-effective query \ifacm \\ \else\fi
simulation-security in the same manner as our previous attempt by giving the simulator access to the single-effective-query oracle $f_{\SEQ}$ instead of the single-physical-query oracle $f_1$: 
\[
    \OTP(f) \approx \Sim^{f_{\SEQ}}
\]

The reader may be concerned that since $f$ is not sampled from any distribution here, this definition is subject to the previously discussed impossibility for deterministic functions. As we will see shortly, the randomization of $f$ is directly baked in to the definition of $f_{\SEQ}$.



\paragraph{Defining the Single-Effective-Query Oracle.}
\label{par:compressed_inspired_def}
To define the single-effective query oracle $f_{\SEQ}$, we use techniques from compressed random oracles, which were introduced by Zhandry to analyze security in the quantum-accessible random oracle model (QROM)~\cite{zhandry19compressed}.
 Very roughly speaking, a compressed oracle gives an efficient method of simulating quantum query access to a random oracle on the fly by lazily sampling responses in superposition which can be "forgotten" as necessary.

The first main idea in \cite{zhandry19compressed} compressed oracle technique is to take a purified view on the joint view of the adversary's query register and the oracle: evaluating a random function in the adversary's view is equivalent to evaluating on some function $H$ from a uniform superposition over all functions (of corresponding input and output length) $\sum_{H} \ket{H}_\cH$. 
When an adversary makes a query of the form  $\sum_{x,u} \alpha_{x,u}\ket{x,u}$, the oracle applies the operation
\begin{align*}
    & \sum_{x,u} \alpha_{x,u}\ket{x,u} \otimes \sum_H\ket{H} 
    \Rightarrow 
    \sum_{x,u, H} \alpha_{x,u}\ket{x,u+H(x)} \otimes  \ket{H}
\end{align*}
Zhandry's second contribution is a method to "compress" the exponentially large superposition into a small database. It will be instructive to first consider the uncompressed version, so we defer details about the compressed version to later.

To define $f_{\SEQ}$, we allow it to maintain a purified version of $H$, which it represents as a truth table. In other words, it maintains a register $\calH = (\calH_x)_{x\in \calX}$ which is initialized to 
\[
    \ket{H_\emptyset}_{\calH} \coloneqq \sum_{H:\calX \rightarrow \calR} \bigotimes_{x\in X} \ket{H(x)}_{\calH_x} 
    = \bigotimes_{x\in X} \sum_{r \in \calR} \ket{r}_{\calH_x}  
\]
When $f_{\SEQ}$ decides to answer a query $\ket{x, u}$, it computes $\ket{x, u\oplus f(x;H(x))}$ by reading register $\calH_x$ in the computational basis. The first query made results in the joint state
\ifacm
\begin{align*}
    \ket{x^*,u}_\cQ \otimes \ket{H_\emptyset} 
    \overset{U_{f_{\$}}}{\longrightarrow}
    &\sum_{r} \ket{x^*,u \oplus f(x^*, r)}_\cQ \otimes \ket{r}_{\cH_{ x^*}} 
    \\&\otimes \sum_{H:H(x^*) = r}\bigotimes_{x \in \cX, x\neq x^*} \ket{H(x)}_{\cH_x}
\end{align*}
\else
\[
    \ket{x^*,u}_\cQ \otimes \ket{H_\emptyset} 
    \overset{U_{f_{\$}}}{\longrightarrow}
    \sum_{r} \ket{x^*,u \oplus f(x^*, r)}_\cQ \otimes \ket{r}_{\cH_{ x^*}} \otimes \sum_{H:H(x^*) = r}\bigotimes_{x \in \cX, x\neq x^*} \ket{H(x)}_{\cH_x}
\]
\fi

If $f(x;r)$ were to uniquely determine $r$, then measuring $f(x^*, H(x^*))$ would fully collapse register $\calH_{x^*}$ while leaving the others untouched. Afterwards, the single-effective-query oracle $f_{\SEQ}$ could detect which input was evaluated and measured by comparing each register $\calH_x$ to the uniform superposition $\sum_{r\in \calR}\ket{r}$. It could then use this information to decide whether to answer further queries. 
On the other hand, if there were many collisions $f(x^*;r^*) = f(x^*;r^*_2)$ or the adversary erased its knowledge of $f(x^*;r^*)$ by querying on the same register again, then $\calH_{x^*}$ might not be fully collapsed. 
In this case, it is actually beneficial that $f_{\SEQ}$ does not completely consider $x^*$ to have been queried, since this represents a "gentle" query which would allow the adversary to continue evaluating a real one-time program.

When we switch to the compressed version of $H$, collapsing $\calH_{x^*}$ to $\ket{r^*}$ corresponds to recording $(x^*, r^*)$ in a database $D$. Since the adversary's queries may be in superposition, the database register $\calH$ may become entangled with the adversary. In other words, the general state of the system is $\sum_{a,D}\alpha_{a,D}\ket{a}_{\calA} \otimes\ket{D}_{\calH}$ where $\calA$ belongs to the adversary and $\calH$ belongs to $f_{\SEQ}$.
Using this view, $f_{\SEQ}$ may directly read the currently recorded query off of its database register to decide whether to answer a new query. The entanglement between the adversary's register and the database register enables $f_{\SEQ}$ to answer or reject new queries $x$ precisely when the adversary does not have another outstanding query $x'$.
As a result, the database register will always contain databases with at most one entry.

\paragraph{Functions for which SEQ Access is Meaningful.} 
The single-effective-query simulation definition captures all functions, including those that are trivially one-time programmable. Similarly to the notion of ideal or virtual black-box obfuscation, any "unlearnability" properties depend on the interaction of the function with the obfuscation definition. 
For example, deterministic functions can be fully learned given access to an SEQ oracle, since measuring evaluations will never restrict further queries.

Intuitively, a function must satisfy two loose properties in order to have any notion of unlearnability with SEQ access:
\begin{itemize}
    \item \textbf{High Randomness.} To restrict further queries, learning (via measuring) $f(x;r)$ must collapse the SEQ oracle's internal state, causing $(x,r)$ to be recorded in the purified oracle $H$.
    \item \textbf{Unforgeability.} To have any hope that $f$ has properties that cannot be learned with SEQ access, $f$ cannot be learnable given, say, a single evaluation $f(x;r)$.
\end{itemize}
As an example, truly random functions exemplify both of these properties. A truly random function has maximal randomness on every input and $f(x;r)$ is independent of $f(x';r')$. We formally explore SEQ access to truly random functions and a few other function families in 
\ifacm the full version. \else \Cref{sec:seq-meaningful}.\fi

\subsection{Positive Results}
\label{par:construction_oracle_intro}

\paragraph{Construction with Classical Oracles.}
We now give an overview on our construction using classical oracles . We show security with respect to the simulation-based definition where the simulator queries a single-effective-query oracle.

Our construction in the oracle model is inspired by the use of the "hidden subspace states" in the literatures of quantum money \cite{aaronson2012quantum}, signature tokens and quantum copy protection \cite{ben2023quantum,coladangelo2021hidden}. 
A subspace state $\ket{A}$ is a uniform superposition over all vectors in some randomly chosen, secret subspace $A \subset \F_2^\lambda$. Specifically, $\ket{A} \propto \sum_{v \in A} \ket{v}$, where dimension of $A$ is $\lambda/2$ and $\lambda$ is the security parameter. These parameters ensure that $A$ has exponentially many elements but is still exponentially small compared to the entire space. 

At a high-level, our one-time scheme requires an authorized user to query an oracle on subspace vectors of $A$ or its dual subspace $A^\perp$. Let $f$ be the function we want to one-time program. Consider the simple case where $x$ is a single bit in $\{0,1\}$. Let $G$ be a PRG or extractor (which can be modeled as a random oracle since we already work in the oracle model).
The one-time program consists of a copy of the subspace state $\ket{A}$ along with access to the following classical oracle:
\begin{align*}
    \cO(x, v) &= \begin{cases}
                            f(x, G(v)) & \text{ if } x = 0, v \in A \\
                            f(x, G(v)) & \text{ if } x = 1, v \in A^\perp \\
                            \bot & \text{ otherwise }
                    \end{cases}.
\end{align*}

To evaluate on input $x$, an honest user will measure the state $\ket{A}$ to obtain a uniform random vector in subspace $A$, if $x = 0$; or apply a binary QFT to $\ket{A}$ and measure to obtain a uniform random vector in the dual subspace $A^\perp$, if $x=1$. It then inputs $(x,v)$ into the oracle $\cO$ and will obtain the evaluation $\cO(x, G(v))$ where the randomness $G(v)$ is uniformly random after putting the subspace vector into the random oracle.

For security, we leverage an "unclonability" property of the state $\ket{A}$ (\cite{ben2023quantum,bartusek2023obfuscation}) called "direct-product hardness": an adversary, given one copy of $\ket{A}$, polynomially bounded in query to the above oracle should not be able to produce two vectors  $v, v'$ which satisfy either of the following: (1) $v \in A, v' \in A^\perp$; (2) $v, v' \in A$ or $v,v' \in A^\perp$ but $v \neq v'$.

First, we consider a simpler scenario: this evaluation is destructive to the subspace state if the user has obtained the outcome $f(x, G(v))$ and the function $f$ behaves random enough so that measuring the output $f(x, G(v))$ is (computationally) equivalent to having measured the subspace state $v$. Now, it will be hard for the user to make a second query into the oracle $\cO$ on a different input $(x', v')$ so that either $x \neq x'$ or $v \neq v'$ because it would lead to breaking the direct-product hardness property mentioned above. 

More generally, however, $f$ may be only somewhat random, or the adversary may perform superposition queries. In these cases, $\ket{A}$ will be only partially collapsed in the real program, potentially allowing further queries. This partial collapse also corresponds to a partial collapse of the single-effective-query oracle's database register, similarly restricting further queries.


To establish security, the main gap that the simulator needs to bridge is the usage of a subspace state $\ket{A}$ versus a purified random oracle to control query access. Additionally, the real world evaluates $f(x, G(v))$, where $v$ is a subspace vector corresponding to $x$, while the ideal world evaluates $f(x, H(x))$ directly.\footnote{Although $H$ and $G$ are both random oracles, we differentiate them to emphasize that they act on different domains.}
If we were to purify $G$ as a compressed oracle, then $\ket{A}$ collapsing corresponds to $G$ recording some subspace vector $v$ in its database. 
At a high level, this allows the simulator to bridge the aforementioned gap by using a careful caching routine to ensure that $\ket{A}$ collapses/$v$ is recorded in the cache if and only if $x$ is recorded in $H$. 
Using the direct product hardness property, we can be confident that at most one $v$ and corresponding $x$ are recorded in the simulator. Thus, to show that the simulator is indistinguishable from the real one-time program, we can simply swap the role of $x$ and $v$ in the oracle, changing between $G$ and $H$. We provide more details in \ifacm the full version.\else\Cref{sec:seq-construction}.\fi

\paragraph{Operational Security for Cryptographic Functionalities.}
\label{par:operational_secure_intro}
 While the classical oracle construction is clean, implementing the oracle itself with concrete code will bump into the plain model versus black-box obfuscation barrier again. 
We cannot hope to make one-time sampling programs for all functions (not even for all high min-entropy output functions) in the plain model, due to the counter-example we provide in the first paragraph of \ref{par:impossibility} we will soon come to.
Moreover, the simulation definition we achieve above captures functions all the way from those that can be meaningfully one-time programmed, like a random function, to those "meaningless" one-time programs of for example a constant function, which can be learned in a single query.  

One may wonder, what are some meaningful functionalities we can implement a one-time program with and what are some possible security notions we can realize for them in the plain model?

We consider a series of relaxed security notions we call "operational one-time security definition", which can be implied by the simulation-based definition. The intuition of these security definitions is to characterize "no QPT adversary can evaluate the program twice".

Consider a cryptographic functionality $f$, we define the security game as follows:
the QPT adversary $\cA$ receives a one-time program for $f$ and will output its own choice of two input-randomness pairs $(x_1,r_1), (x_2, r_2)$. For each $(x_i,r_i)$, $\cA$ needs to answer some challenges from the challenger with respect to the cryptographic functionality $f$. The security guarantees that $\cA$'s probability of winning both challenges for $i \in \{1,2\}$ is upper bounded by its advantage in a cryptographic security game of winning a single such challenge, but without having acess to the one-time program.

For some cryptographic functionalities, this cryptographic challenge is simply to compute $y_i = f(x_i, r_i)$. A good example is a one-time signature scheme: $\cA$ produces two messages of its own choice, but it should not be able to produce valid signatures for both of them with non-negligible probability.

More generically, $\cA$ produces two message-randomness pairs and gives them to the challenger. Then the challenger prepares some challenges independently for $i \in \{1,2\}$ and $\cA$ has to provide answers so that $\cA$'s inputs, the challenges and final answer need to satisfy a predicate.  In the above signature example, the predicate is simply verifying if the signature is a valid one for the message. Another slightly more contrived predicate is answering challenges for a pseudorandomness game of a PRF: receiving an $\OTP$ for a PRF, no QPT adversary can produce two input-randomness pairs $(x_1, r_1), (x_2, r_2)$ of its own choice such that it can win the pseudorandomness game with respect to both of these inputs. That is, the challenger will flip two independent uniform bits for each $i \in \{1,2\}$ to decide whether to let $y_i = \prf(x_i,r_i)$ or let $y_i$ be real random. The security says that $\cA$'s overall advantage should not be noticeably larger than $1/2$; $\cA$ can always evaluate once and get to answer one of the challenges with probability 1, but for the other challenge it can only make a random guess.

\paragraph{One-Time Program for PRFs in the Plain Model.}
\label{par:construction_plain_model_intro}
 However, one cannot hope to achieve a one-time program construction that is secure for all functions in the \emph{plain model}, even if we restrict ourselves to the weakest operational definition and high min-entropy output functions. As aforementioned, we give the counter-example of such a circuit (assuming some mild computational assumptions) in the next paragraph \ref{par:impossibility}.
 
We therefore turn to considering constructions for specific functionalities in the plain model and a give a secure construction for a family of PRFs, with respect to the aforentioned security guarantee: no QPT adversary can produce two input-randomness pairs $(x_1, r_1), (x_2, r_2)$ of its own choice such that it can win the pseudorandomness game with respect to both of these inputs.

To replace the oracle in the above construction, we use $\iO$ (which stands for indistinguishability obfuscation, \cite{barak2001possibility}) which guarantees that the obfuscations of two functionally-equivalent circuits should be indistinguishable.

The construction in the plain model bears similarities to the one in the oracle model.
Let $\prf_k(\cdot)$ be the PRF as our major functionality. In our actual construction, we use another $\prf$ $G$ on the subspace vector 
$v$ to extract randomness, but we omit it here for clarity of presentation. 
We give out a subspace state $\ket{A}$ and an $\iO$ of a program.

The following program is a simplification of the actual program we put into $\iO$:
\begin{align*}
    \prf_{k,A}(x, v) &= \begin{cases}
                            \prf_{k}(x, v) & \text{ if } x = 0, v \in A \\
                            \prf_k(x, v) & \text{ if } x = 1, v \in A^\perp \\
                            \bot & \text{ otherwise }
                    \end{cases}.
\end{align*}
To show security, we utilize a constrained PRF (\cite{boneh2017constrained}):  a constrained PRF key $k_C$ constrained to a circuit $C$ will allow us to evaluate on inputs $x$ that satisfy $C(x) = 1$ and output $\bot$ on the inputs that don't satisfy.  The constrained pseudorandomness security guarantees that the adversary should not be able distinguish between $(x, \prf_k(x))$ and $(x, y \gets \text{random})$, where $\cA$ can choose $x$ such that $C(x) = 0$ and  $k$ is the un-constrained key. 

In our proof, we can use hybrid arguments to show that we can  use the adversary to  violate the constrained pseudorandomness security. 
Let us denote $A^0 = A, A^1 = A^\perp$.
First we invoke the security of $\iO$: we change the above program to one using a constrained key $k_C$ to evaluate $\prf$, where $C(x,v) = 1$  if and only if $v \in A^x$. The circuit has the equivalent functionality as the original one. Next, we invoke a
computational version of subspace state direct-product hardness property: we change the one-time security game to rejecting all adversary's chosen inputs $(x_1, v_1), (x_2, v_2)$ such that both $v_1 \in A^{x_1}$ and $v_2 \in A^{x_2}$ hold. Such a rejection only happens with negligible probability due to the direct-product hardness property. Finally, the adversary must be able to produce some $(x,v)$ where $C(x,v) = 0$ and distinguish $\prf_k(x,v)$ from a random value. We therefore use it to break the constrained pseudorandomness security.

\paragraph{More Applications and Implications: Generic Signature Tokens, NIZK, Quantum Money}
We show a way to lift signature schemes that satisfy a property called blind unforgeability to possess one-time security. Unlike the \cite{ben2023quantum} signature token scheme where the verification key has to be updated once we delegate a one-time signing token to some delegatee, our signature token scheme can use an existing public verification procedure.

Apart from the above one-time PRF in the plain model, we also instantiate one-time NIZK from iO and LWE in the plain model, using a similar construction as in the one-time PRF and the NIZK from iO in \cite{sahai2014use}. The proof requires more careful handling because  the NIZK proof also has publicly-verifiable property.

Finally, we show that one-time program for publicly verifiable programs (e.g. signature, NIZK) implies public-key quantum money. Despite the destructible nature of the one-time program, we can design a public verification procedure that gently tests a program's capability of computing a function and use a one-time program token state as the banknote.

\subsection{Impossibility Results}
\label{par:impossibility}
\paragraph{Impossibility Result in the Plain Model.}
 In this paragraph, we come back to the discussions about impossibility results in the paragraph "Barriers for stateless one-time programs"(\ref{par:discussions_sim_security_barrier}).
We describe the high level idea on constructing the program used to show an impossibility result in the plain model, inspired by the approach in \cite{ananth2020secure,alagic2020impossibility}. This impossibility holds \emph{even for the weakest definition} we consider: $\cA$ is not able to produce two input-output pairs after getting one copy of the $\OTP$.

We have provided a table in \Cref{fig:definitions_impossibilities_figure} in \Cref{sec:our_results}  that demonstrates the several security definitions we discuss in this work, their corresponding impossibility results and positive results. Please refer to the table so that the relationships between the several definitions proposed in this work are clearer.


We design an encryption circuit $C$ with a random "hidden point" such that when having non-black-box access to the circuit, one can "extract" this hidden point using a quantum fully homomorphic encryption on the one-time program. However, when having only oracle access, one cannot find this hidden point within any polynomial queries.

Let $\SKE$ be a secret key encryption scheme. Let $a,b$ be two randomly chosen values in $\{0,1\}^n$.
\begin{description}

    \item Input: $(x \in \{0,1\}^n, r \in \{0,1\}^\ell)$

    \item Hardcoded:  $(a, b, \Enc.\sk)$

    \item \quad if $x = a$:
       output $\SKE.\Enc(\Enc.\sk, b; r)$

    \item \quad  else: output $\SKE.\Enc(\Enc.\sk, x; r)$. 
\end{description}
The above circuit also comes with some classical auxiliary information, given directly to $\cA$ in the real world (and $\Sim$ in the simulated world, apart from giving oracle access).
Let $\qhe$ be a quantum fully homomorphic encryption with semantic security. We also need a compute-and-compare obfuscation program (which can be built from LWE). A compute-and-compare obfuscation program is obfuscation of a circuit $\CC[f,m,y]$ that does the following: on input $x$, checks if $f(x) = y$; if so, output secret message $m$, else output $\bot$. The obfuscation security guarantees that when $y$ has a high entropy in the view of the adversary, the program is indistinguishable from a dummy program always outputing $\bot$ (a distributional generalization of a point function obfuscation).

In the auxiliary information, we give out $\ct_a = \qhe.\Enc(a)$ and the encryption key, evaluation key $\qhe.\pk$ for the QFHE scheme. We also give the compute-and-compare obfuscation for the following program:
    \begin{align*}
      \CC[f, (\SKE.\sk, \qhe.\sk), b] = \begin{cases}
                            (\SKE.\sk, \qhe.\sk) & \text{ if }    f(x) = b \\
                            \bot & \text{ otherwise }
                        \end{cases}
    \end{align*}
where $f(x) = \SKE.\Dec(\SKE.\sk, \qhe.\Dec(\qhe.\sk, x))$.
Any adversary with non-black-box access to the program can homomorphically encrypt the program and then evaluate to get a QFHE encryption of a ciphertext $\SKE.\Enc(b)$ (doubly encrypted by $\SKE$ and then $\qhe$ ). This ciphertext, once put into the above $\CC$ obfuscation program, will give out all the secrets we need to recover the functionality of the circuit $C$.

However, when only given oracle access, we can invoke the obfuscation security of the compute-and-compare program and the semantic security of QFHE, finally removing the information of $b$ completely from the oracle, rendering the functionality as a regular SKE encryption sheme (which behaves like a random oracle in the classical oracle model). The actual proof is more intricate, involving a combination of hybrid arguments, quantum query lower bounds and induction, since the secret information on $a,b$ is scattered around in the auxiliary input. We direct readers to \ifacm the full version \else \Cref{sec:impossibility} \fi for details. 

\paragraph{Impossibility Result in the Oracle Model.}
\label{par:impossibility_oracle}
In this section, we show a circuit family with high-entropy output which cannot be one-time programmed in the oracle model, with respect to the classical-output simulator definition discussed in \Cref{sec:tech_overview_definitional}. It can be one-time programmed with respect to our single-effective-query simulator definition, but only in a "meaningless" sense, since both the simulator and the real-world adversary can fully learn the functionality.

In short, it demonstrates several separations: (1) It separates a single-physical-query unlearnable function from single-effective-query unlearnable: one can fully recover the functionality once having a single effective query oracle, but one cannot output two input-output pairs when having only one physical query.
(2) It is a single-physical-query unlearnable function that cannot be securely one-time programmed with respect to the operational definition where we require the adversary to output two different correct evaluations. (3) Having high min-entropy output distributions is not sufficient to prevent the adversary from evaluating twice (i.e. "meaningfully" one-time programmed). We would also like to make a note that this result does not contradict our result on the single-effective-query unlearnable families of functions \ifacm mentioned in the full version \else in \Cref{sec:exampls_seq_unlearnable}\fi because it is not truly random or pairwise-independent.




Now consider the following circuit. An adversary receives an oracle-based one-time program and a simulator gets only one (physical) query to the functionality's oracle. Both are required to output a piece of classical information to a distinguisher.

Let $a$ be a uniformly random string in $\{0,1\}^n$. Let $k$ be a random PRF key. The following $\prf_k(\cdot)$ maps $\{0,1\}^{2n} \to \{0,1\}^{2n}$. Let our circuit be the following:

\begin{align}
\label{eqn:partially-deterministic_intro}
    f_{a, k}(x; r) = 
    \begin{dcases}
        (a, \mathsf{PRF}_k(0 \| r))&\text{if } x=0,\\
        (k, \mathsf{PRF}_k(a \| r))&\text{if } x=a,\\
        (0, \mathsf{PRF}_k(x \| r))&\text{otherwise.}
    \end{dcases}
\end{align}
When $\cA$ is given an actual program, even using an oracle-aided circuit, $\cA$ can do the following: evaluate the program on input $0$ (and some randomness $r$ it cannot control) to get output $(a, \prf_k(0\Vert r))$; but instead of measuring the entire output, only measure the first $n$-bits to get $a$ with probability 1; evaluate the program again on $a$ and some $r'$ to obtain $(k, \prf_k(a \vert r'))$. Now it can reconstruct the classical description of the entire circuit using $a$ and $k$.

But when given only a single \emph{physical} query, we can remove the information of $k$ using \cite{BBBV97} argument since for a random $k$, no adversary should have non-negligible query weight on $k$ with just one single query. 
\ifexabs
\else
\subsection{Concurrent Work and Related Works}
\paragraph{Concurrent Work.}
A concurrent and independent work \cite{gunn2024quantum} presents a construction of quantum one-time programs for randomized classical functions. While our main constructions are very similar, there are some differences between our works which we outline next. (1) We undertake a more comprehensive study of security definitions for quantum one-time programs and come up with both simulation definitions in the oracle model as well as operational definitions in the oracle and plain model. \cite{gunn2024quantum} focuses on the oracle model.  
(2) We show simulation-based security for our construction in the oracle model, which is stronger than the security definition used in \cite{gunn2024quantum}. 
(3) We show an impossibility result for generic one-time randomized programs in the plain model; (4) We give constructions for PRFs and NIZKs in the plain model whereas all constructions in \cite{gunn2024quantum} are in the oracle model; (5) We also show a generic way to lift a plain signature scheme satisfying a security notion called blind unforgeability to one-time signature tokens. (5) On the other hand, the oracle construction in \cite{gunn2024quantum} is more generic by using any signature token state as the quantum part of the one-time program, whereas we use the subspace state (namely, the signature token state in \cite{ben2023quantum}). 

\paragraph{One-Time Programs.} One-time programs were first proposed by Goldwasser, Kalai, and Rothblum~\cite{C:GolKalRot08} and further studied in a number of followup works~\cite{DBLP:conf/tcc/GoyalISVW10,TCC:GoyGoy17,EC:ACEGMPRT22}. Although these are impossible in the plain model, a number of alternative models have been proposed to enable them, ranging from hardware assumptions to protein sequencing. Broadbent, Gutoski, and Stebila~\cite{broadbent2013quantum} asked the question of whether quantum mechanics can act as a stand-in for hardware assumptions. However, they found that quantum one-time programs are only possible for ``trivial'' functions, which can be learned in a single query, and are generally impossible for deterministic classical functions. To get around this impossibility, they use the much stronger one-time memory model, which allows them to achieve one-time programs for general quantum channels.
A pair of later works circumvented the impossibility by allowing the program to output an incorrect answer with some probability~\cite{NATURE:RKBFW2018advantage,NATURE:RKFW21probabilistic}. Although their results are quite interesting, they do not give formal security definitions for their scheme, and seem to assume a weaker adversarial model where the evaluator must make many intermediate physical measurements in an online manner. In contrast, we present a formal treatment with an adversary who may perform arbitrary quantum computations on the one time program as a whole.

\cite{chung2019cryptography}
develops a first quantum one-time program for classical message-authentication codes, assuming stateless classical hardware tokens. 

Besides, \cite{liu2020quantum} studied security of classical one-time memory under quantum superposition attacks.  \cite{liu2023depth} builds quantum one-time memory with quantum random oracle 
in the depth-bounded adversary model, where the honest party only needs a short-term quantum memory but the adversary, which attempts to maintain a quantum memory for a longer term cannot perform attacks due to bounded quantum depth.

\paragraph{Signature Tokens.} Signature tokens are a special case of one-time programs that allow the evaluator to sign a single message, and no more. They were proposed by Ben-David and Sattath~\cite{ben2023quantum} in the oracle model and subsequently generalized to the plain model using indistinguishability obfuscation~\cite{C:CLLZ21}. Both of these works consider a very specific form of one-time security: an adversarial evaluator should not be able to output two (whole) valid signatures. 
\fi

\section{Preliminaries}

\subsection{Quantum Information and Computation}

We provide some basics frequently used in this work and refer to \cite{nc02} for comprehensive details.

A projection is a linear operator $P$ on a Hilbert space that satisfies the property: $P^2 = P$ and is Hermitian $P^\dagger = P$.

A projective-valued measurement (PVM)  a generalization of this idea to an entire set of measurement outcomes. A PVM is a collection of projection operators $\{P_i\}$ that are associated with the outcomes of a measurement, and they satisfy two important properties: (1) orthogonality: $P_iP_j = 0, \forall i \neq j$; (2) completeness: $\sum_i P_i = \mathbf{I}$ where $\mathbf{I}$ is the identity operator.

\begin{definition}[Trace distance]\label{def:tracenorm}
Let $\rho,\sigma\in \C^{2^n\times 2^n}$ be the density matrices of two quantum states. The trace distance between $\rho$ and $\sigma$ is
\begin{align*}
    \|\rho-\sigma\|_{\mathrm{tr}}:=\frac{1}{2}\sqrt{\Tr[(\rho-\sigma)^\dagger (\rho-\sigma)]},
\end{align*}
\end{definition}

Here, we only state a key lemma for our construction: the Gentle Measurement Lemma proposed by Aaronson \cite{aaronson2004limitations}, which gives a way to perform measurements without destroying the state.

\begin{lemma}[Gentle Measurement Lemma \cite{aaronson2004limitations}] 
\label{lem:gentle_measure}
Suppose a measurement on a mixed state $\rho$ yields a
particular outcome with probability 
$1-\epsilon$.  Then after
the measurement, one can recover a state $\tilde{\rho}$ such that $ \left\lVert \tilde{\rho} - \rho \right\rVert_{\mathrm{tr}} \leq \sqrt{\epsilon}$.
Here $\lVert\cdot\rVert_{\mathrm{tr}}$ is the trace distance (defined in Definition~\ref{def:tracenorm}).
\end{lemma}

\subsection{Quantum Query Model}
\label{sec:prelim_quantum_query_classical_oracle}
We consider the quantum query model in this work, which gives quantum circuits access to some oracles. 
\begin{definition}[Classical Oracle]
\label{def:classic_oracle}
A classical oracle $\mathcal{O}$ is a unitary transformation of the form $U_f \ket{x,y, z} \rightarrow \ket{x, y+f(x), z}$ for classical function $f: \{0,1\}^n \rightarrow \{0,1\}^m$. Note that a classical oracle can be queried in quantum superposition.
\end{definition}
In the rest of the paper, unless specified otherwise, we refer to the word ``oracle'' as ``classical oracle''. 
A quantum oracle algorithm with oracle access to $\mathcal{O}$ is a sequence of local unitaries $U_i$ and oracle queries $U_f$. Thus, the query complexity of a quantum oracle algorithm is defined as the number of oracle calls to $\mathcal{O}$.

\subsection{Compressed Random Oracles}
\label{sec:compressedRO}
Zhandry~\cite{zhandry19compressed} gives an efficient method of simulating a random oracle on the fly. The method maintains a database register $\calD$ which enables it to answer queries. At the start of time, the database is initialized to $\ket{\emptyset}$, representing the even superposition over all possible functions $H:\calX \rightarrow \calY$. As the random oracle is queried, the database is updated. A database $D$ takes the form of a set containing pairs $(x,y)$. The state $\ket{D}$ represents the even superposition over all functions $H:\calX \rightarrow \calY$ which are consistent with $H(x) = y$ for all $(x, y)\in D$. At any given time, the database register $\calD$ is in superposition over one or more databases $\ket{D}$. We write $D(x) = y$ if $(x,y)\in D$. If no such entry exists, then we write $D(x) = \bot$.

In detail, the compressed oracle is defined using the following procedures:
\begin{itemize}
    \item $\Decomp$ is defined by $\ket{x, u} \otimes \ket{D} \mapsto \ket{x, u} \otimes \Decomp_x \ket{D}$, where $\Decomp_x$ is defined as follows for any $D$ such that $D(x) = \bot$:
    \begin{align}
        \Decomp_x \ket{D} &\coloneqq \frac{1}{\sqrt{|\calY|}}\sum_{y\in \calY} \ket{D\cup \{(x,y)\}}
        \\
        \Decomp_x \left(\frac{1}{\sqrt{|\calY|}}\sum_{y\in \calY} \ket{D\cup \{(x,y)\}} \right) &\coloneqq \ket{D}
        \\
        \Decomp_x \left(\frac{1}{\sqrt{|\calY|}}\sum_{y\in \calY} (-1)^{y\cdot u} \ket{D\cup \{(x,y)\}} \right) &\coloneqq \left(\frac{1}{\sqrt{|\calY|}}\sum_{y\in \calY} (-1)^{y\cdot u} \ket{D\cup \{(x,y)\}} \right) \text{ for } u\neq 0
    \end{align}
    \item $\CO'$ maps $\ket{x, u, D} \mapsto \ket{x, u\oplus D(x), D}$
\end{itemize}
On query in register $\calQ$, the compressed oracle applies 
\[
    \CO = \Decomp \circ \CO' \circ\Decomp 
\]
to registers $(\calQ, \calD)$.

In more descriptive words, if $D$ already is specified on $x$, and moreover if the corresponding $y$ registers are in a state orthogonal to the uniform superposition (i.e. if we apply a QFT to the register, the resulting state contains no $\ket{0}$), then there is no need to decompress and
$\Decomp$
is the identity. On the other hand, if $D$ is specified at $x$ and the
corresponding $y$ registers are in the state of the uniform superposition, $\Decomp$ will remove $x$ and
the $y$ register superposition from $D$.

We mention a few other results about compressed oracles. First, any adversary who successfully finds valid input/output pairs from the random oracle must cause the corresponding input/output pairs to be recorded in the compressed version of the oracle.

\begin{lemma}[\cite{zhandry19compressed}, Lemma 5]
   \label{lem:zhandry_lemma5} 
   Consider a quantum algorithm $\calA$ making queries to a random oracle
$H$ and outputting tuples $(x_1, \cdots , x_k, y_1, \cdots , y_k)$
. Let $R$ be a collection of such
tuples. Suppose with probability $p$, $\calA$ outputs a tuple such that (1) the tuple is in
$R$ and (2) $H(x_i) = y_i$ for all $i$. Now consider running $\calA$ with the compressed oracle defined in \Cref{sec:compressedRO},
and suppose the database $D$ is measured after $\calA$ produces its output. Let $p'$
be the probability that (1) the tuple is in $R$, and (2) $D(x_i) = y_i$ for all $i$ (and in
particular $D(x_i) \neq \bot$). Then:
   $$ \sqrt{p} \leq \sqrt{p'} + \sqrt{k/\vert \calY \vert}$$
\end{lemma}

Second, the probability of the compressed oracle's database containing a collision after $q$ queries is bounded in terms of $q$ and the size of the range.

\begin{lemma}[\cite{zhandry19compressed}]\label{lem:compressed-collision}
    For any adversary with query access to a compressed oracle $G:\calX \rightarrow \calY$, if the database register is measured after $q$ queries to obtain $D$, then $D$ contains a collision with probability at most $O(q^3/|\calY|)$.

    As a corollary of this and \Cref{lem:zhandry_lemma5}, any adversary making $q$ queries to a random oracle finds a collision with probability at most $O(q^3/|\calY|)$.
\end{lemma}

We mention a few facts about the compressed oracle that are not explicitly mentioned in the original work. First, if some $x$ has never been queried to the oracle, i.e. has $0$ amplitude in all prior queries, then $D(x) = \bot$. This is evident from the definition of $\Decomp$. Second, the compressed oracle acts identically on all inputs, up to their names. In other words, querying $x_1$ is the same as renaming $x_1$ to $x_2$ in both the query register and the database register, then querying $x_2$, then undoing the renaming.

\begin{claim}\label{claim:compressed-oracle-renaming}
    Let $\mathsf{SWITCH}_{x_1,x_2}$ be the unitary which maps any database $D$ to the unique $D'$ defined by $D'(x_1) = D(x_2)$, $D'(x_2) = D(x_1)$, and $D'(x_3) = D(x_3)$ for all $x_3\notin \{x_1, x_2\}$ and acts as the identity on all orthogonal states. Let $U_{x_1,x_2}$ be the unitary which maps $\ket{x_1}\mapsto \ket{x_2}$, $\ket{x_2}\mapsto \ket{x_1}$ and acts as the idenity on all orthogonal states. Then
    \[
        \CO  
        = 
        (U_{x_1,x_2} \otimes I \otimes \mathsf{SWITCH}_{x_1,x_2}) \CO (U_{x_1,x_2} \otimes I \otimes \mathsf{SWITCH}_{x_1,x_2})
    \]
\end{claim}
\begin{proof}
    This follows from the observation that $(U_{x_1,x_2} \otimes I \otimes \mathsf{SWITCH}_{x_1,x_2})$ commutes with both $\Decomp$ and $\CO'$ and it is self-inverse.
\end{proof}

\subsection{Subspace States and Direct Product Hardness}
\label{sec:subspace_state_prelims}

In this subsection, we provide the basic definitions and properties of subspace states. 

For any subspace $A$, its complement is $A^\perp = \{ b \in \mathbb{F}^n \,|\,  \langle a, b\rangle \bmod 2 = 0 \,,\, \forall a \in A \}$. It satisfies $\dim(A) + \dim(A^\perp) = n$. We also let $|A| = 2^{\dim(A)}$ denote the size of the subspace $A$.
\begin{definition}[Subspace States]
For any subspace $A \subseteq \mathbb{F}_2^n$, the subspace state $\ket A$ is defined as $$ \ket{A} = \frac{1}{\sqrt{|A|}}\sum_{a \in A} \ket a \,.$$
\end{definition}
Note that given $A$, the subspace state $\ket A$ can be constructed efficiently.

\begin{definition}[Subspace Membership Oracles]
    A subspace membership oracle $\cO_A$ for subspace $A$ is a classical oracle that outputs 1 on input vector $v$ if and only if $v \in A$.
\end{definition}

\paragraph{Query Lower Bounds for Direct-Product Hardness}
\label{sec:prelim_IT_direct_product}

We now present a query lower bound result for "cloning" the subspace states. These are useful for our security proof in the classical oracle model. 

The theorem states that when given one copy of a random subspace state, it requires exponentially many queries to the membership oracles to produce two vectors either: one is in the primal subspace, the other in the dual subspace; or both are in the same subspace but are different.

\begin{theorem}[Direct-Product Hardness, \cite{ben2023quantum,bartusek2023obfuscation}]
\label{thm: direct product oracle}
Let $A \subseteq \mathbb{F}_2^n$ be a uniformly random subspace of dimension $n/2$. 
Let $\epsilon > 0$ be such that $1/\epsilon = o(2^{n/2})$. Given one copy of $\ket{A}$, and quantum access to membership oracles for $A$ and $A^{\perp}$, an adversary needs $\Omega(\sqrt{\epsilon} 2^{n/2})$ queries to output with probability at least $\epsilon$ \emph{either} of the following: (1) a pair $(v,w)$ such that $v \in A$ and $w \in A^{\perp}$; (2) $v, w$ are both in $A$ or $A^\perp$, $v \neq w$.
\end{theorem}

\subsection{Tokenized Signature Definitions}\label{sec:token-sig-defs}

A tokenized signature scheme is a signature scheme $(\Gen, \Sign, \Verify)$ equipped with two additional algorithms $\mathsf{GenTok}$ and $\mathsf{TokSign}$. $\mathsf{GenTok}$ takes in a signing key $\sk$ and outputs a quantum token $\ket{T}$. We overload it to also take in an integer $n$, in which case it outputs $n$ signing tokens. $\mathsf{TokSign}$ takes in a quantum token $\ket{T}$ and a message $m$, then outputs a signature $\sigma$ on $m$.

A tokenized signature scheme must satisfy correctness and tokenized unforgeability. Correctness requires that a signature token can be used to generate a valid signature on any $m$:
\[
    \Pr\left[\Verify(\vk, m, \sigma) = \mathsf{Reject} : 
    \begin{array}{c}
         (\sk, \vk) \gets \Gen(1^\secpar) 
         \\
         \ket{T} \gets \mathsf{GenTok}(\sk)
         \\
         \sigma \gets \mathsf{TokSign}(m, \ket{T})
    \end{array}
    \right]
    = \negl  
\]

\begin{definition}\label{def:token-unforge}
    A tokenized signature scheme $(\Gen, \Sign, \Verify, \mathsf{GenTok}, \mathsf{TokSign})$ has tokenized unforgeability if for every QPT adversary $\adv$, 
    \[
        \Pr\left[\Verify(\vk, m_i, \sigma_i) = \mathsf{Accept}\ \forall i\in [n+1]
        :
        \begin{array}{c}
             (\sk, \vk) \gets \Gen(1^\secpar) 
             \\
             \bigotimes_{i=1}^{n} \ket{T_i} \gets \mathsf{GenTok}(\sk)
             \\
             ((m_1, \sigma_1), \dots, (m_{n+1}, \sigma_{n+1}) \gets \adv(\vk, \bigotimes_{i=1}^{n} \ket{T_i})
        \end{array}
        \right]
        = \negl
    \]
\end{definition}

\section{Definitions of One-Time Sampling Programs}
\label{sec:definitions}
A one-time sampling program compiler $\OTP$ compiles a randomized function $f$ into a quantum state $\OTP(f)$ that can then be evaluated on any input $x$ to learn $f(x, R)$ for a uniformly random string $R$. The syntax and correctness of $\OTP$ are defined below. Then in the rest of the section, we will then present various notions of security and discuss their feasibility as well as their relative strengths. Also, we will sometimes refer to the one-time sampling programs as ``one-time programs'' for short.\\

Let $\calF$ be a family of randomized classical functions, such that any $f \in \calF$ takes as input $x \in \calX$, then samples $R \getsr \calR$, and outputs $f(x, R)$. For simplicity, let $\cX = \bit^m$ for some parameter $m \in \bbN$, and let $1^m$ be implicit in the description of $f$. 

\begin{definition}[Syntax of $\OTP$]
    Let $\calF$ be a family of randomized classical functions. For any given $f \in \cF$, define $\cX, \cR, \cY$ to be the sets for which $f: \cX \times \cR \to \cY$, and let $m$ be the bit-length of inputs to $f$; i.e. let $\cX = \bit^m$.
    
    Next, $\OTP$ is the following set of quantum polynomial-time algorithms:
    \begin{description}
        \item $\generate(1^\lambda, f)$: Takes as input the security parameter $1^\lambda$ and a description of the function $f \in \calF$, and outputs a quantum state $\OTP(f)$. We assume that $1^m$ is implicit in $\OTP(f)$.
        \item $\Eval(\OTP(f), x)$: Takes as input a quantum state $\OTP(f)$ and classical input $x \in \{0,1\}^m$, and outputs a classical value $y \in \cY$. 
    \end{description}
\end{definition}

\begin{definition}[Correctness]
\label{def:otp_correctness}
    $\OTP$ satisfies correctness for a given function family $\cF$ if $\generate$ and $\mathsf{Eval}$ run in polynomial time, and for every $f \in \calF$ and $x \in \{0,1\}^m$, there exists a negligible function $\negl(\cdot)$, such that for all $\lambda \in \N$, the distributions
    \begin{align*}
        \{\mathsf{Eval}(\generate(1^\lambda, f), x)\} \quad \text{ and } \quad \{f(x, r)\}
    \end{align*}
    are $\negl(\lambda)$-close in statistical distance, where the randomness in the second distribution $\{f(x, r)\}$ is over the choice of $r \getsr \cR$.
\end{definition}

\begin{remark}
    In this work, we consider sampling uniform randomness as the input. Such a distribution suffices for many applications we will discuss (one time signatures, one-time encryptions). 
\end{remark}

\subsection{Simulation-based Security Definitions for One-Time Sampling Programs}
Ideally we want to achieve a \textit{simulation-based} security definition where we define the desired one-time program functionality and insist that the real scheme be indistinguishable from this idealized functionality. 
A natural simulation-based security definition along these is as follows. In the real world, the challenger samples $f \leftarrow \calF$, and gives the adversary the output $\OTP(f)$ of the one-time program obfuscator. On the other hand, in the ideal world, the challenger samples the function $f \leftarrow \calF$, and the simulator is allowed to make one (quantum) query to an oracle $O_{f(\cdot, \$)}^{(1)}$ that implements the randomized function $f$: on query input $x$, the oracle samples $r \leftarrow \calR$ and outputs $f(x, r)$. This single-query oracle is modeled as a \textit{stateful} oracle that shuts down after being queried once (regardless of the query being quantum/classical).

In the end, the distinguisher is given the output of the adversary (in the real world), or that of the simulator (in the ideal world) and is tasked with telling apart these two worlds. In both worlds, the distinguisher also receives an oracle $O_f$ that implements the deterministic functionality of $f$: given query $(x, r)$, it outputs $f(x, r)$. 

\begin{definition}[Single physical query simulation-based one-time security]\label{def:single-query-simulation-security}
    Fix input length $n$ and security parameter $\lambda$. For all (non-uniform) quantum polynomial-time adversaries $\calA$, there is a (non-uniform) quantum polynomial-time simulator $\Sim$ that is given single (quantum) query access to $f$ and classical auxiliary information $\aux_f$ such that for any QPT 
    distinguisher $\calD$, and for any $f \in \calF$, there exists a negligible $\negl(\cdot)$ such that for all $\lambda \in \N$,
    \begin{align*}
        \left|\Pr[1 \leftarrow \calD^{O_f}(\calA(1^\lambda, \OTP(f), \aux_f),\aux_f)] - \Pr[1 \leftarrow \calD^{O_f}(\Sim^{O^{(1)}_{f(\cdot, \$)}}(\aux_f),\aux_f)] \right| \le \negl(\lambda).
    \end{align*}
where both $\cA$ and $\Sim$ are allowed to output oracle-aided circuits using oracles given to them (respectively).   
$\cO_f$ is a stateless oracle for $f$ that outputs $y = f(x,r)$ ion any input $(x,r)$ and not restricted in the number of queries that it can answer.
\end{definition}

\begin{remark}[Auxiliary Input]
\label{remark:auxiliary input}
    The auxiliary information $\aux_f$ is a piece of classical, public information sampled together when sampling/choosing $f$. Therefore, all parties, $\cA, 
\cD, \Sim$ get to see $\aux_f$.

    For instance, when $f$ is a signing function or a publicly-verifiable proof algorithm, this $\aux_f$ can be the public verification key.
\end{remark}
\jiahui{Add a comment on what this $\aux_f$ can look like}
\jiahui{ Impossibility works even if relaxed:
$\cO_f$ is a verification oracle for $f$ that outputs 1 on input $(x,r,y)$ if and only if $f(x,r) = y$}

The above definition is a worst-case definition for all $f$, but we will see in the following discussions that even if we relax the definition to an average case $f$ sampled from the function family, a strong impossibility result still holds.

\paragraph{Impossibilities for \Cref{def:single-query-simulation-security}}
As often applicable to simulation-based definitions, such as Virtual-Black-Box Obfuscation (\cite{barak2001possibility}) 
the above definition suffers from several impossibilities.

The first impossibility results from the separation between the simulated oracle world and a plain model where the one-time program can be accessed in a possibly non-black-box way.
While our simulator is given only oracle access to $f$, 
the program we give to $\calA$ consists of actual code (and quantum states) .
Once $\calA$ has non-black-box access to the given one-time program, $\calA$ may be able to perform various attacks where the simulator cannot do: for instance, evaluating homomorphically on the one-time program. As demonstrated in \cite{alagic2020impossibility, ananth2020secure}, this type of non-black-box attacks are applicable even if the program is a quantum state.

Formalizing the actual non-black-box attack takes some effort due to the randomized evaluation on $f$ in our setting.
Nevertheless, we show that even for one-time programs with a sampling circuit, there exists circuits that can never be securely one-time programmed when given non-black-box access in \Cref{sec:impossibility}.

\paragraph{Barriers for Stateless One-Time Programs}
Even worse, the above definition suffers from impossibilities even in the \emph{oracle} model, where we make sure that the program $\cA$ receives also consists of oracle-aided circuits, so that the above non-black-box attack does not apply.  

This barrier mainly results from the fact that we give $\Sim$ a stateful oracle but $\cA$ a stateless one-time program (which contains a stateless oracle).

Consider the following $\cA$ and distinguisher $\cD$: $\cA$ receives a possibly oracle-aided program and simply outputs the program itself to $\cD$. $\cD$ is given arbitrary oracle access to $\cO_f$ so $\cD$ can perform the following attack using gentle measurement (\Cref{lem:gentle_measure}) and un-computation, discussed in \Cref{sec:tech_overview_definitional}, paragraph "Overcoming barriers for stateless one-time programs". 

Conclusively, unless the function $f$ itself is "learnable" through a single oracle query (for example, a constant function), one cannot achieve the above definition. Only in this trivial case, the simulator can fully recover the functionality of $f$ and make up a program that looks like a real world program, since both $\Sim$ and $\calA$ can learn everything about $f$.
This argument is formalized in \cite{broadbent2013quantum}, which rules out stateless one-time programs even in the oracle model. Note that this impossibility holds even if we consider a randomized $f$ sampled from the function family and or give a verification oracle that verifies whether a computation regarding $f$ is correct, instead of full access to $f$. 

To get around the above oracle-model attack, we first consider the following weakening: what if we limit both the adversary and simulator to output only \emph{classical information}?
Intuitively, this requires both $\cA$ and $\Sim$ to dequantize and "compress" what they can learn from the program/oracle into a piece of classical information.

\begin{definition}[Single physical query classical-output simulation-based one-time security]\label{def:single-query-classical-output-simulation-security}
Let $\lambda$ be the security parameter.
For all (non-uniform) quantum polynomial-time adversaries $\calA'$ with \textit{classical} output, there is a (non-uniform) quantum polynomial-time simulator $\Sim$ that is given single (quantum) query access to $f$ such that for any QPT
distinguisher $\calD$, for any $f \in \calF_\lambda$, there exists a negligible $\negl(\cdot)$ such that:
    \begin{align*}
        \left|\Pr[1 \leftarrow \calD^{O_f}(\calA'(1^\lambda, \OTP(f), \aux_f),\aux_f)] - \Pr[1 \leftarrow \calD^{O_f}(\Sim^{O^{(1)}_{f(\cdot, \$)}}(\aux_f),\aux_f)] \right| \le \negl(\lambda).
    \end{align*}
     where both $\cA$ and $\Sim$ are allowed to output oracle-aided circuits, but only classical information.

$\cO_f$ is a stateless oracle for $f$ that outputs $y = f(x,r)$ on any input $(x,r)$ and not restricted in the number of queries that it can answer.
\end{definition}

However, we will demonstrate in \Cref{sec:impossibility} that even for $f$ with high min-entropy output, there exists a family of functions such that: no single physical query simulator can 
"learn" much about 
$f$ when only given oracle access to stateful oracle $\cO^{(1)}_{f(\cdot), \$}$, but there exists an efficient $\calA$  that can fully recover the functionality of $f$ when given a stateless oracle-aided one-time sampling program. Therefore, $\cA$ can output a classical output that separates itself from $\Sim$.

\paragraph{Re-defining the stateful Single-Query Oracle} These barriers inspire us to look into the definition of "single query oracle" we give to $\Sim$ and consider a weakening on the restriction of the "single-query" we allow $\Sim$ to make: $\Sim$ can merely make one physical query, but $\cA$ and $\cD$ can actually make many queries, as long as the measurements on those queries they make are "gentle".

Is it possible to allow $\Sim$ to make "gentle" queries to the oracle for $f$, just as the adversary and distinguisher can do with their stateless oracles? We hope that \emph{for certain functionalities} with sampled random inputs, we will be able to detect whether $\Sim$ makes a destructive (but meaningful) measurement on its query or a gentle (but likely uninformative) query, so that the stateful oracle can turn off once it has made a "destructive" query, but stay on when it has only made gentle queries. 

The second weaker security definition we propose examines the above "gentle measurement attack" by the distinguisher more closely, and redefines what a single query means in the quantum query model. We develop this notion in the following subsection.


\subsection{Single effective query (SEQ) simulation-based one-time security}
\label{sec:seq_oracle_definition}

In this section, we describe an alternative model of access which we call the single-effective-query (SEQ) model. This model intuitively allows the adversary to make multiple queries, but it can only have information about one evaluation at a time. 

For example, the adversary can make a query, then uncompute it, and then make a query on a different input. After any query, the adversary only has information about at most one point of the function, so they never make more than a single effective query. We will use the compressed oracle technique to record queries the adversary makes and enforce that they cannot make more than a single effective query.\newline

Let $f: \cX \times \cR \to \cY$ be a randomized function, where the user specifies $x$ and $r$ is chosen randomly\footnote{In general, this function may or may not output its randomness $r$.}. Let $f_{\$}$ be an implementation of $f$ that uses a random oracle $H$ to compute the randomness. On receiving input $x$, $f_{\$}$ computes $r = H(x)$ and outputs 
\[f_{\$}(x) = f(x,H(x))\]

The single effective query (SEQ) oracle $f_{\$, 1}$ augments $f_{\$}$ with the ability to recognize how many distinct (effective) queries the user has made to $H$ by implementing $H$ as a compressed oracle that records the user's queries. If answering a given query would increase the number of recorded queries to $2$ or more, then $f_{\$,1}$ does not answer it. Also, when $f_{\$,1}$ does answer a query, it will flip a bit $b$ to indicate that it has done so. The formal definition $f_{\$, 1}$ is given below.

\paragraph{Compressed Single-Effective-Query Oracle $f_{\$, 1}$}


We define $f_{\$, 1}$ to implement the random oracle $H$ using the compressed oracle technique described in \Cref{sec:compressedRO}. In this case, the oracle's $\calH$ register contains a superposition over databases $D$, which are sets of input/output pairs $(x, r)$. When $f_{\$, 1}$ is initialized, $\calH$ is initialized to $\ket{\emptyset}$.

$f_{\$,1}$ responds to queries by an isometry implementing the following procedure coherently on basis states $\ket{x, u, b}_{\calQ}\otimes \ket{D}_{\calH}$:
\begin{enumerate}
    \item If $(x', r) \in D$ for some $x'\neq x$ and some $r\in \calR$, skip the following steps.
    \item Otherwise, prepare $\ket{x, 0}$ in a register $\calQ'$ and query it to the compressed random oracle $H$.
    \item Apply the isometry 
    \[
        \ket{x, u, b}_{\calQ} \otimes \ket{x, r}_{\calQ'} 
        \mapsto 
        \ket{x, u\oplus f(x;r), b\oplus 1}_{\calQ} \otimes \ket{x, r}_{\calQ'}
    \]
    to registers $(\calQ, \calQ')$
    \item Uncompute step 2 by querying $\calQ'$ to the compressed oracle $H$.
\end{enumerate}


It is not hard to see that $\calH$ records at most one query at a time.

\begin{claim}
\label{claim:single_entry_database}
    Let $A$ be an oracle algorithm interacting with $f_{\$,1}$. After every query in this interaction, $\calH$ is supported on databases $D$ with at most one entry.
\end{claim}
\begin{proof}
    This is clearly true after $0$ queries. We proceed by induction. 
    Consider the projector onto the space spanned by states of the form
    \begin{align*}
        \ket{x, u}_{\calQ} &\otimes \ket{\emptyset}
        \\
        \ket{x, u}_{\calQ} &\otimes \ket{\{(x,r)\}}
    \end{align*}
    for some $x\in \calX$, $u\in\calY$, and $r \in \calR$. $f_{\$,1}$ acts as the identity on all states outside of this space. Furthermore, this space is invariant under queries to $H$. This follows from the fact that the compressed oracle operation $\Decomp$ maps $\ket{x, u}\otimes \ket{D} \mapsto \ket{x, u}\otimes \ket{D'}$ where $D$ and $D'$ are the same, except for the possibility that $D(x) \neq D'(x)$, and the fact that the compressed oracle operation $\CO'$ does not modify $\ket{D}$. Finally, $f_{\$, 1}$ only operates on $\calH$ by querying $H$ on states in this space, so it is also invariant on the space.
\end{proof}

\begin{definition}[Single-Effective-Query Simulation-based One-Time Program Security]\label{def:simulation-style-otp-security}
Let $\lambda$ be the security parameter.
Let $f$ be a function in a function family $\cF$ that possibly comes with a public, classical auxiliary input $\aux_f$. Let $D$ be a distinguisher that is bounded to make a polynomial number of oracle queries but is otherwise computationally unbounded. $D$ is given a one-time program (real or simulated), a classical plaintext description of $f$, as well as $\aux_f$. 

    $\OTP(f)$ is a \textbf{secure OTP} for $\cF$ if there exists a QPT simulator $\Sim$ such that for every $f\in \calF$ and all distinguishers $D$, there exists a  negligible function $\negl(\cdot)$ such that for all $\lambda\in \N$: 
    \begin{equation}
        \left| \Pr\left[1\gets D(\left(\OTP(f), f, \aux_f\right)\right] - \Pr\left[1\gets D\left(\Sim^{f_{\$,1}}(\aux_f), f, \aux_f \right) \right]\right| \leq \negl(\lambda)
    \end{equation}

\end{definition}
Note that this is in the oracle model, so both $\OTP$ and $\Sim$ can output oracle-aided programs.

This definition helps us circumvent the strong impossibility result aforementioned: $\Sim$ can also make up a program that uses the single effective query oracle $f_{\$, 1}$. If the distinguisher tries to gently query the oracle by only checking the correctness of evaluations, $\Sim$ can also gently query $f_{\$, 1}$ oracle as well, which does not prevent further queries.

\subsection{Single-Query Learning Game and Learnability}
\label{sec:unlearnability_defs}

In this section, we will give some further characterization and generalization of what types of functions can be made into one-time sampling programs with a meaningful security notion.

Eventually, we want the adversary not to evaluate the program twice. Clearly, this is only possible if the function itself cannot be learned with a single oracle query or at least it should be hard to learn two input-output pairs given one oracle query. We formalize such unlearnability through several definitions in this section.

The first definition is the most natural, which requires the adversary to be able to evaluate on two different inputs of its own choice correctly.
\begin{definition}[Single-Query Learning Game]
\label{def:single_query_learning_game}
A learning game for a function family $\cF = \{\cF_\lambda: [N] \to [M]\}$, a distribution family $\cD = \{D_f\}$, 
and an adversary $\calA$ is denoted as 
$\LG^\calA_{\mathcal{F}, \cD}(1^\lambda)$ 
which consists of the following steps:
\begin{enumerate}
    \item \textbf{Sampling Phase}: At the beginning of the game, the challenger takes a security parameter $\lambda$ and samples a function $(f, \aux_f) \gets \cF_\lambda$ at random according to distribution $\cD_f$,  where $\aux_f$ is some \emph{classical} auxiliary information.
    
    \item \textbf{Query Phase}: $\calA$ then gets a single oracle access to $f$\footnote{If we consider a single physical (effective, resp.) query to the oracle, then the learnability property will correspond to single physical (effective, resp.) query learnability. Unless otherwise specified (by emphasizing whether the query is physical/effective), the remaining definitions in the rest of this work refer to both cases.} and some classical auxiliary information $\aux_f$;
    \item \textbf{Challenge Phase}: 
    \begin{enumerate}
        \item $\calA$ outputs two input-output tuples $(x_1, r_1; y_1), (x_2, r_2; y_2)$ where $(x_1, r_1) \neq (x_2, r_2)$.
        \item Challenger checks if $f(x_1, r_1) = y_1$ and $ f(x_2, r_2) = y2$.
    \end{enumerate}
Challenger outputs 1 if and only if both the above checks pass.
\end{enumerate}
  
\end{definition}

\begin{definition}[Single-Query $\gamma$-Unlearnability]
\label{remark:gamma_unlearnability}
 Let $\lambda$ be the security parameter and let $\gamma = \gamma(\lambda)$ be a function. A function family $\calF = \{\calF_\lambda\}_{\lambda \in \N}$ is single-query $\gamma$-unlearnable if for all (non-uniform) quantum polynomial-time adversaries $\calA$,
  \begin{align*}
        \Pr_{f \leftarrow \calF_n} [\LG^\calA_{\mathcal{F}, \cD}(1^\lambda) = 1] \le \gamma(\lambda).
    \end{align*}
\end{definition}

The most often used unlearnability definition refered to in this work is when $\gamma = \negl(\lambda)$.
We will often refer to the following definition as "Single-Query unlearnable" for short.

\begin{definition}[Single-query $\negl(\lambda)$-unlearnable functions]\label{def:single-query-unlearnable}
    Let $\lambda$ be the security parameter. A function family $\calF = \{\calF_\lambda\}_{\lambda \in \N}$ is single-query unlearnable if for all (non-uniform) quantum polynomial-time adversaries $\calA$,
  \begin{align*}
        \Pr_{f \leftarrow \calF_n} [\LG^\calA_{\mathcal{F}, \cD}(1^\lambda) = 1] \le \negl(\lambda).
    \end{align*}
\end{definition}

\begin{remark}[Examples of Single-query $\negl(\lambda)$-Unlearnable Functions]
\label{remark:negl_unlearnable_example}
     For example, a random function $\cF: \calX \to \calY$, when the range size $\calY$ of $f$ is superpolynomially large in the security parameter, it is easy to show that it is single-physical-query $\negl(\lambda)$-unlearnable via existing quantum random oracle techniques, such as by applying Theorem 4.2 in \cite{yamakawa2021classical} or Lemma 5 in \cite{zhandry19compressed}. We refer to \Cref{sec:exampls_seq_unlearnable} for more details. 
\end{remark}

\begin{remark}[Examples of Single-Query $\gamma$-Unlearnability]
\label{remark:examples_gamma}
    For some functionalities, $\calA$ may be able to learn two input-outputs with a larger probability, (e.g. for a binary outcome random function, a random guess would succeed with probability at least $1/2$), but not non-negligibly larger than some threshold $\gamma(\lambda)$. 
\end{remark}

\paragraph{Generalized Unlearnability}
Note that the above single-query learnability is not strong enough when we want $f$ to be a cryptographic functionality: the adversary may learn something important about $(x_1, r_1, y_1), (x_2,r_2,y_2)$ without outputting the entire input-output pair.

For example, when $f$ is a pseudorandom function: it may not suffice to guarantee that $\calA$ cannot compute correctly two pairs of input-outputs. We should also rule out $\calA$'s ability to win an indistinguishability-based pseudorandomness game for  \emph{both inputs} $(x_1, r_1, y_1), (x_2,r_2,y_2)$. 

Before going into this more generic definition, we first define a notion of "predicate" important to our generic definition.
\begin{definition}
[Predicate]
\label{def:predicate}
    A predicate $P(f, x, r, z, \ans)$ is a binary outcome function that runs a program $f$ on a some input $(x,r)$ to get output $y$, and 
    outputs 0/1 depending on whether the tuple $(x, r, y, z, \ans) $ satisfies a binary relation $R_f$  corresponding to $f$: $(x, r, y, z, \ans) \in R_f$.
 $z$ is some auxiliary input that specifies the
relation.
\end{definition}

Note that the above predicate definition implicates that with the capability to evaluate $f$ (even if using only oracle) access, one has the ability to verify whether the predicate $P(f,x,r,z,\ans)$ is satisfied.

\begin{remark}
We provide two concrete examples:

\begin{enumerate}
    \item A first concrete example for the above predicate is a secret-key  encryption scheme: $f$ is an encryption function. The predicate is encrypting a message $x$ using $r$ and $\ans$ is an alleged valid ciphertext on message $x$ using randomness $r$. Then the predicate $P$ is simply encrypting $x$ using $r$ to check if $\ans$ is the corresponding ciphertext.

    \item Another example is when $f$ is a signing function. The predicate signs message $x$ using randomness $r$ and checks if $\ans$ is a valid signature for message $x$ with randomness $r$.

    \item When $f$ is a PRF, a possible predicate is to check if an alleged evaluation $y$ is indeed the evaluation $\prf(k, x\vert r)$. 
\end{enumerate}
\end{remark}

\begin{definition}[Generalized Single-Query Learning Game]
\label{def:generalized_learning_game}
 learning game for a sampler $\samp$ (which samples a function in $\mathcal{F}_\lambda$), a predicate $P = \{P_\lambda\}$, 
and an adversary $\calA$ is denoted as $\GLG^\calA_{\samp, P}(1^\lambda)$, which consists the following steps:
\begin{enumerate}
    \item \textbf{Sampling Phase}: At the beginning of the game, the challenger samples $(f, \aux_f) \gets \samp(1^\lambda)$, where $\aux_f$ is some \emph{classical} auxiliary information.
  
    \textbf{Query Phase}: $\calA$ then gets a single oracle access to $f$ and also gets $\aux_f$;
    \item \textbf{Challenge Phase}: 
    \begin{enumerate}
        \item $\calA$ outputs two input-randomness pairs $(x_1, r_1), (x_2, r_2)$ where $(x_1, r_1) \neq (x_2, r_2)$.
        \item Challenger prepares challenges $\ell_1, \ell_2$ i.i.d using $(x_1, r_1)$ and $(x_2, r_2)$ respectively and sends them to $\calA$.

        \item $\calA$ outputs answers $\ans_1, \ans_2$ for challenges $\ell_1, \ell_2$.
    \end{enumerate}
\end{enumerate}
The game outputs $1$ if and only if the answers satisfy the predicate $P_\lambda(f, x_1, r_1, \ell_1, \ans_1)$ and $P_\lambda(f, x_2, r_2, \ell_2, \ans_2)$ are both satisfied.

\end{definition}

\begin{definition}[Generalized Single-query $\gamma$-Unlearnable functions]
\label{def:generalized-single-query-unlearnable}
    Let $\lambda$ be the security parameter. A function family $\calF = \{\calF_\lambda\}_{\lambda \in \N}$ is generalized single-query $\gamma$ unlearnable for some $\gamma = \gamma(\lambda)$ if for all (non-uniform) quantum polynomial-time adversaries $\calA$,
  \begin{align*}
        \Pr_{(f, \aux_f) \leftarrow \samp(1^\lambda)}[\GLG^\calA_{\samp, P}(1^\lambda) = 1] \le \gamma. 
    \end{align*}
\end{definition}

\begin{remark}[Example]
To demonstrate how the above definition works on a concrete level,
we give an example of PRF: in the challenge phase, the adversary will provide  $(x_1, r_1), (x_2, r_2)$. The challenger then samples two independent, uniform random bits $b_1, b_2$: if $b_1 = 0$, let $y_1$ be the $\prf$ evaluation on $(x_1, r_1)$; else let $y_1$ be a random value. We assign $y_2$ correspondingly using $b_2$ similarly. The security guarantees that $\calA$ should not be able to output  correct guesses for both $b_1, b_2$ with overall probability non-negligibly larger than $1/2$: $\cA$ supposedly always has the power to compute one of them correctly with probability 1; but for the other challenge, $\cA$ should not be able to win with probability larger than it can do in a regular pseudorandomness game, when it was not given the power to evaluate the PRF.
\end{remark}

All the above definitions are well-defined for single-physical-query oracles or single effective query oracles. We also make the following simple observation:

\begin{claim} \label{claim:physical_unlearnable_imply_effective_unlearnable}
Single-effective-query $\gamma$-unlearnability implies single-physical-query $\gamma$-unlearnability.
\end{claim}
The single-physical-query 
oracle is a strictly stronger oracle and therefore anything unlearnable with a single-effective query is unlearnable with a single physical query.

\begin{remark}
\label{remark:separation_physical_effective}
   However, the other way of the above implication is not true. We will give a counter example in \Cref{sec:impossibility_oracle_model}.
\end{remark}
\jiahui{ add a remark about testing quantum adv and refer to appendix }

\subsection{Operational security definitions}
\label{sec:operational_defs}



In this section, inspired by the above unlearnability definitions, we consider a further relaxation of the above simulation based definition, which we can realize security for certain functionalities in the \emph{plain model} without oracles in \Cref{sec:construction_plain_model}.

In these definitions, we can partially characterize the operations an adversary will do.

\paragraph{Strong Operational One-Time Security}  We first give a one-time security that relaxes the simulation based definitions \Cref{def:simulation-style-otp-security} and \Cref{def:single-query-classical-output-simulation-security}.

The following definition says that given a one time program for $f$, any $QPT$ (or polynomial quantum query in the oracle model) adversary should not be able to learn two input-output pairs with a noticeably larger probability than 
a simulator with single query (physical/effective, resp.) to the oracle.

\begin{definition}[Strong Operational Security]\label{def:strong-operational-security}
   A one-time (sampling) program for a function family $\calF_\lambda$ satisfies strong operational security if: for all (non-uniform) quantum polynomial-time adversaries $\calA$ receiving one copy of $\OTP(f), \aux_f, (f, \aux_f) \gets \calF_\lambda$, there is a (non-uniform) quantum polynomial-time simulator $\Sim$ that is given single (physical/effective, resp.) quantum query access to $f$, there exists a negligible function $\negl(\lambda)$, the following holds for all $\lambda \in \N$:
  \begin{align*}
          & \lvert \Pr_{f,\aux_f \leftarrow \calF_\lambda} [f(x_1, r_1) = y_1 \wedge f(x_2,r_2) = y_2: ((x_1,r_1, y_1), (x_2, r_2,y_2)) \leftarrow \calA(\OTP(f), \aux_f)] - \\ 
          &\Pr_{f,\aux_f \leftarrow \calF_\lambda} [\LG^\Sim_{f, \cD}(1^\lambda) = 1] \rvert \le \negl(\lambda).
    \end{align*}
where $\LG^\Sim_{\mathcal{F}, \cD}(1^\lambda)$ is the single query learnability game defined in  
\Cref{remark:gamma_unlearnability} and  $(x_1, r_1) \neq (x_2, r_2)$. 
\end{definition}

\begin{remark}
    We will call the above $\gamma$-strong operational one-time security if we have $\Pr[\LG^\Sim_{f,\cD}(1^\lambda)] = \gamma$.
\end{remark}


\paragraph{Generalized One-time Security}
Similar to the discussions in \Cref{sec:unlearnability_defs} the above security is not sufficient when we want $f$ to be a cryptographic functionality: the adversary may learn something important about $(x_1, r_1, y_1), (x_2,r_2,y_2)$ without outputting the entire input-output pair. 


Corresponding to the above generalized learning game in \Cref{def:generalized_learning_game}, we give the following definition:

\begin{definition}
[Generalized Operational One-Time Security Game]
\label{def:generalized_operational_otp_security_game}
A Generalized Operational One-Time Security game for a sampler $\samp$ (which samples a function in $\mathcal{F}_\lambda$), a predicate $P = \{P_\lambda\}$, 
and an adversary $\calA$ is denoted as $\GOTP^\calA_{\samp, P}(1^\lambda)$, which consists the following steps:
\begin{enumerate}
    \item \textbf{Sampling Phase}: At the beginning of the game, the challenger samples $(f, \aux_f) \gets \samp(1^\lambda)$, where $\aux$ is some \emph{classical} auxiliary information.
  
    \textbf{Query Phase}: $\calA$ then gets a single copy $\OTP(f)$ and classical auxiliary information $\aux_f$;
    \item \textbf{Challenge Phase}: 
    \begin{enumerate}
        \item $\calA$ outputs two input-randomness pairs $(x_1, r_1), (x_2, r_2)$ where $(x_1, r_1) \neq (x_2, r_2)$.
        \item Challenger prepares challenges $\ell_1, \ell_2$ i.i.d using $(x_1, r_1)$ and $(x_2, r_2)$ respectively and sends them to $\calA$.\bhaskar{What distribution are $\ell_1, \ell_2$ sampled from? Also, does the distribution that $\ell_i$ is sampled from depend on $(x_i, r_i)$?}

        \item $\calA$ outputs answers $\ans_1, \ans_2$ for challenges $\ell_1, \ell_2$.
    \end{enumerate}
\end{enumerate}
The game outputs $1$ if and only if the predicate $P_\lambda(f, x_1, r_1, \ell_1, \ans_1)$ and $P_\lambda(f, x_2, r_2, \ell_2, \ans_2)$ are both satisfied.

\end{definition}

\begin{definition}[Generalized Strong Operational Security]\label{def:gen-strong-operational-security}

   A one-time (sampling) program for a function family $\calF_\lambda$ satisfies generalized strong operational security if: for all (non-uniform) quantum polynomial-time adversaries $\calA$ 
   there is a (non-uniform) quantum polynomial-time simulator $\Sim$ that is given single (physical/effective, resp.) quantum query access to $f$, there exists a negligible function $\negl(\lambda)$, the following holds for all $\lambda \in \N$:
  \begin{align*}
          &\lvert \Pr_{f, \aux_f\leftarrow \calF_\lambda} [\GOTP_{P, \samp}^\calA = 1] - \Pr_{f,\aux_f \leftarrow \calF_\lambda} [\GLG^\Sim_{\mathcal{F}, \cD}(1^\lambda)  = 1] \rvert \le \negl(\lambda).
    \end{align*}
    where $\GLG^\Sim_{\samp, P}(1^\lambda)$ is the single query learnability game defined in \Cref{def:generalized_learning_game}.
\end{definition}

\begin{remark}[Example]
    We will give a formal concrete example for the above definition for PRFs for the above in \Cref{def:generalized_operational_prf}.
\end{remark}

\paragraph{Weak Operational One-time Security} We finally present a relatively limited but intuitive definition: no efficient quantum adversary should be able to output two distinct samples, i.e.~tuples of the form $(x, r, f(x, r))$ with non-negligible probability, given the one-time program $\OTP(f)$ for $f$.

We can observe that this definition is only applicable to single-query $\negl(\lambda)$-unlearnable functions defined in \Cref{def:single_query_learning_game}. But this function class already covers many cryptographic applications that have search-based security, such as one-time signatures, encryptions and proofs.

\begin{definition}[Weak operational one-time security]\label{def:weak-operational-security}
    A one-time (sampling) program for a function family $\calF_\lambda$ satisfies weak operational one-time security if: for all (non-uniform) quantum polynomial-time adversaries $\calA$ 
   there exists a negligible function $\negl(\lambda)$, the following holds for all $\lambda \in \N$:
    \begin{align*}
        \Pr_{f,\aux_f\leftarrow \calF_\lambda}[ f(x_1, r_1) = y_1 \wedge f(x_2, r_2) = y_2 : ((x_1, r_1, y_1), (x_2, r_2, y_2)) \leftarrow \calA(\OTP(f), \aux_f)] \le \negl(\lambda).
    \end{align*}
    where $(x_1, r_1) \neq (x_2, r_2)$.
\end{definition}

\begin{remark}
\label{remark:weak_op_def_equivalent_strong_op}
    It is to observe that for a family of functions $\cF_\lambda$ which are $\gamma$-unlearnable for \emph{any inverse polynomial} $\gamma$ (i.e. $\negl(\lambda)$-unlearnable), \Cref{def:weak-operational-security} and \Cref{def:strong-operational-security} are equivalent. 

    First, it is easy to observe that \Cref{def:strong-operational-security} implies \Cref{def:weak-operational-security}.
    When $\cF_\lambda$ is $\negl(\lambda)$-unlearnable,  the winning probability of the learning game for $\Sim$ is $\negl(\lambda)$, which makes $\Pr_{f,\aux_f \leftarrow \calF_\lambda} [f(x_1, r_1) = y_1 \wedge f(x_2,r_2) = y_2: ((x_1,r_1, y_1), (x_2, r_2,y_2)) \leftarrow \calA(\OTP(f), \aux_f)]$  negligible. Therefore, if a construction satisfies strong operational security, then the adversary's probability of outputing two input-output pairs must be negligible and therefore it satisfies weak operational security. If a construction satisfies weak operational security,  the adversary's probability of outputing two input-output pairs must be negligible and therefore its difference with $\Sim$'s probability of outputing two input-output pairs must be negligible since both values are negligible, satisfying strong operational security.
\end{remark}


\paragraph{Operational one-time security for verifiable functions} Another natural definition could require security against adversaries to produce two input-output $(x, f(x, r))$ pairs without necessarily providing the randomness $r$ used to generate the output. This notion would make most sense when there is an (not necessarily efficient) verification algorithm $\mathsf{Verify}_f$ that takes pairs of the form $(x, y)$ and either accepts or rejects. Call such function families verifiable. 


\begin{definition} A randomized function family $\calF =\{f : \calX \times \calR \rightarrow \calY\}$ is \textbf{verifiable} if there is an efficient procedure to sample a function along with an associated verification key $(f, vk_f) \leftarrow \calF$, and there exists an efficient verification procedure $\Ver$ such that for all $x \in \calX$,
\begin{align*}
    \Pr_{r\leftarrow \calR}[1 \leftarrow \Ver_{\calF}(vk_f, x, f(x, r))] \ge 1 - \negl(\lambda).
\end{align*}
\end{definition}

\begin{definition}[Operational one-time security for verifiable functions]\label{def:verifiable-operational-security}
    A one-time (sampling) program for a function family $\calF_\lambda$ satisfies verifiable operational one-time security if: for all (non-uniform) quantum polynomial-time adversaries $\calA$ 
   there exists a negligible function $\negl(\lambda)$, the following holds for all $\lambda \in \N$:
    \begin{align*}
        \Pr_{(f,vk_f)\leftarrow \calF_\lambda}[ \Ver_{\calF} (vk_f, x_1, y_1) = \Ver_{\calF} (vk_f, x_2, y_2) = 1 : ((x_1, y_1), (x_2, y_2)) \leftarrow \calA(\OTP(f), vk_f)] \le \negl(\lambda).
    \end{align*}
    where $x_1 \neq x_2$.
\end{definition}

\jiahui{add a comment about the relation between Sim-based def and generalized def}

\newcommand{\bfr}{\mathbf{r}}
\newcommand{\bfR}{\mathbf{R}}

\subsection{Relationships among the definitions}
\label{sec:relations_defs}
\jiahui{I will come back to make changes to this section}
\begin{figure}[H]
    \centering
\begin{tikzpicture}[
    box/.style={rectangle,draw,align=center,text width=5cm},
    arrow/.style={-{Stealth[scale=1.5]}}]

\node (singlephysicalquery) [box] {Single-(physical) query simulation-based security\\\Cref{def:single-query-simulation-security}};
\node (singlephysicalqueryclassicaloutput) [box, below of=singlephysicalquery,xshift=-4cm,yshift=-2cm] {Single-query classical-output simulation-based security \\\Cref{def:single-query-classical-output-simulation-security}};
\node (SEQ) [box, below of=singlephysicalquery,xshift=4cm,yshift=-2cm] {Single Effective Query simulation-based security \\\Cref{def:simulation-style-otp-security}};
\node (generalizedoperational) [box, below of=SEQ,xshift=-4cm,yshift=-2cm] {Generalized Strong operational security\\\Cref{def:gen-strong-operational-security}};
\node (strongoperational) [box, below of=SEQ,xshift=-4cm,yshift=-4cm] {Strong operational security\\\Cref{def:strong-operational-security}};
\node (weakoperational) [box, below of=strongoperational,yshift=-1cm] {Weak operational security\\\Cref{def:weak-operational-security}};

\draw [arrow] (singlephysicalquery) -- (singlephysicalqueryclassicaloutput);
\draw [arrow] (singlephysicalquery) -- (SEQ);
\draw [arrow] (singlephysicalqueryclassicaloutput) -- 
(strongoperational);
\draw [arrow] (SEQ) -- 
(generalizedoperational);
\draw [arrow] (strongoperational) -- (weakoperational);
\draw[arrow](generalizedoperational) -- (strongoperational);
\draw [arrow] (weakoperational) -- node[right,text width=5cm,align=center]{Single-query $\negl$- \\unlearnable \Cref{def:single-query-unlearnable}}(strongoperational);
\end{tikzpicture}
    \caption{Relationships between our OTP definitions}
    \label{fig:relations_defs}
\end{figure}

We sort out the relations among the definitions we discussed in \Cref{fig:relations_defs}.

Note that as we discussed in \Cref{remark:weak_op_def_equivalent_strong_op}: the weak operational security also implies strong operational security when the function family is single (physical/effective)-query $\negl(\lambda)$-unlearnable. The implication from generalized strong operational security to strong operational security is also easy to see and we omit the proof.

The next two lemmas are true simply because the adversary in the hypothesized security definition (\Cref{def:single-query-simulation-security}) is stronger than the adversary in the implied security definitions (\Cref{def:single-query-classical-output-simulation-security} and \Cref{def:simulation-style-otp-security}).

\begin{lemma}
    Suppose $\OTP$ is a one-time program compiler that satisfies the single-query simulation-based security definition (Definition~\ref{def:single-query-simulation-security}) for a function family $\calF$. Then, it also satisfies the single-query classical-output simulation-based security definition (Definition~\ref{def:single-query-classical-output-simulation-security}) for $\calF$.
\end{lemma}



\begin{lemma}
    Suppose $\OTP$ is a one-time program compiler that satisfies the single-query simulation-based security definition (Definition~\ref{def:single-query-simulation-security}) for a function family $\calF$. Then, it also satisfies the single-effective-query simulation-based security definition (Definition~\ref{def:simulation-style-otp-security}) for $\calF$.
\end{lemma}

\begin{lemma}
    Suppose $\OTP$ is a one-time program compiler that satisfies the single-query classical-output simulation-based security (Definition~\ref{def:single-query-classical-output-simulation-security}) for a function family $\calF$. Then, it satisfies  strong operational security (Definition~\ref{def:strong-operational-security}) $\calF$.
\end{lemma}

\begin{proof}
We can observe by looking into the definition~\ref{def:strong-operational-security} that the adversary and simulator in this definition are simply a special case of the classical-output adversary and simulator, by outputting two correctly evaluated input-output pairs. 
\end{proof}

\begin{lemma}
    Suppose $\OTP$ is a one-time program compiler that satisfies the single-effective-query simulation-based security (Definition~\ref{def:simulation-style-otp-security}) for a function family $\calF$. Then, it satisfies generalized strong operational security (Definition~\ref{def:gen-strong-operational-security}) $\calF$.
\end{lemma}

\begin{proof}
The difference between definition~\ref{def:gen-strong-operational-security} (resp. \Cref{def:generalized_learning_game}) and \Cref{def:strong-operational-security} (respectively \Cref{def:single_query_learning_game}) is that we can view the adversary/simulator as outputting a potentially quantum state together with its own choice of $(x_1, r_1)$ and $(x_2, r_2)$ in the challenge phase; then the quantum state is going to answer the challenger's challenges. Therefore, it is a special case of the single-physical/effective-query (depending which oracle we give to $\Sim$) simulation definition with quantum outputs, but not necessarily a special case of the simulation definition with classical outputs.
\end{proof}

\section{Single-Effective-Query Construction in the Classical Oracle Model}
\subsection{Construction}\label{sec:construction}\label{sec:seq-construction}
For a function $f:\mathcal{X} \times \mathcal{R} \rightarrow \mathcal{Y}$ with $\cX = \{0,1\}^m$, and security parameter $\lambda$, we construct a one-time program $\OTP(f)$ for $f$ as described in \Cref{fig:otp-construction}.

\begin{figure}[h]
\begin{mdframed}
$\generate(1^\lambda, f)$:
\begin{enumerate}
  \item Let $n = \lambda$. Then sample a random oracle $G: \{0,1\}^{m \cdot n} \rightarrow \mathcal{R}$.
  \item For each $i \in [m]$: sample a random subspace $A_i \subseteq \mathbb{F}_2^n$ of dimension $n/2$, and create the membership oracles $\mathcal{O}_{A_i^0}, \mathcal{O}_{A_i^1}$ for $A_i$ and $A_i^\perp$ respectively.
  \begin{align*}
    \mathcal{O}_{A_i^0}(v) = \begin{dcases}
      1 &\text{if }v \in A_i \backslash \{0\},\\
      0 &\text{otherwise},
    \end{dcases}
    \quad \text{ and }\quad
    \mathcal{O}_{A_i^1}(v) = \begin{dcases}
      1 &\text{if }v \in A_i^\perp \backslash \{0\},\\
      0 &\text{otherwise}.
    \end{dcases}
  \end{align*}
  Also let $\bfA = (A_i)_{i \in [m]}$.
  \item Create oracle $\mathcal{O}_{f, G, \bfA}$ that takes as input $x \in \{0,1\}^m$, $\vecv  = (v_1, \ldots, v_m)$ where $v_i \in \bbF_2^n$, and $u \in \cY$, and outputs
  \begin{align*}
    \mathcal{O}_{f, G, \bfA}(x, \vecv, u) = \begin{dcases}
      (x, \vecv, u \oplus f(x, G(\vecv))) &\text{if }\mathcal{O}_{A^{x_i}_i}(v_i) = 1 \text{ for all }i\in[m],\\
      (x, \vecv, u) &\text{otherwise}
    \end{dcases}
  \end{align*}
  \item Output $\OTP(f) = \left(\left(\ket{A_i}\right)_{i\in[m]}, \mathcal{O}_{f, G, \bfA}\right)$. 
\end{enumerate}

\noindent $\Eval(\OTP(f), x)$:
\begin{enumerate}
    \item For each $i \in [m]$, compute:
    \[\ket{\psi_i} = \left(H^{\otimes n}\right)^{x_i} \ket{A_i}\]
    \item Prepare the following state on the query register $\cQ$:
    \[\ket{x}_{\cQ_\cX} \otimes \ket{\psi_1} \otimes \dots \otimes \ket{\psi_m} \otimes \ket{0}_{\cQ_\cY}\]
    where $0 \in \cY$. Apply $\cO_{f, G, \bfA}$ to $\cQ$.
    \item Measure the $\cQ_\cY$ register to obtain a value $y \in \cY$ and output $y$.
\end{enumerate}
\end{mdframed}
\caption{Construction of $\OTP(f)$}\label{fig:otp-construction}
\end{figure}

\begin{theorem}\label{thm:oracle-construction-is-correct}
    For any function $f \in \cF$ that maps $\cX \times \cR \to \cY$ where $\cX = \{0,1\}^m$, the OTP construction given in \Cref{fig:otp-construction} satisfies correctness (\Cref{def:otp_correctness}).
\end{theorem}
\begin{proof}
    For every $i \in [m]$,
    \[\ket{\psi_i} = \left(H^{\otimes n}\right)^{x_i} \ket{A_i} = 
    \begin{cases}
        \ket{A_i}, & x_i = 0\\
        \ket{A_i^\perp}, & x_i = 1
    \end{cases}\]
    Then the state $\ket{\psi_1} \otimes \dots \otimes \ket{\psi_m}$ gives an overwhelming fraction of its amplitude to values $\vecv = (v_1, \dots, v_m)$ for which $\mathcal{O}_{A^{x_i}_i}(v_i) = 1$ for all $i\in[m]$.

    Then after applying $\cO_{f, G, \bfA}$ to $\cQ$, the state of the $\cQ$ register gives an overwhelming fraction of its amplitude to values $(x, \vecv, y)$ such that $y = f(x, G(\vecv))$.

    Note that $G(\vecv)$ is uniformly random over $\cR$ due to the randomness of $G$. Then the output of $\Eval(1^\lambda, \OTP(f), x)$ is negligibly close in statistical distance to $f(x, R)$, where $R$ is uniformly random over $\cR$.
\end{proof}

\begin{theorem}\label{thm:oracle-construction-is-secure}
Let $\lambda$ be the security parameter. Let $m, \ell \in \N$.
    For any function $f \in \cF: \cX \times \cR \to \{0,1\}^{\ell}$ for which we let $\cX = \{0,1\}^m$, the OTP construction given in \Cref{fig:otp-construction} satisfies the single-effective query simulation-based OTP security notion of \Cref{def:simulation-style-otp-security}.
\end{theorem}

In general, $f(x, r)$ may or may not output its randomness $r$. This will not modify our proof.

Our construction does not have any requirements on $f$'s input lengths, for either the message or randomness. However, we remark that the SEQ simulation definition may not be very meaningful when the randomness is small. For example, if $|\calR|$ is only polynomially large, then any measurement made to $f$ can at most disturb the program state by $1-1/\poly$. Thus, the adversary may be able to perform a second query with a $1/\poly$ success rate.


The above theorem combined with the relationships between security definitions in \Cref{sec:relations_defs} 
directly gives the following corollary:
\begin{corollary}
\label{cor:otp_for_unlearnable}
For function families satisfying the unlearnability definitions in 
\Cref{def:generalized-single-query-unlearnable} (or \Cref{remark:gamma_unlearnability}, \Cref{def:single-query-unlearnable} respectively), there exists secure one-time sampling programs for them in the classical oracle model with respect to security definition \Cref{def:gen-strong-operational-security} (or \Cref{def:strong-operational-security}, \Cref{def:weak-operational-security}) resp.).
\end{corollary}

\subsection{Proof of Security (\Cref{thm:oracle-construction-is-secure})}
The simulator is defined in \Cref{fig:simulator}.

\paragraph{Intuition.} To gain intuition about the simulator, we first recall the differences between the real and ideal worlds.
In the real world, the OTP uses randomness $G(v)$ to evaluate $f$, where $v$ is a measurement of $\ket{A}$ in a basis corresponding to the chosen input $x$. If $f$ is sufficiently random, then measuring $f(x, G(v))$ may collapse $\ket{A}$, preventing further queries. On the other hand, in the ideal world, the SEQ oracle uses randomness $H(x)$ to evaluate $f$. If $f$ is sufficiently random, then measuring $f(x;H(x))$ may collapse the SEQ oracle's internal state, preventing further queries.

The main gaps that the simulator needs to bridge between these worlds are the usage of $\ket{A}$ versus the usage of an internal state to control query access, as well as the usage of $G(v)$ versus $H(x)$. 
Since $G$ and $H$ are internal to the OTP oracle and SEQ oracle, respectively, the latter is not an issue even if $f$ outputs its randomness directly.
The simulator addresses the former by maintaining a cache for subspace vectors. If it detects that the SEQ oracle will not permit other queries, it stores the most recent subspace vector locally. This collapses $\ket{A}$ in the view of the adversary, ensuring that a successful $f$ evaluation in the ideal world looks similar to if $f$ were evaluated using $\ket{A}$ in the real world.

\begin{figure}[H]
\begin{mdframed}
\begin{enumerate}
    \item For each $i \in [m]$, sample a subspace $A_i \subseteq \bbF_2^\secpar$ of dimension $\secpar/2$ uniformly at random. 
    \item Initialize vector cache register $\cV = \cV_1 \times \dots \times \cV_m$ to $\ket{0^\secpar}_{\cV_1} \otimes \dots \otimes \ket{0^\secpar}_{\cV_m}$.
    \item Prepare an oracle $\cO_{\Sim}$ as follows.
    \begin{enumerate}
        \item $\cO_{\Sim}$ acts on a query register $\calQ = (\calQ_x, \calQ_\vecv, \calQ_u)$, which contains superpositions over states of the form $\ket{x, \vecv, u}$, where $x \in \cX$, $\vecv = (v_1, \dots, v_m) \in \bit^{n \cdot m}$, $u \in \cY$.
        \item If $\cV$ contains $0^{n \cdot m}$ or $\vecv$, and if $v_i \in A_i^{x_i}\backslash\{0\}$ for all $i\in [m]$, then $\cO_{\Sim}$ does the following:
        \begin{enumerate}
            \item Prepare $\ket{x\oplus 1, 0, 0}$ in a register $\calQ' = (\calQ'_x, \calQ'_u, \cB')$, then query $f_{\$,1}$ on register $\calQ'$. \label{oraclesim:cachecheck1} 
            \item If $\cB'$ has value $0$, apply a CNOT from register $\calQ_\vecv$ to register $\calV$.
            \item Uncompute step \ref{oraclesim:cachecheck1}

            
            
            \item Query $f_{\$,1}$ on $\ket{x, u, 0}_{\calQ_x, \calQ_u, \cB}$.

            \item Prepare $\ket{x\oplus 1, 0, 0}$ in a register $\calQ' = (\calQ'_x, \calQ'_u, \cB')$, then query $f_{\$,1}$ on register $\calQ'$. \label{oraclesim:cachecheck2}
            \item If $\cB'$ has value $0$, apply a CNOT from register $\calQ_\vecv$ to register $\calV$. 
            \item Uncompute step \ref{oraclesim:cachecheck2}.
        \end{enumerate}
    \end{enumerate}
    \item Output $\left(\left(\ket{A_i}\right)_{i \in [m]}, \cO_{\Sim}\right)$.
\end{enumerate}
\end{mdframed}
\caption{Simulator $\Sim^{f_{\$,1}}$}\label{fig:simulator}
\end{figure}

\paragraph{Analysis of the Simulator.} 
Consider the following hybrid experiments:
\begin{itemize}
    \item $\underline{\Hyb_0}:$ The real distribution $\OTP(f)$. Recall that $\OTP(f)$ outputs an oracle $\calO_{f,G,A}$ which acts as follows on input $(x, v, u)$:
    \begin{enumerate}
        \item \textbf{Vector Check:} It checks that $v_i\in A^{x_i}_i\backslash \{0\}$ for all $i\in [\secpar]$. If not, it immediately outputs $(x, v, u)$.
        \item \textbf{Evaluation:} Compute $u\oplus f(x; G(v))$.
        \item Output $(x, v, u\oplus f(x; H(v))$.
    \end{enumerate}
    
    \item $\underline{\Hyb_1}:$ The only difference from $\Hyb_0$ is that $G$ is implemented as a compressed oracle. $\calO_{f,G,A}$ maintains the compressed oracle's database register $\calD$ internally. To evaluate $f$ in step 2 \jiahui{which figure?}, it queries $\ket{v, 0}_{\calQ'}$ to $G$ in register $\calQ'$ to obtain $\ket{v, G(v)}_{\calQ'}$, then applies the isometry
    \[
        \ket{x, v, u}_{\calQ} \otimes \ket{v, r}_{\calQ'} 
        \mapsto 
        \ket{x, v, u \oplus f(x;r)}_{\calQ} \otimes \ket{v, r}_{\calQ'}
    \]
    to registers $\calQ$ and $\calQ'$, and finally queries $G$ on register $\calQ'$ again to reset it to $\ket{v, 0}_{\calQ'}$.

    \item $\underline{\Hyb_2}:$ The only difference from $\Hyb_1$ is in step 1 of $\calO_{f,G,A}$. Instead of checking that $v_i\in A^{x_i}_i\backslash \{0\}$, it checks that $v_i \in A^{x_i}_i\backslash (A_i \cap A^\perp_i)$. 

    \item $\underline{\Hyb_3}:$ The only difference from $\Hyb_2$ is we add a single-effective-query check to $\calO_{f,G,A}$. It now answers basis state queries $\ket{x, v, u}$ as follows:
    \begin{enumerate}
        \item \textbf{Vector Check:} Check that $v_i\in A_i^{x_i}\backslash (A_i \cap A_i^\perp)$ for all $i\in [\secpar]$. If not, immediately output $(x, v, u)$.
        \item \textbf{SEQ Check {\color{red}(New)}:} Look inside $G$'s compressed database register $\calD$ to see if there is an entry of the form $(v', r)$ for some $r$ and $v'\notin\{0, v\}$. If so, then immediately output register $\calQ$. \label{proof:SEQ-hyb3-seqcheck}
        \item \textbf{Evaluation:} Compute $\ket{x, v, v} \mapsto \ket{x, v, u\oplus f(x; G(v))}$ on register $\calQ$. This involves querying the compressed oracle $G$ twice, as described in $\Hyb_1$.
        \item Output register $\calQ$.
    \end{enumerate}

    \item $\underline{\Hyb_4}:$ The only difference from $\Hyb_2$ is we add a caching routine to $\calO_{f, G,A}$.\footnote{Intuitively, this hybrid will ensure that whenever some $v$ corresponding to input $x$ is recorded, $x$ is also recorded. Thus, the internal oracle state will look like $\ket{0, \emptyset}$ or $\ket{x, \{(v, r)\}}$ in $\Hyb_4$. This intuition is made formal in \Cref{claim:SEQ-proof-hyb4-invariant}.}
    The oracle maintains a register $\calR_x$ which is initialized to $\ket{0}$. On receiving query $(x, v, u)$, the oracle $\calO_{f,G,A}$ does the following:
    \begin{enumerate}
        \item \textbf{Vector Check:} Check that $v_i\in A_i^{x_i}\backslash (A_i \cap A_i^\perp)$ for all $i\in [\secpar]$. If not, immediately output $(x, v, u)$.
        
        \item \textbf{SEQ Check:} Do the single-effective query check that was added in $\Hyb_3$.
        
        \item \textbf{Cache 1 {\color{red}(New)}:} Look inside $G$'s compressed database register $\calD$. If it contains a nonempty database $D \neq \emptyset$, then perform a CNOT from register $\calQ_x$ to register $\calR_x$.
        \label{proof:SEQ-hyb4-cache1}
        \jiahui{Remind the readers what is register $\calR_x$ ?}
        \item \textbf{Evaluation:} Apply the isometry $\ket{x,v, u}_{\calQ} \mapsto \ket{x, v, u \oplus f(x; G(v))}_{\calQ}$ to register $\calQ$. This involves querying the compressed oracle $G$ twice, as described in $\Hyb_1$.
        
        \item \textbf{Cache 2 {\color{red}(New)}:} Look inside $G$'s compressed database register $\calD$. If there is an entry of the form $(v', r)$ for some $r$ and $v'$, then perform a CNOT from register $\calQ_x$ to register $\calR_x$.
        \label{proof:SEQ-hyb4-cache2}

       \item It outputs register $\calQ$.
    \end{enumerate}

    \item $\underline{\Hyb_5}:$ In this hybrid, $\calO_{f, G, A}$ swaps the role of $v$ and $x$ in the cache and compressed oracle.\footnote{Intuitively, this modifies the internal oracle state from $\ket{x, \{(v, r)\}}$ to $\ket{v, \{(x, r)\}}$, without modifying the case where the oracle state would be $\ket{0, \emptyset}$.}
    The oracle, which we rename to $\calO_{f, H, A}$, maintains a cache register $\calV$ and a compressed oracle $H:\calX \rightarrow \calR$ instead of $G:\bbF_2^n \rightarrow \calR$. On query $(x, v, u)$, it does the following:
    \begin{enumerate}
        \item \textbf{Vector Check:} Check that $v_i\in A_i^{x_i}\backslash (A_i \cap A_i^\perp)$ for all $i\in [\secpar]$. If not, immediately output $(x, v, u)$.
        \item \textbf{SEQ Check {\color{red}(Modified)}:} Look inside $\calV$ to see if it contains some $v' \notin \{0, v\}$. If so, immediately output $\calQ$.\label{proof:SEQ-hyb5-SEQcheck}
        \item \textbf{Cache 1 {\color{red}(Modified)}:} Look inside $H$'s compressed database register $\calD$. 
        If it contains a nonempty database $D \neq \emptyset$, then perform a CNOT from register $\calQ_x$ to register $\calR_x$.
        \label{proof:SEQ-hyb5-cache1}
        \item \textbf{Evaluation {\color{red}(Modified)}:} Applies the isometry $\ket{x,v, u}_{\calQ} \mapsto \ket{x, v, u \oplus f(x; H(x))}_{\calQ}$ to register $\calQ$. This involves querying the compressed oracle $H$ twice, analogously to the procedure in $\Hyb_1$. \label{proof:SEQ-hyb5-eval}
        \item \textbf{Cache 2 {\color{red}(Modified)}:} It looks inside $H$'s compressed database register $\calD$. 
        If it contains a nonempty database $D \neq \emptyset$, then perform a CNOT from register $\calQ_x$ to register $\calR_x$.
        \label{proof:SEQ-hyb5-cache2}
        \item Output register $\calQ$.
    \end{enumerate}

    \item $\underline{\Hyb_6}:$ The only difference from $\Hyb_4$ is a change to the SEQ check in step \ref{proof:SEQ-hyb5-SEQcheck}. Instead of looking inside $\calV$ to see if it contains some $v' \notin \{0, v\}$, the oracle $\calO_{f, H, A}$ instead looks inside $H$'s compressed database register $\calD$ to see if there is an entry of the form $(x', r)$ for $x\neq x'$.

    \item $\underline{\Hyb_7 = \Sim}:$ The only differences from $\Hyb_5$ are in the caching routine in steps \ref{proof:SEQ-hyb5-cache1} and \ref{proof:SEQ-hyb5-cache2}. It replaces each of these steps with the following procedure:
    \begin{enumerate}
        \item Prepare a $\ket{0}$ state in register $\calB$. Controlled on $H$'s database register having an entry of the form $(x', r)$ for $x' \neq x \oplus 1$, apply a NOT operation to register $\calB$.\label{proof:SEQ-hyb6-cache}
        \item If register $\calB$ contains $\ket{1}$, apply a CNOT operation from register $\calQ_v$ to register $\calV$.
        \item Uncompute step \ref{proof:SEQ-hyb6-cache}.
    \end{enumerate}
\end{itemize}

$\Hyb_0$ is perfectly indistinguishable from $\Hyb_1$ by the properties of a compressed random oracle. We now show indistinguishability for each of the other sequential pairs of hybrid experiments.

\begin{claim}\label{claim:SEQ-hyb1-hyb2}
    $\Hyb_1$ is computationally indistinguishable from $\Hyb_2$.
\end{claim}
\begin{proof}
    The only difference between these is that $\Hyb_2$ additionally returns early whenever some $v_i \in (A_i\cap A_i^\perp) \backslash \{0\}$. This condition can only occur with negligible probability; otherwise, we can break direct product hardness for subspace states (\Cref{thm: direct product oracle}) by embedding a challenge subspace state in a random index $i^*$, using polynomially many queries to the challenge membership oracles to run $\Hyb_1$, and after every query, if $v_i \in (A_i\cap A_i^\perp) \backslash \{0\}$ for some $i \in \secpar$, measuring the query register to obtain $v$. Since $i^*$ is independent of the adversary's view, $i=i^*$ with probability $1/m$ whenever this occurs. Note that $m$ is polynomial in $\secpar$.
\end{proof}

\begin{claim}\label{claim:SEQ-hyb2-hyb3}
    $\Hyb_2$ is computationally indistinguishable from $\Hyb_3$.
\end{claim}
\begin{proof}
    It is sufficient to show that the single-effective-query check added in step \ref{proof:SEQ-hyb3-seqcheck} causes an early output only with probability $\negl(\secpar)$. We reduce this fact to the direct product hardness of subspace states (\Cref{thm: direct product oracle}). Say that some adversary $\cA_{\OTP}$ caused this event to occur with noticeable probability. We construct an adversary $\cA_{DP}$ to break direct product hardness as follows: 
    \begin{enumerate}
        \item $\cA_{DP}$ receives from the challenger $\left(\ket{A_*}, \cO_{A_*}, \cO_{A_*^\perp}\right)$ for some subspace $A_* \subseteq \bbF_2^n$ of dimension $n/2$ sampled uniformly at random.
        \item $\cA_{DP}$ samples an index $i^* \getsr [m]$ in which to embed $A_*$, and they set $\left(\ket{A_i}, \cO_{A_i}, \cO_{A_i^\perp}\right) = \left(\ket{A_*}, \cO_{A_*}, \cO_{A_*^\perp}\right)$. 
        \item For each $i \in [m] \backslash \{i^*\}$, $\cA_{DP}$ samples a subspace $A_i \subseteq \bbF_2^n$ of dimension $n/2$ uniformly at random. Then they prepare $\left(\ket{A_i}, \cO_{A_i}, \cO_{A_i^\perp}\right)$.
        \item $\cA_{DP}$ uses $\left(\ket{A_i}, \cO_{A_i}, \cO_{A_i^\perp}\right)_{i \in [m]}$ to construct $\OTP(f)$, as described in $\cH_3$. Then they run $\cA_{\OTP}$ on this construction, with the following modification: in \ref{proof:SEQ-hyb3-seqcheck}, it \emph{measures} the early return condition, instead of checking it coherently. 
        \item $\cA_{DP}$ terminates $\Hyb_3$ as soon as the measurement result indicates to return early. Then, it measures registers $\calQ_v$ and the oracle database $\calD$ to obtain $v$ and $(v', r)$. It outputs $v_{i^*}$ and $v'_{i^*}$.
    \end{enumerate}

    Let $\nu(\secp)$ be the probability that $\cA_{DP}$ finds and outputs two vectors $v_{i^*}$ and $v'_{i^*}$. $\nu(\secp)$ is non-negligible because otherwise, $\cA_{\OTP}$ would trigger the early return condition with only negligible probability.

    By definition of the early return condition, there is at least one index where $v_i\neq v'_i$. With probability $\geq \frac{1}{m}$, this index is $i^*$ because the value of $i^*$ is independent of $\cA_{\OTP}$'s view. Furthermore, by definition of $\Hyb_2$, the only vectors $w$ that are queried to $G$ are those that satisfy $w_i \in (A\cup A^\perp) \backslash(A\cap A^\perp)$, so at all times, the entries $(w,r)$ in $G$'s database satisfy this form. Therefore $v'_i \in (A\cup A^\perp) \backslash (A\cap A^\perp)$. Similarly, if step \ref{proof:SEQ-hyb3-seqcheck} is reached, then $v_i\in (A\cup A^\perp) \backslash (A\cap A^\perp)$. 
    
    Therefore whenever $v_{i^*} \neq v'_{i^*}$, $\cA_{DP}$ wins the direct product hardness game. This occurs with probability $\geq \frac{\nu(\secp)}{m}$. By \Cref{thm: direct product oracle}, $\nu(\secp)$ must be negligible.
\end{proof}

To aid in showing that $\Hyb_3$ is indistinguishable from $\Hyb_4$, we prove that the cache introduced in $\Hyb_4$ maintains the invariant that $x$ is cached if and only if some $v$ uniquely corresponding to $x$ is recorded in the compressed oracle $G$.

\begin{claim}\label{claim:SEQ-proof-hyb4-invariant}
    After every query in $\Hyb_4$, the internal states of $\calO_{f,G,A}$, consisting of the cache register $\calR_x$ and $G$'s database register $\calD$, lies entirely within the space spanned by states of the form
    \begin{align*}
        \ket{0}_{\calR_x} &\otimes \ket{\emptyset}_{\calD}
        \\
        \ket{x}_{\calR_x} &\otimes \ket{\{(v, r)\}}_{\calD}
    \end{align*}
    for some $x\in \calX$, $r\in \calR$, and $v\in \{0,1\}^{m\cdot n}$ such that $v_i\in A_i^{x_i}\backslash (A_i \cap A_i^\perp)$ for all $i\in [\secpar]$. 
    
    In other words, if the state at time $t$ is $\rho^t_{\calR_x,\calD}$ and $\Pi_{\Hyb_4}$ projects onto this space, then
    \[
        \Tr[\Pi_{\Hyb_4} \rho^t_{\calR_x,\calD}] = 1
    \]
\end{claim}
\begin{proof}
    This is true at when $f_{\$,1}$ is initialized, since registers $\calR_x, \calD$ are initialized to $\ket{0, \emptyset}$. We now show that $\calO_{f, G, A}$ is invariant on the space determined by $\Pi_{\Hyb_4}$.
    It suffices to consider the action of $\calO_{f, G,A}$ on these basis states. 

    We go through the operations of $\calO_{f, G,A}$ step-by-step. If step 1 (vector check) does not return early, then $v_i\in A_i^{x_i}\backslash (A_i \cap A_i^\perp)$ for all $i\in [\secpar]$. In step 2 (SEQ check), there are two cases.
    If $D$ has an entry $(v', r)$, then the SEQ check causes $\calO_{f,G,A}$ to return register $\calQ$ without modifying its internal state.
    On the other hand, if $v = v'$ or if $D = \emptyset$, then $\calO_{f,G,A}$ proceeds to step 3. 

    We claim that at the end of step 3 (cache 1), $\calR_x$ contains $\ket{0}$. 
    If $D = \emptyset$, step 3 leaves $\calR_x$ as $\ket{0}$. 
    If $v = v'$, $\calO_{f,G,A}$ applies a CNOT operation from $\calQ_x$ to $\calR_x$. 
    Before this operation, $\calQ_x$ contained a value $x'$ such that $v'_i=v_i \in A_i^{x_i}$ for all $i\in [\secpar]$.
    Since $v_i \notin (A\cap A^\perp)$ for all $i\in [\secpar]$, $v$ uniquely determines $x$, so $x' = x$.
    Therefore after applying the CNOT, register $\calQ_x$ contains $\ket{x \oplus x} = \ket{0}$.

    In step 4 (evaluation), the oracle queries $G$ on $v$. During this, the compressed oracle modifies its database register $\calD$, but the new databases $D'$ in the support of the state satisfy $D'(w) = D(w)$ for all $w\neq v$.
    Since $D(w) = \bot$ for all $w\neq v$, register $\calD$ is supported on $\ket{\emptyset}$ and $\ket{\{(v, r)\}}$ for some $r\in \calR$ at the end of this step.

    Applying step 5 (cache 2) produces $\ket{x}_{\calR_x} \otimes \ket{\{(v, r)\}}_{\calD}$ when $D= \{(v, r)\}$ and produces $\ket{0}_{\calR_x} \otimes \ket{\emptyset}_{\calD}$ when $D=\emptyset$. Step 6 does not modify these registers further.
\end{proof}

\begin{claim}\label{claim:SEQ-hyb3-hyb4}
    $\Hyb_3$ and $\Hyb_4$ are perfectly indistinguishable.
\end{claim}
\begin{proof}
    By \Cref{claim:SEQ-proof-hyb4-invariant}, the state of registers $(\calX, \calD)$ in $\Hyb_3$ is always supported on states of the form $\ket{0, \emptyset}$ or $\ket{x, (v, r)}$ where $x\in \calX$, $r\in \calR$, and $v\in \{0,1\}^{m\cdot n}$ such that $v_i\in A_i^{x_i}\backslash (A_i \cap A_i^\perp)$ for all $i\in [\secpar]$. This condition implies that $v$ uniquely determines $x$, which we now denote by $x_v$. Therefore there is an isometry mapping 
    \[
        \ket{0, (v, r)}_{\calX, \calD} \mapsto \ket{x_v, (v, r)}_{\calX, \calD}
    \] 
    Let $U$ be a unitary which implements this isometry. The state of $\Hyb_4$ at any time $t$ can be generated by running $\Hyb_3$ until time $t$ with the following modification: before each query, apply $U^\dagger$ and after answering it, apply $U$.
\end{proof}

\begin{claim}\label{claim:SEQ-hyb4-hyb5}
    $\Hyb_4$ and $\Hyb_5$ are perfectly indistinguishable.
\end{claim}
\begin{proof}
    Let $\ket{\psi^{\Hyb_4}_t}_{\calA, \calQ,\calV,\calD}$ be the joint state of the adversary and $\calO_{f,G,A}$ in $\Hyb_4$ when query $t$ is submitted and denote the state for $\Hyb_5$ similarly.
    Let $\calO_{f,G,A}^{\Hyb_4}$ and $\calO_{f,G,A}^{\Hyb_5}$ denote the OTP oracle's unitary operations in $\Hyb_4$ and $\Hyb_5$, respectively.
    Let $U$ be the unitary mapping
    \[
    \ket{x, (v,r)} \mapsto \ket{v, (x,r)}
    \]
    and acting as the identity on all orthogonal states. 
    We show by induction over the time $t$ that 
    \[
        (I_{A,\calQ} \otimes U_{\calV,\calD}) \ket{\psi^{\Hyb_4}_t}_{\calA, \calQ,\calV,\calD}
        =
        \ket{\psi^{\Hyb_5}_t}_{\calA, \calQ,\calV,\calD}
    \]

    This is clearly true for $t=0$, since both hybrids initialize $\calV, \calD$ to $\ket{0, \emptyset}$, which $U$ acts as the identity on. Now consider some time $t$. By the inductive hypothesis and linearity of quantum computation, it suffices to consider the actions of $\calO^{\Hyb_4}_{f,G,A}$ and $\calO^{\Hyb_5}_{f,G,A}$ on the \emph{same} basis state $\ket{x, v, u}_{\calQ} \otimes \ket{x', D}_{\calR_x, \calD}$. Furthermore, by \Cref{claim:SEQ-proof-hyb4-invariant}, we may restrict ourselves to basis states of the form
    \begin{align*}
        \ket{x, v, u}_{\calQ} &\otimes \ket{0, \emptyset}_{\calR_x, \calD}
        \\
        \ket{x, v, u}_{\calQ} &\otimes \ket{x', (v', r)}_{\calR_x, \calD}
    \end{align*}
    where $v'_i \in A_i^{x'_i}\backslash (A_i \cap A_i^\perp)$.

    Observe that step 1 (vector check) of $\calO^{\Hyb_4}_{f,G,A}$ and $\calO^{\Hyb_5}_{f,G,A}$ are identical, and that step 2 (SEQ check) is also identical after applying $U$ to registers $\calR_x=\calV$ and $\calD$. If either step prompts an early return, then the states are identical. Otherwise, the input state is now constrained to (1) $D= \emptyset$ and $x' = 0$ or (2) $D\in \{\{(v, r)\}\}_{r\in \calR}$ and $x = x'$. In either case, after step 3, the state of $\Hyb_3$ and $\Hyb_4$ are, respectively:
    \begin{align*}
        \ket{x,v,u}_{\calQ} &\otimes \ket{0, D_{v}}_{\calR_x, \calD}
        \\
        \ket{x,v,u}_{\calQ} &\otimes \ket{0, D_{x}}_{\calV, \calD}
    \end{align*}
    where either (1) $D_{v} = D_{x} = 0$ or $D_{v} = \{(v,r)\}$ and $D_{x} = \{(x,r)\}$ for some $r\in \calR$. In either case, $D_v$ is related to $D_x$ by $D_v(v) = D_x(x)$, $D_v(x) = \bot = D_x(v)$ and $D_v(w) = D_v(w)$ for all $w\notin \{x, v\}$. 
    
    By \Cref{claim:compressed-oracle-renaming}, the states after each query to $G$ (respectively $H$) in step 4 (evaluation) are identical up to renaming $v$ to $x$ in the compressed database. Thus, at the end of step 4, the states are, respectively:\footnote{$\ket{\psi_\emptyset}$ is not necessarily associated with a single $r$ value, due to the potential of collisions $f(x;r_1) = f(x;r_2)$ which might lead to partial compression.}
    \begin{gather*}
        \alpha_{\emptyset} \ket{x,v,\psi_\emptyset}_{\calQ} \otimes \ket{0, \emptyset}_{\calR_x, \calD} + \sum_{r} \alpha_{r} \ket{x,v, u \oplus f(x;r)}_{\calQ} \otimes \ket{0, \{(v, r)\}}_{\calR_x, \calD}
        \\
        \alpha_{\emptyset} \ket{x,v,\psi_\emptyset}_{\calQ} \otimes \ket{0, \emptyset}_{\calV, \calD} + \sum_{r} \alpha_{r} \ket{x,v, u \oplus f(x;r)}_{\calQ} \otimes \ket{0, \{(x, r)\}}_{\calV, \calD}
    \end{gather*}
    After applying step 5, the states become
    \begin{gather*}
        \alpha_{\emptyset} \ket{x,v,\psi_\emptyset}_{\calQ} \otimes \ket{0, \emptyset}_{\calR_x, \calD} + \sum_{r} \alpha_{r} \ket{x,v, u \oplus f(x;r)}_{\calQ} \otimes \ket{x, \{(v, r)\}}_{\calR_x, \calD}
        \\
        \alpha_{\emptyset} \ket{x,v,\psi_\emptyset}_{\calQ} \otimes \ket{0, \emptyset}_{\calV, \calD} + \sum_{r} \alpha_{r} \ket{x,v, u \oplus f(x;r)}_{\calQ} \otimes \ket{v, \{(x, r)\}}_{\calV, \calD}
    \end{gather*}
    Finally, $U^\dagger$ maps the latter to the former.
\end{proof}

To aid in the proof that $\Hyb_5$ and $\Hyb_6$ are indistinguishable, we show that a cache invariant holds for $\Hyb_5$ and $\Hyb_6$ which is similar to the one for $\Hyb_4$.

\begin{claim}\label{claim:SEQ-proof-hyb5-invariant}
    After every query in $\Hyb_5$, the internal states of $\calO_{f,G,A}$, consisting of the cache register $\calR_x$ and $G$'s database register $\calD$, lies entirely within the space spanned by states of the form
    \begin{align*}
        \ket{0}_{\calR_x} &\otimes \ket{\emptyset}_{\calD}
        \\
        \ket{x}_{\calR_x} &\otimes \ket{\{(v, r)\}}_{\calD}
    \end{align*}
    for some $x\in \calX$, $r\in \calR$, and $v\in \{0,1\}^{m\cdot n}$ such that $v_i\in A_i^{x_i}\backslash (A_i \cap A_i^\perp)$ for all $i\in [\secpar]$.
\end{claim}
\begin{proof}
    As shown in the proof of \Cref{claim:SEQ-hyb4-hyb5}, the internal states of $\calO_{f,G,A}$ in $\Hyb_4$ and $\Hyb_5$ are related by the register renaming $\calR_x = \calV$ and a unitary $U$ mapping 
    \[
        \ket{x}_{\calR_x} \otimes \ket{\{(v, r)\}}_{\calD} 
        \mapsto 
        \ket{v}_{\calV} \otimes \ket{\{(x, r)\}}_{\calD} 
    \]
    and acting as the identity on all other states. By \Cref{claim:SEQ-proof-hyb4-invariant}, the internal state in $\Hyb_4$ is supported on $\ket{0, \emptyset}$, which $U$ acts trivially on, or $\ket{x, \{(v, r)\}}$ where $v_i \in A_i^{x_i}\backslash (A_i \cap A_i^\perp)$, which $U$ maps to $\ket{v, \{(x,r)\}}$.
\end{proof}

\begin{claim}\label{claim:SEQ-hyb5-hyb6}
    $\Hyb_5$ and $\Hyb_6$ are perfectly indistinguishable.
\end{claim}
\begin{proof}
    The only change between these two hybrids is the single-effective-query check in step 2. 
    $\Hyb_5$ returns early if and only if $\calV$ contains some $v'\notin\{0,v\}$, whereas $\Hyb_6$ returns early if and only if $\calD$ contains a database with an entry $(x', r)$ for $x\neq x'$. Neither condition can occur for the first query, since the internal state of $\calO_{f, H, A}$ is initialized to $\ket{0, \emptyset}$. Thus, the internal states of the two hybrids after query 0 are identical. Inducting over the number of queries, the invariant show in \Cref{claim:SEQ-proof-hyb5-invariant} applies to both $\Hyb_5$ and $\Hyb_6$ when query $t+1$ is submitted (but not yet answered). The invariant shows that at this point, $v'\notin \{0,v\}$ if and only if $\calD$ contains an entry $(x', r)$ for $x\neq x'$. Therefore $\Hyb_5$ will return early when answering query $t+1$ if and only if $\Hyb_6$ will.
\end{proof}

\begin{claim}\label{claim:SEQ-hyb6-hyb7}
    $\Hyb_6$ and $\Hyb_7$ are perfectly indistinguishable.
\end{claim}
\begin{proof}
    Observe that the new caching procedure in $\Hyb_7$ for steps \ref{proof:SEQ-hyb4-cache1} and \ref{proof:SEQ-hyb4-cache2} applies a CNOT operation from register $\calQ_v$ to register $\calV$ if and only if $H$'s database register contains an entry $(x', r)$ for $x' \neq x\oplus 1$. The single-effective-query check from step 2 ensures that $H$'s database register is in the span of $\ket{D}$ where $D = \emptyset$ or $D$ contains exactly one entry of the form $(x, r)$. Since $x\neq x\oplus 1$, the new caching procedure applies a CNOT if and only if $H$'s database register contains an entry of the form $x$. 
    
    On the other hand, steps \ref{proof:SEQ-hyb5-cache1} and \ref{proof:SEQ-hyb5-cache2} in $\Hyb_6$ apply the same CNOT operation if $H$'s database register is non-empty. The single-effective-query check in step 2 ensures that this occurs only when $H$ contains an entry of the form $x$. This is identical to the new caching procedure in $\Hyb_7$.
\end{proof}

\subsection{Which Functions is SEQ Access Meaningful For?}\label{sec:seq-meaningful}

Although it is possible to achieve SEQ simulation security for every function, not all functions are meaningfully restricted by SEQ access. For example, deterministic functions clearly can be fully learned with access to an SEQ oracle. 
In this section, we explore what properties imply that a function is unlearnable given SEQ access, culminating in a general criteria for achieving \Cref{def:weak-operational-security}.

Intuitively, a function must satisfy two loose properties in order to have any notion of unlearnability with SEQ access:
\begin{itemize}
    \item \textbf{High Min-Entropy.} If $f(x;r)$ is not sufficiently dependent on the randomness $r$, then measuring $f(x;r)$ may only gently measure $r$. In this case, the SEQ oracle will allow additional queries with some lower, but still inverse polynomial, amplitude.

    \item \textbf{Unforgeability.} If it is possible to compute some $f(x';r')$ given only $f(x;r)$, then the adversary could learn two function evaluations using one query.
\end{itemize}

We emphasize that any reasonable notion of unlearnability must be \emph{average-case} over the choice of $f$ from some family. Otherwise, an adversary could trivially learn everything about $f$ by receiving it as auxiliary input.

\paragraph{Truly Random Functions.}
As a concrete example, a truly random function exemplifies both of the above properties; it has maximal entropy on every input and $f(x;r)$ is completely independent of $f(x';r' )$. Indeed, we are able to show that any adversary with SEQ access to a truly random function cannot output two input/output pairs, except with negligible probability.

\begin{proposition}
\label{claim:random_function_unlearnable-body}
    Random functions with superpolynomial range size
   are \emph{single-effective-query} $\negl(\lambda)$-unlearnable \Cref{def:single-query-unlearnable}. More formally, for functions $\calF: \calX \times \calR \to \calY$, where $\vert\calR\vert = 2^\lambda$ and $\vert\calY\vert$ is superpolynomial in $\lambda$, and for all (non-uniform) quantum polynomial-time adversaries $\calA$, there exists a negligible function $\negl(\cdot)$ such that:
  \begin{align*}
        \Pr_{f \leftarrow \calF_n} [\calA^{f_{\$, 1}}(1^\lambda) \to (x_1,r_1,y_1 = f(x_1,r_1)), (x_2,r_2, y_2 = f(x_2,r_2))] \le \negl(\lambda).
    \end{align*}
$f$ is sampled uniformly at random from $\calF$ and $f_{\$, 1}$ is the single-effective-query oracle for $f$ defined in \Cref{sec:seq_oracle_definition}.
\end{proposition}

To prove this claim, we introduce the following technical lemma about compositions $f\circ H$ of random functions where $f$ and $H$ are implemented as compressed oracles. It shows that if $f$ records a query $x\concat y$, then $H$ must record a corresponding query $x$ where $f(x) = y$. 
Intuitively, this implies that $f$ can only record a single query, since that is the restriction on $H$. Any input/output pairs that the adversary learns will be recorded by $f$, so they can only learn a single one.
We prove the technical lemma in \Cref{sec:compressed-chaining}.\footnote{\Cref{sec:compressed-chaining} also contains a related technical lemma which shows that if an adversary has access to $H\circ G$ and $H$ records an entry $(y, z)$, then $G$ records an entry $(x, y)$ for some $x$. The difference from the lemma mentioned here is that $H$ does not take additional input beyond $G$'s output.}


\begin{lemma}[Compressed Oracle Chaining]\label{lemma:compressed-chaining-carry-body}
    Let $G:\calX_G \rightarrow \calY$ and $H:\calX_H \times \calY \rightarrow \calZ$ be random oracles implemented by the compressed oracle technique. Let $\calX\subset \calX_G \times \calX_H$. Define the function $F: \calX \rightarrow \calZ$ by $F(x_g, x_h) = H(x_g, G(x_h))$.
    Consider running an interaction of an oracle algorithm with $F$ until query $t$, then measuring the internal state of $G$ and $H$ to obtain $D_G$ and $D_H$. 

    Let $E_t$ be the event that after the measurement at time $t$, for all $(x_G\concat y, z)\in D_H$, there exists a entry $(x_H, y) \in D_G$.
    Then
    \[
        \Pr[E_t]
        \geq 
        1 - 4t^2\left(\frac{2}{|\calY|} - \frac{1}{|\calY|^2}\right)
    \]
    

\end{lemma}

Note that in the statement of this lemma, $F$ is defined on a subset of $\calX_G\times \calX_H$. Thus, there may be dependence between $x_G$ and $x_H$ in that space. As a special case, this also covers when $x_G = x_H$, i.e. $F(x) = H(x, G(x))$. However, the statement is more general for additional flexibility.



\begin{proof}[Proof of \Cref{claim:random_function_unlearnable-body}]
    We modify our view on the use of the function $f$ to be the compressed oracle defined in \Cref{sec:compressedRO}. Recall that in our implementation of the SEQ oracle $f_{\$,1}$ definition in \Cref{sec:seq_oracle_definition}, we already maintain a compressed oracle database for a random oracle $H: \calX \to \calR$ which computes a fresh randomness $r = H(x)$. Let us call this database for compressed oracle of $H$ as $D_H$. 
    Now we additionally have a database $D_f$ for the compressed oracle of the function $\calF: \calX \times \calR \to \calY$. 
    
    By \Cref{claim:single_entry_database}, $D_H$ is a state in superposition of basis states where each represents a classical database that records only one entry $(x, r)$, where $r =h(x)$ for some random $h \gets H$. These basis states are orthogonal to each other since each state records a different value $x$ by our implementation of the SEQ oracle.
    
    We take a further low-level view on our SEQ oracle implementation for a random function $f$: we can consider the adversary and the implementation of the SEQ oracle together as a "semi-honest" adversary $\calA'$ that queries the (compressed) random oracle $f$ if and only if there is no entry in the database $D_H$. Note that the (compressed) random oracle $f$ is a regular compressed oracle without any query restrictions. $\calA'$ receives $\calA$'s query and simulates the SEQ oracle for $\calA$ using its access to oracle $f$. In the end, $\calA'$ outputs $\calA$'s output. It is easy to observe that their advantage is the same since $\calA'$ perfectly simulates the SEQ oracle for $\calA$.
    
    Let us denote $p : = \calA^{f_{\$, 1}}(1^\lambda) \to (x_1,r_1,y_1 = f(x_1,r_1)), (x_2,r_2, y_2 = f(x_2,r_2))$. By the above observation, we can also denote $p$ as $\calA'$'s advantage of outputting $(x_1,r_1,y_1 = f(x_1,r_1)), (x_2,r_2, y_2 = f(x_2,r_2))$. 
    By \Cref{lem:zhandry_lemma5},  we have $\sqrt{p} \leq \sqrt{p'} + \sqrt{2/\vert \calY \vert}$, where $p'$ is the probability that $D_f(x_1,r_1) = y_1, D_f(x_2,r_2) = y_2$ after a computational basis measurement on $D_f$. 
    The compressed oracle chaining lemma (\Cref{lemma:compressed-chaining-carry-body}) shows that whenever $D_f(x_1,r_1) = y_1$ and $D_f(x_2,r_2) = y_2$, there are corresponding entries $D_H(x_1) = r_1$ and $D_H(x_2) = r_2$ in $H$'s compressed database, except with probability $4q^2\left(2/|\calR| - 1/|\calR|^2\right)$, where $q$ is the number of queries the adversary has made. Since $D_H$ contains at most one entry at a time (\Cref{claim:single_entry_database}) and $\calR$ has superpolynomial size, $p'$ must be negligible in $\secpar$. Since $\calY$ also has superpolynomial size, $p$ must be negligible as well.
\end{proof}

As an immediate corollary of \Cref{claim:random_function_unlearnable-body}, psuedorandom functions are SEQ-unlearnable under \Cref{def:single-query-unlearnable}.

\paragraph{Pairwise Independent Functions.}
We are also able to relax the requirement that $f$ is a truly random function to just require that it is pairwise independent and has high entropy. Intuitively, the pairwise independence plays the role of unforgeability by ensuring the adversary cannot use one evaluation $f(x;r)$ to learn anything about other evaluations $f(x';r')$. We prove the following statement in \Cref{sec:pairwise-indep}.

\begin{proposition}
    Let $\calF$ be a family of functions mapping $\cX \times \cR \to \cY$ that satisfies:
    \begin{enumerate}
        \item \textit{Pairwise independence:} For any $(x, r, y), (x', r', y') \in \cX \times \cR \times \cY$ such that $(x, r) \neq (x', r')$, $\Pr_f[f(x, r) = y \land f(x', r') = y'] = \Pr_f[f(x, r) = y] \cdot \Pr_f[f(x', r') = y']$.
        \item \textit{High Randomness:} There is a negligible function $\nu(\secp)$ such that for any $(x, r, y) \in \cX \times \cR \times \cY$, 
        \[\Pr_f[f(x, r) = y] \leq \nu(n)\]
        \item $\frac{1}{|\cR|} = \negl(\secp)$
    \end{enumerate}
    Then $\cF$ is SEQ-$\negl(\secp)$-unlearnable.
\end{proposition}

\paragraph{Computational Unforgeability.}
So far, we have considered functions which are unforgeable in a very strong, information-theoretic sense. We can also consider functions which satisfy a computational notion of unforgeability. It is important that this notion of unforgeability consider \emph{quantum} query access. We introduce the following generalization of quantum blind unforgeability for signatures~\cite{EC:AMRS20}.

\begin{definition}\label{def:gen-blind-unforge}
    Let $\calF = \{f:\calX \rightarrow \calY\}_f$ be a function family associated with a distribution $\Distr_{\calF}$ over function and auxiliary input pairs $(f, \aux_f)$ and let $P$ be a predicate on $f$, $\aux_f$, and a pair of strings $(x, s)$.

    $\calF$ and $\Distr_{\calF}$ are \emph{quantum blind unforgeable} with respect to $P$ if for every QTP adversary $\adv$ and blinding set $B\subset \calX\times \calR$,
    \[
        \Pr\left[
                x \in B
                \land 
                P(f, \aux_f, x, s) = \mathsf{Accept} 
            :
            \begin{array}{c}
                (f, \aux_f) \gets \Distr_{\calF}
                \\
                (x, s) \gets \adv^{f_B}(\aux_f)
            \end{array}
        \right]
    \]
    where $f_B$ denotes a (quantumly-accessible) oracle that takes as input $x$ then outputs $f(x)$ if $x\notin B$, and otherwise outputs $\bot$.
\end{definition}

We will show that if a randomized function $f$ is quantum blind unforgeable and uses the sampled randomness in a particular way, then it is hard to come up with two input/output pairs of $f$ when given SEQ access to it.

\begin{proposition}\label{claim:SEQ-unlearnable-blind-unforge}

    Let $\calF = \{f:\calX\times \calR \rightarrow \calY\}_{f}$ be a function family associated with distribution $\Distr$. If $|\calX|/|\calR|^2 = \negl(\secpar)$ and $(\calF, \Distr)$ is blind-unforgeable with respect to the predicate that outputs $\Accept$ on input $(f, (x, r), y)$ such that $f(x, r) = y$, then $(\calF, \Distr)$ is single-effective-query unlearnable under \Cref{def:weak-operational-security}.
\end{proposition}

Intuitively, because the blind unforgeability property is also being applied to $r$, it also ensures that $f$ must be highly randomized. For example, if $f$ ignored its randomness $r$, then it would be trivially forgeable simply by outputting the same $f(x;r)$ with two different $r$ and $r'$.

To prove this claim, we use another technical lemma which shows that adversaries who can find images of an expanding random function $G_e$ must know corresponding preimages. 
At a high level, we can then use this to prove \Cref{claim:SEQ-unlearnable-blind-unforge} as follows. We will first embed a random expanding function $G_e$ into the randomness generation for $r$. Blind unforgeability will imply that if the adversary knows $(x_1,r_1, y_1)$ and $(x_2, r_2, y_2)$ such that $f(x_1, r_1) = y_1$ and $f(x_2, r_2) = y_2$, then $r_1$ and $r_2$ are in the image of the expanding function $G_e$. Then the technical lemma will imply that if $G_e$ were implemented as a compressed oracle, its database should contain corresponding entries. 
Finally, the chaining lemma will allow us to show that the SEQ oracle should also have two corresponding entries, which is a contradiction.
We prove the following technical lemma in \Cref{sec:expanding-preimage-knowledge}.

\begin{lemma}\label{lemma:expanding-preimage-knowledge-body}
    Let $G: \calX_1 \times \calX_2 \rightarrow \calY$ be a random function where $|\calX_2| < |\calY|$. Consider an oracle algorithm $A$ makes $q$ of queries to $G$, then outputs two vector of $k$ values $\vec{x^{(1)}} = \left(x_1^{(1)}\dots, x_k^{(1)}\right)$ and $\vec{y} = (y_1, \dots, y_k)$.
    Let $p$ be the probability that for every $i$, there exists an $x_i^{(2)}\in \calX$ such that $G\left(x_i^{(1)}, x_i^{(2)}\right) = y_i$.

    Now consider running the same experiment where $G$ is instead implemented as a compressed oracle, and measuring its database register after $A$ outputs to obtain $D$. Let $p'$ be the probability that for every $i$, there exists an $x_i^{(2)}\in \calX_2$ such that $D\left(x_i^{(1)}\concat x_i^{(2)} \right) = y_i$. If $k$ and $q$ are $\poly(\secpar)$ and $|\calX_2|^k/|\calY| = \negl(\secpar)$, then\footnote{We remark that the reliance on the number of queries is unlikely to be tight. A tighter bound might be achieved by performing a direct computation of the effects of querying $G$ on every $x\in \calX$ at the end of the experiment.}
    \[
        p \leq p' + \negl(\secpar)
    \]

\end{lemma}

\begin{proof}[Proof of \Cref{claim:SEQ-unlearnable-blind-unforge}]
    We consider the following series of hybrid experiments.
    \begin{itemize}
        \item $\Hyb_0$ is the original SEQ unlearnability game. The adversary gets access to an SEQ oracle and attempts to output two valid input/output pairs. Recall that the SEQ oracle answers queries to $x$ by $f(x;H(x))$ for an internal random oracle $H$.
        
        \item $\Hyb_1$ modifies the SEQ oracle as follows. Sample random permutations $P_x:\calR\rightarrow \calR$. Instead of setting $r_x = H(x)$ and returning $f(x;r_x)$ on query $x$, set $r_x = P_x(H(x))$.
        
        \item $\Hyb_2$ modifies the SEQ oracle further by changing each $P_x$ to be a random function $G_x:\calX \times \calR \mapsto \calR$.
        
        \item $\Hyb_3$ modifies the SEQ oracle further by changing the random functions $G_x$ to be small range distributions. In more detail, let $\calR' \subset \calR$ be an arbitrary subspace with $|\calR'| = \sqrt{|\calR|}$. Let $G_{e,x}: \calR'\rightarrow \calR$ and $G_{s,x}: \calR\rightarrow \calR'$ be random functions, independently sampled for each $x\in \calX$. $G_{e,x}$ expands and $G_{s,x}$ shrinks. Then, replace each $G_x$ with $G_{e,x} \circ G_{s,x}$. 
        
        The aggregate result of this change can be expressed using two random functions $G_{e}:\calX\times \calR' \rightarrow \calR$ and $G_{s}:\calX\times \calR \rightarrow \calR'$. The SEQ oracle sets $r_x = G_e(x, G_s(x))$ when responding to a query on $x$.
        
        \item $\Hyb_4$ modifies the SEQ oracle further by implementing $G_e$ and $G_s$ as 
        compressed oracles.
    \end{itemize}
    To see that $\Hyb_0 = \Hyb_1$, fix any $x^*$ and consider adding just $P_{x^*}$. It is equivalent to answer queries to $x$ using $r_{x^*} = P_{x^*}\circ P_{x^*}^{-1}\circ H(x^*)$. Furthermore, the compressed database for $H$ may be modified to incorporate $P_x^{-1}$ using the isometry that maps any entry $(x^*,y)$ in the database to $(x, P_{x^*}^{-1}(y))$. The SEQ oracle only ever accesses the first element of any entry, which is not affected by the isometry. Since this is equivalent, it is also equivalent after doing the change for every possible $x^*$. $\Hyb_1$ and $\Hyb_2$ are indistinguishable to a poly-query adversary because random permutations cannot be distinguished from random functions in $q$ queries with advantage better than $O(q^3/|\calR|)$~\cite{QIC:Yuen14,QIC:Zha15}; doing this $|\calX|$ times results in advantage $O(q^3|\calX|/|\calR|) = \negl(\secpar)$. $\Hyb_2$ and $\Hyb_3$ are indistinguishable to a poly-query adversary because small-range distributions cannot be distinguished from a random function in $q$ queries with advantage better than $O(q^3/|\calR'|) $~\cite{zhandry2012construct}; doing this $|\calX|$ times results in advantage $O(q^3|\calX|/|\calR'|)=\negl(\secpar)$.

    We now show that the adversary's success probability in $\Hyb_4$ is negligible. First, in $\Hyb_3$ no adversary can output a tuple $((x, r), y)$ such that $f(x,r) = y$ and $r$ is not in the image of $G_e$, except with negligible probability.
    This follows from quantum blind unforgeability on the image of $G_e$ along with the observation that the SEQ oracle in $\Hyb_3$ only requires the ability to compute $f$ on in the image of $G_e$.


    Next, we claim that in $\Hyb_3$ if an adversary outputs two valid input/output tuples $((x_1, r_1), y_1)$ and $((x_2, r_2), y_2)$ where $x_1 \neq x_2$, then with overwhelming probability $r_1 \neq r_2$. Say the adversary did so with probability $p$. By the previous claim, whenever the adversary succeeds, there exist $r_1'$ and $r_2'$ such that $r_1 = G(x_1, r_1')$ and $r_2 = G(x_2,r_2')$, except with negligible probability. If $r_1=r_2$, then $G(x_1, r_1') = G(x_2, r_2')$. 
    Thus, we could find a collision in $G$ with probability $p/|\calR'|^2 - \negl(\secpar)$ by guessing $r_1'$ and $r_2'$. \Cref{lem:compressed-collision} shows that after $q$ queries, the probability of finding a collision is $O(q^3/|\calR|)$. 
    Since $|\calR'|^2/|\calR| = \negl(\secpar)$, we have $p\leq O(q^3|\calR^2|/|\calR|) = \negl(\secpar)$.

    \Cref{lemma:expanding-preimage-knowledge-body} shows that if we were to implement $G$ as a compressed oracle, then whenever the adversary finds two valid input/output tuples where $r_1 \neq r_2$, with overwhelming probability $G$'s compressed database contains entries $(x_1' \concat r_1', r_1)$ and $(x_2'\concat r_2', r_2)$. 
    By the chaining lemma (\Cref{lemma:compressed-chaining-carry-body}), whenever this occurs, $H$ also contains two entries, except with negligible probability. However, $H$ never contains more than one entry (\Cref{claim:single_entry_database}). Combining all of these facts together, any QPT adversary given SEQ access to $f_G\gets \Distr_G$ cannot output two distinct tuples $((x_1, r_1), y_1)$ and $((x_2, r_2), y_2)$ such that $f(x_1,r_1) = y_1$ and $f(x_2,r_2) = y_2$, except with negligible probability.
\end{proof}

In \Cref{sec:applications}, we use the techniques developed in this section to compile any signature scheme satisfying quantum blind unforgeability to enable signature tokens, almost without modifying the verification process.

\section{Construction in the Plain Model}
\label{sec:construction_plain_model}

In this section, we give a construction of a one-time sampling program in the plain model for constrained PRFs. Informally, a constrained PRF is a PRF with an additional algorithm $\mathsf{Constrain}(k, C)$ that takes as input the PRF key $k$ and a circuit $C$ and outputs a constrained key $k_C$, such that given $k_C$, it is easy to evaluate the value of the PRF on all inputs $x$ such that $C(x) = 1$. On the other hand, for all $x$ such that $C(x) = 0$, the output of the PRF should be pseudorandom. Constrained PRFs exist assuming LWE or subexponentially-secure iO and one-way functions. \anote{}

We prove (something stronger than)\footnote{In Section~\ref{sec:otp_prf_operational_def_generalized}, we give a pseudorandomness-based one-time security definition more suited to PRFs, and our plain model construction actually satisfies this stronger notion of security.} the following theorem:

\begin{theorem}

    Assuming the security of post-quantum indistinguishability obfuscation, and LWE (or alternatively, assuming sub-exponentially secure iO and \anote{injective?} OWFs), then there exists secure one-time sampling programs for constrained PRFs  $\prf:  \{0,1\}^{k} \times \{0,1\}^n \times \{0,1\}^{\ell} \to  \{0,1\}^m$  with respect to the weak operational security \Cref{def:weak-operational-security}.  Here we let $\lambda, k, m \in \N$ and $\ell \geq n\cdot \lambda$. $\{0,1\}^{k}$ is the key space for the PRF, $\{0,1\}^n$ the input space and $\{0,1\}^{\ell}$ the randomness space.
\end{theorem}

As ingredients in the plain model construction, we will need a subspace-hiding obfuscation scheme (which can be constructed from iO and injective one-way functions) and an extracting invertible PRF (which can be constructed from either LWE or subexponentially-secure iO and one-way functions).\anote{}

\subsection{Preliminaries}
\subsubsection{Indistinguishability Obfuscation}
\begin{definition}[Indistinguishability Obfuscator (iO)~\cite{barak2001possibility,garg2016candidate,sahai2014use}]
A uniform PPT machine $\iO$ is an indistinguishability obfuscator for a circuit class $\{\Cs_\lambda\}_{\lambda \in \mathbb N}$ if the following conditions are satisfied:
\begin{itemize}
    \item For all $\lambda$, all $C \in \Cs_\lambda$, all inputs $x$, we have 
    \begin{align*}
        \Pr\left[\widehat{C}(x) = C(x) \,|\, \widehat{C} \gets \iO(1^\lambda, C) \right] = 1
    \end{align*}
    
    \item (Post-quantum security): For all (not necessarily uniform) QPT adversaries $(\samp, D)$, the following holds: if $\Pr[\forall x, C_0(x) = C_1(x) \,:\, (C_0, C_1, \sigma) \gets \samp(1^\lambda)] > 1 - \alpha(\lambda)$ for some negligible function $\alpha$, then there exists a negligible function $\beta$ such that:
    \begin{align*}
        & \Bigg|\Pr\left[D(\sigma, \iO(1^\lambda, C_0)) =1\,:\, (C_0, C_1, \sigma) \gets \samp(1^\lambda)\right] \\ 
        - & \Pr\left[D(\sigma, \iO(1^\lambda, C_1)) = 1\,:\, (C_0, C_1, \sigma) \gets \samp(1^\lambda)\right] \Bigg| \leq \beta(\lambda)
    \end{align*}
\end{itemize}
\end{definition}

Whenever we assume the existence of $\iO$ in the rest of the paper, we refer to $\iO$ for the class of polynomial-size circuits, i.e. when $\mathcal{C}_{\lambda}$ is the collection of all circuits of size at most $\lambda$.



\subsubsection{Subspace Hiding Obfuscation}
\label{sec: shO}

Subspace-hiding obfuscation was introduced by Zhandry~\cite{zhandry2017quantum} as a key component in constructing public-key quantum money. This notion requires that the obfuscation of a circuit that computes membership in a subspace $A$ is indistinguishable from the obfuscation of a circuit that computes membership in a uniformly random superspace of $A$ (of dimension sufficiently far from the full dimension). The formal definition is as follows.
\begin{definition}[\cite{zhandry2017quantum}]
\label{def:sho_security}
A subspace hiding obfuscator (shO) for a field $\F$ and dimensions $d_0, d_1$ is a PPT algorithm $\shO$ such that:
\begin{itemize}
    \item \textbf{Input.} $\shO$ takes as input the description of a linear subspace $S \subseteq \F^n$ of dimension $d \in \{d_0, d_1\}$.
    
    For concreteness, we will assume $S$ is given as a matrix whose rows form a basis for $S$.
    \item \textbf{Output.} $\shO$ outputs a circuit $\hat{S}$ that computes membership in $S$. Precisely, let $S(x)$ be the function that decides membership in $S$. Then there exists a negligible function $\negl$,
    \begin{align*}
        \Pr[\hat{S}(x) = S(x)~~\forall x : \hat{S} \leftarrow \shO(S)] \geq 1 - \negl(n)
    \end{align*}

    \item \textbf{Security.} For security, consider the following game between an adversary and a challenger.
    \begin{itemize}
        \item The adversary submits to the challenger a subspace $S_0$ of dimension $d_0$.
        \item The challenger samples a uniformly random subspace $S_1 \subseteq \F^n$ of dimension $d_1$ such that $S_0 \subseteq S_1$.
        
        It then runs $\hat{S} \leftarrow \shO(S_b)$, and gives $\hat{S}$ to the adversary.
        \item The adversary makes a guess $b'$ for $b$.
    \end{itemize}
    $\shO$ is secure if all QPT adversaries have negligible advantage in this game.
\end{itemize}
\end{definition}

Zhandry \cite{zhandry2017quantum} gives a construction of a subspace hiding obfuscator based on one-way functions and $\iO$.
\begin{theorem}[Theorem 6.3 in \cite{zhandry2017quantum}]
If injective one-way functions exist, then any indistinguishability obfuscator, appropriately padded, is also a subspace hiding obfuscator for field $\F$ and dimensions $d_0, d_1$, as long as $|\F|^{n-d_1}$ is exponential.
\end{theorem}

\subsubsection{Subspace Coset States and Computational Direct-Product Hardness}
\label{sec:prelim_comp_direct_product_hardness}

In this section, we provide the computational version for the direct product hardness property given in \Cref{sec:prelim_IT_direct_product}.

In order to use the direct product hardness property assuming only the security of $\iO$ and injective OWFs, we use a variant of the subspace states \Cref{sec:subspace_state_prelims} called "subspace coset states" as follows, often referred to as "coset states" for short.

\begin{definition}[Subspace Coset States, \cite{coladangelo2021hidden,vidick2021classical}]
For any subspace $A \subseteq \mathbb{F}_2^n$ and vectors $s, s' \in \mathbb{F}_2^n$, the coset state $\ket {A_{s,s'}}$ is defined as:
\begin{align*}
    \ket {A_{s,s'}} = \frac{1}{\sqrt{|A|}} \sum_{a \in A} (-1)^{\langle s', a\rangle} \ket {a + s}\,.
\end{align*}
\end{definition}

Note that by applying $H^{\otimes n}$, which is {\sf QFT} for $\F_2^n$, to the state $\ket {A_{s,s'}}$, one obtains exactly $\ket {A^{\perp}_{s', s}}$. 

Additionally, note that given $\ket A$ and $s, s'$, one can efficiently construct $\ket {A_{s, s'}}$ as follows: 
\begin{align*}
     & \sum_a \ket a \,\xrightarrow[]{\text{add } s}\, \sum_a \ket {a + s} \,\xrightarrow[]{H^{\otimes n}}\, \sum_{a' \in A^\perp} (-1)^{\langle a', s\rangle} \ket {a'} \\
      \xrightarrow[]{\text{adding } s'}\, & \sum_{a' \in A^\perp} (-1)^{\langle a', s\rangle} \ket {a' + s'} \,\xrightarrow[]{H^{\otimes n}}\, \sum_{a \in A} (-1)^{\langle a, s'\rangle} \ket {a + s}
\end{align*}

For a subspace $A$ and vectors $s,s'$, we define $A+s = \{v +s : v \in A\}$, and $A^{\perp}+s' = \{v +s': v \in A^{\perp}\}$.

We then present the computational version of the direct product hardness property of \Cref{thm: direct product oracle}, by combining two theorems from \cite{coladangelo2021hidden} and  \cite{chevalier2023semi}.

We have defined $\shO$ above. For our construction, we will need the following variant of $\shO$.
\begin{definition}[Coset Subspace Obfuscation Programs]
\label{def:coset_sho}
We denote $\shO(A+s)$ for the following program: 
$\iO(\shO_A(\cdot-s))$, where $\shO_A()$ denotes the subspace-hiding program $\shO(A)$, and $\shO$ is the subspace hiding obfuscator defined in Section \ref{sec: shO}. Therefore, $\shO_A(\cdot-s)$ is the program that on input $x$, runs program $\shO(A)$ on input $x-s$. $\iO(\shO_A(\cdot-s))$ is an indistinguishability obfuscation of $\shO_A(\cdot-s)$.

\end{definition}

\begin{theorem}[Computational Direct Product Hardness, \cite{coladangelo2021hidden,chevalier2023semi}]
\label{thm: direct product comp}
{Assume the existence of post-quantum $\iO$ and injective one-way function.} 
Let $A \subseteq \mathbb{F}_2^n$ be a uniformly random subspace of dimension $n/2$, and $s, s'$ be uniformly random in $\mathbb{F}_2^n$. 
Given one copy of $\ket{A_{s,s'}}$, $\mathsf{shO}(A+s)$ and $\mathsf{shO}(A^{\perp}+s')$, any polynomial time adversary outputs a pair $(v,w)$ with only negligible probability such that either of the following is satisfied: (1) $v \in A+s$ and $w \in A^{\perp}+s'$; (2) $v, w \in A+s$ or $v, w \in A^{\perp}+s'$  and $v \neq w$.
\end{theorem}

\subsubsection{Puncturable, Constrained and Invertible PRFs}
\label{sec:constrained_prf}

\paragraph{Puncturable PRFs.}
A \emph{puncturable} PRF is a PRF augmented with a procedure that allows to ``puncture'' a PRF key $K$ at a set of points $S$, in such a way that the adversary with the punctured key can evaluate the PRF at all points except the points in $S$. Moreover, even given the punctured key, an adversary cannot distinguish between a uniformly random value and the evaluation of the PRF at a point in $S$ (with respect to the original unpunctured key). Formally:

\begin{definition}[(Post-quantum) Puncturable PRF]
\label{def:puncturable_prf}
A PRF family $F: \{0,1\}^{k(\lambda)} \times \{0,1\}^{n(\lambda)} \rightarrow \{0,1\}^{m(\lambda)}$ with key generation procedure $\kgen_F$ is said to be puncturable if there exists an algorithm $\puncture_F$, satisfying the following conditions:

\begin{itemize}
    \item \textbf{Functionality preserved under puncturing:} 
Let $S \subseteq \{0,1\}^{n(\lambda)}$. For all $x \in \{0, 1\}^{n(\lambda)}$ where $x \notin S$, we have that:
\begin{align*}
    \Pr[F(K,x) = F(K_S, x): K \gets \kgen(1^\lambda), K_S \gets \puncture_F(K, S)] = 1
\end{align*}

\item \textbf{Pseudorandom at punctured points:}
For every $QPT$ adversary $(A_1, A_2)$, there exists a negligible function $\negl$ such that the following holds. Consider an experiment where $K \gets \kgen_F(1^\lambda)$, $(S,\sigma) \gets A_1(1^{\lambda})$, and $K_S \gets \puncture_F(K, S)$. Then, for all $x \in S$,
\begin{align*}\left|\Pr[A_2(\sigma, K_S, S, F(K, x)) = 1] - \Pr_{r \gets \{0,1\}^{m(\lambda)}}[A_2(\sigma, K_S, S, r) = 1] \right|  \leq \negl(\lambda)
\end{align*}

\end{itemize}
\end{definition}

\paragraph{Constrained PRFs.}
A PRF $F : \mathcal{K} \times \mathcal{X} \to \mathcal{Y}$ 
has an additional key space $\mathcal{K}_\constrain$ and two additional algorithms $F.\constrain$ and $F.\constraineval$ as follows.
A constrained key $K_C$ with respect to a circuit $C$  enables the evaluation of $F(K, x)$ for all $x$ such that $C(x) = 1$ and no other $x$.

\begin{description}

\item $\keygen(1^\lambda, 1^n
) \to K$. On input the security parameter $\lambda$,
outputs the master secret key $K$.

\item $\eval(K, x) \to y \in \mathcal{Y}:$ On master secret key $K$ and value $x \in \mathcal{X}$,
outputs the evaluation $y \in \mathcal{Y}$.

\item $\constrain(K, C) \to \sk_C$:  takes as input a PRF key
$K \in \mathcal{K}$ and the description of a circuit $C$ (so that domain of $C$ $\subseteq \mathcal{X}$ ); outputs a
constrained key $K_C $. 

\item $\constraineval(K_C, x) \to y / \bot$:  On input a secret key $K_C$, and an input $x \in \{0, 1\}^n$, the constrained evaluation algorithm $\prf.\constraineval$ outputs an element $y \in \{0, 1\}^m$.
\end{description}

\begin{definition}[Constrained PRF Correctness]
A constrained PRF is correct for a circuit class $\mathcal{C}$ if $\msk \gets \prf.\keygen(1^\lambda)$, for
every circuit $C \in \mathcal{C}$ and input $x \in \{0, 1\}^n$
such that $C(x) = 1$, it is the case that:
\begin{align*}  F.\constraineval(F.\constrain(\msk,C), x) = F.\eval(\msk, x).
\end{align*}
    
\end{definition}

\begin{definition}[Adaptive single-key constrained pseudorandomness]
\label{def:constrain_prf_security}[Constrained PRF Security]
  We say that a $\cprf$ satisfies adaptive single-key constrained pseudorandomness security if for any stateful admissible QPT adversary $\adv$ there exists a negligible function $\negl(\cdot)$, such that for all $\lambda \in \N$, the following holds:
  \[
    \Pr\left[\adv^{\Eval(\msk, \cdot), \constrain(\msk, \cdot)}(r_b) = b:
      \begin{array}{cl}
        \msk \gets \keygen(1^{\lambda}),\ b \gets \bit \\
        x \gets \adv^{\Eval(\msk, \cdot), \constrain(\msk, \cdot)}(1^\lambda) \\
        r_0 \gets \{0,1\}^m,\ r_1 \gets \Eval(\msk, x)
      \end{array}
    \right] \leq \frac{1}{2} + \negl(\lambda).
  \]
  Here the adversary $\adv$ is said to be admissible as long as it satisfies the following conditions --- (1) it makes at most one query to the constrain oracle $\constrain(\msk, \cdot)$, and its queried circuit $C$ must be such that $C(x) = 0$, (2) it must send $x$ as one of its evaluation queries to $\Eval(\msk, \cdot)$.
\end{definition}

We only need single-key security of the above definition for our use.

\begin{remark}[Double-challenge security]
\label{remark:double_challenge_cprf}
   We will make a simple remark here on the following variant of the above security game: the adversary  submits two arbitrarily chosen $x_1, x_2$; the challenger chooses $r_{1,0}, r_{1,1},$ and $r_{2,0}, r_{2,1}$ independently as in the above security game. $\calA$ receives both $r_{1,b_1}$ and $r_{2,b_2}$ and has to guess both $b_1, b_2$ correctly. The winning probability of any $\adv$ in this "double-challenge" version of the game is upper bounded by the probability of it winning the single challenge game. We will make use of this fact later.
\end{remark}

The type of constrained PRF we use in this work can be built from standard lattice assumptions (\cite{brakerski2015constrained}) or alternatively from subexponentially-secure iO and OWFs \cite{boneh2017constraining}. 

\paragraph{Invertible PRFs}
\label{sec:invertible_prf_def}
An invertible pseudorandom function (IPF) is an injective PRF whose inverse function can be computed efficiently (given the secret key).

Therefore, it has the following additional algorithm apart from the PRF $\keygen$ and $\eval$:

\begin{description}
    \item  $\invert(\sk, y) \to  x$: on key $\sk$ and value $y \in \mathcal{Y}$, output a value $x \in \mathcal{X} \cup \{\bot\}$.
\end{description}


\begin{definition}[Correctness for Injective IPR]
\label{def:ipf_correctness}
A invertible PRF is correct if for all $\msk \gets \prf.\keygen(1^\lambda)$, it is the case that:
$$ F.\invert(\sk, F.\eval(\sk,x)) = x.
$$ 
and $F.\invert(\sk, y) = \bot$ where $y$ is not an image of any $x \in \calX$.
\end{definition}

In this paper, we only need the "regular" pseudorandomness property of the IPF and therefore we do not provide additional security definitions as in \cite{boneh2017constrained}.
We also do not need the IPF in our construction to be puncturable/constrainable.

In addition, we would like the IPF we use to act as extractors on their inputs:
\begin{definition}[Extracting PRF]
An extracting PRF with error $\epsilon(\cdot)$ for min-entropy $k(\cdot)$ is a 
PRF $F$ mapping $n(\lambda)$ bits to $\ell(\lambda)$ bits such that for all $\lambda$, if $X$ is any distribution over $n(\lambda)$
bits with min-entropy greater than $k(\lambda)$, then the statistical distance between $(\sk, F(\sk, X))$
and $(\sk, r \gets \{0,1\}^{\ell(\lambda)})$ is at most $\epsilon(\cdot)$, where $\sk \gets \kgen(1^\lambda)$. 
\end{definition}
In this work, we only need a weaker version of extracting property where the distribution does not give out $\sk$, but only a program that computes $F(\sk, \cdot)$. 
The constrained PRFs and extracting, invertible PRFs used in this work can all be obtained from LWE or from subexponentially-secure iO plus OWFs. \cite{sahai2014use,boneh2017constrained,brakerski2015constrained}. We will discuss in the appendix on a simple construction of an invertible and "extracting" PRF of our need based on the above works. \jiahui{Add later}

\subsubsection{One-Time Sampling Program for PRF: Indistinguishability Operational Definition}
\label{sec:otp_prf_operational_def_generalized}

In this section, we present the strengthened version of the operational one-time sampling program security for PRFs, which is a specific case of \Cref{def:strong-operational-security}.

\begin{definition}
[Indistinguishability Operational Security Definition for PRF One-Time Sampling Programs]
\label{def:generalized_operational_prf}
We define security through the following game:
\begin{enumerate}
    \item The challenger samples $\sk \gets \prf.\keygen(1^\lambda)$ and prepares the program $\OTP(\prf.\eval(\sk, \cdot))$. $\adv$ gets a copy of the one-time program for $\OTP(\prf.\eval(\sk, \cdot))$. 

    \item $\adv$ outputs two (input, randomness) pairs $(x_1, r_1), (x_2, r_2)$ , which satisfies $x_1 \neq x_2$ or $r_1 \neq r_2$, otherwise $\calA$ loses.

    \item Challenger samples two independent, uniform random bits $b_1 \gets \{0,1\}, b_2 \gets \{0,1\}$. 
    
    If $b_1 = 0$, then let $y_1 = \prf.\eval(\sk, x_1, r_1)$; else let $y_1 \gets \{0,1\}^m$.

     If $b_2 = 0$, then let $y_2 = \prf.\eval(\sk, x_2, r_2)$; else let $y_2 \gets \{0,1\}^m$.

    Challenger sends $(y_1, y_2)$ to $\adv$.
    
    \item $\adv$ outputs guesses $(b_1', b_2')$ for $(b_1, b_2)$ respectively. $\adv$ wins if and only if $b_1' = b_1$ and $b_2' = b_2$.
    \end{enumerate}
We say that a one-time sampling program for PRF satisfies indistinguishability operational security if for any QPT adversary $\adv$, there exists a negligible function $\negl(\cdot)$, such that for all $\lambda \in \N$, the following holds:
\begin{align*}
    \Pr\left[\adv \text{ wins the above game }\right] \leq \frac{1}{2} + \negl(\lambda).
\end{align*}
\end{definition}

\subsection{Construction and Security}
\label{sec:construction_io}

We first give the following building blocks for our construction. 

\begin{enumerate}
    \item A constrained PRF $   F_1: \{0,1\}^{k_1} \times \{0,1\}^{n+\ell}   \to \{0,1\}^m$, where $\{0,1\}^{k_1}$ is the key space and $\{0,1\}^{n+\ell}$ is the input space.

    Let $\sk_1 \gets F_1.\keygen(1^\lambda)$.
    
    \item An extracting invertible PRF $ F_2: \{0,1\}^{k_2} \times \{0,1\}^{n\cdot\lambda}  \to \{0,1\}^{\ell}$, where $\{0,1\}^{k_2}$ is the key space and $\{0,1\}^{\lambda}$ is the input space. $\ell \geq n\cdot \lambda$ and $\lambda$ is the security parameter. 
    
    The PRF is extracting with negligible error for inputs for min-entropy $n\cdot\lambda/2$. 

    Let $\sk_2 \gets F_2.\keygen(1^\lambda)$. 

    \jiahui{change to coset states}
    \item Sample $n$ random subspaces $A_1, \cdots, A_n$ independently from $\F^\lambda_2$, where each $\dim(A_i ) = \lambda/2$.

    Sample $2n$ random strings, $s_1, s_1' , \cdots s_n, s_n'$ each uniformly random from $\{0,1\}^\lambda$.

    Prepare the coset subspace-hiding obfuscation programs $\{ (\shO(A_1+s_1), \shO(A_1^\perp+s_1'), \cdots, (\shO(A_n+s_n), \shO(A_n^\perp+s_n'))\}$ as defined in \Cref{def:coset_sho}.

    For convenience, we will use the notation $\shO_{i}^0$ for $\shO(A_i+s_i)$, and $\shO_i^1$ for $\shO(A_i+s_i')^\perp$ for the rest of this section.
\end{enumerate}

The $\OTP(\cprf(\sk_1, \cdot))$ consists of the subspace states $(\ket{A_1}, \cdots, \ket{A_n})$ and an $\iO$ of the following program \Cref{fig:program_OTP_prf}:
\begin{figure}[hpt]
\centering
\begin{mdframed}[
  linecolor=black,
  leftmargin =4em,
  rightmargin=4em,
  usetwoside=false,
]

Hardcoded: $\sk_1, \sk_2, \{\shO_i^0,\shO_i^1\}_{i \in [n]}$.

On input $(x \in \{0,1\}^{n}, u = u_1 || u_2 || \cdots || u_n \in \{0,1\}^{n\cdot \lambda})$ (where each $u_i \in 
\mathbb{F}_2^\lambda$):

\begin{enumerate}
\item If for all $i \in [n]$, $\shO_i^{x_i}(u_i) = 1$, where $x_i$ is the $i$-th bit of $x$:

    \quad Let $r \gets F_2.\eval(\sk_2, u)$.
    
    \quad Output $(r, F_1.\eval(\sk_1, x || r))$.
\item Else:

    \quad Output $\bot$
\end{enumerate}

\end{mdframed}
\caption{Program $\cprf_{\OTP}$}
\label{fig:program_OTP_prf}
\end{figure}

\paragraph{Correctness}
The correctness follows from the extracting property of the PRF $F_2$: in any honest evaluation with $u\in\{0,1\}^{\lambda\cdot n}$ satisfies $u_i \in \shO_{A_i}^{x_i}$, for each $i \in [n]$.
Therefore $u$ has min-entropy $n \cdot \lambda/2$. By the extracting property of $F_2$ and the evaluation correctness of $F_1$, the above scheme satisfies correctness \Cref{def:otp_correctness}.

\subsubsection{Security Proof}
To prove security, we consider the following hybrids:

\paragraph{$H_0$:} 
In this hybrid, the challenger plays the original game defined in \Cref{def:generalized_operational_prf} using the above construction:
\begin{enumerate}
    \item The challenger prepares the program $\OTP(F_1.\eval(\sk_1, \cdot))$ as in \Cref{sec:construction_io}. $\adv$ gets a copy of the one-time program for $\OTP(F_1.\eval(\sk_1, \cdot))$. 

    \item $\adv$ outputs two (input, randomness) pairs $(x_1, r_1), (x_2, r_2)$ such that $x_1 \neq x_2$ or $r_1 \neq r_2$.

    \item Challenger samples two independent, uniform random bits $b_1 \gets \{0,1\}, b_2 \gets \{0,1\}$. 
    
    If $b_1 = 0$, then let $y_1 = F_1.\eval(\sk_1, x_1, r_1)$; else let $y_1 \gets \{0,1\}^m$.

     If $b_2 = 0$, then let $y_2 = F_1.\eval(\sk_1, x_2, r_2)$; else let $y_2 \gets \{0,1\}^m$.

    Challenger sends $(y_1, y_2)$ to $\adv$.
    
    \item $\adv$ outputs guesses $(b_1', b_2')$ for $(b_1, b_2)$ respectively. $\adv$ wins if and only if $b_1' = b_1$ and $b_2' = b_2$.
    \end{enumerate}

\paragraph{$H_1$:}
In this hybrid, the challenger modifies the original $\OTP$ program in \Cref{fig:program_OTP_prf} and the game as follows:

\begin{enumerate}
    \item The challenger prepares the program $\OTP(F_1.\eval(\sk_1, \cdot))$ as follows:
    \begin{enumerate}
        \item Sample $\msk \gets F_1.\keygen(1^\lambda), \sk_2 \gets F_2.\keygen(1^\lambda)$. Sample the subspaces and prepares the $\{\shO_{A_i}^b\}_{i \in [n], b \in \{0,1\}}$ programs. 

    \jiahui{add color to the differences}
        \item Compute constrained key $\sk_{A} \gets F_1.\constrain(\msk, C_{A})$ for the circuit $C_A$ in Figure~\ref{fig:circuit_C_A}, which takes as input $(x, r)$ and checks if the preimage of $r$ under $F_2$ are the correct coset vectors corresponding to $x$.

\begin{figure}[hpt]
\centering
\begin{mdframed}[
  linecolor=black,
  leftmargin =8em,
  rightmargin=8em,
usetwoside=false,
]

Hardcoded: $\sk_2,\{\shO_i^0,\shO_i^1\}_{i \in [n]}$.

On input $(x \in \{0,1\}^n, r \in \{0,1\}^\ell$):

\begin{enumerate}
\item Compute $(u = u_1 || u_2 || \cdots || u_n \in \{0,1\}^{n\cdot \lambda}) \gets F_2.\invert(\sk_2, r)$, where each $u_i \in 
\mathbb{F}_2^\lambda$. 

If the inversion result is $\bot$, output $0$.

\item If for all $i \in [n]$, $\shO_i^{x_i}(u_i) = 1$, where $x_i$ is the $i$-th bit of $x$:

    \quad Output 1.
\item Else:

    \quad Output $0$
\end{enumerate}

\end{mdframed}
\caption{Constraint Circuit $C_A$}
\label{fig:circuit_C_A}
\end{figure}
        
\item Prepare the $\iO$-ed classical program in the $\OTP(F_1(\cdot))$ scheme as in \Cref{fig:program_OTP_prf_hybrid1}, using the constrained key $\sk_A$ for $F_1$ : 

\begin{figure}[hpt]
\centering
\begin{mdframed}[
  linecolor=black,
  leftmargin =8em,
  rightmargin=8em,
  usetwoside=false,
]

Hardcoded: $\sk_A, \sk_2, C_A$.

On input $(x \in \{0,1\}^{n}, u = u_1 || u_2 || \cdots || u_n \in \{0,1\}^{n\cdot \lambda})$ (where each $u_i \in 
\mathbb{F}_2^\lambda$):

\begin{enumerate}
\item  Let $r \gets F_2.\eval(\sk_2, u)$.
    
 \item Output $(r, F_1.\constraineval(\sk_A, x || r))$.

\end{enumerate}

\end{mdframed}
\caption{Program $\cprf_{\OTP}$ in $H_1$}
\label{fig:program_OTP_prf_hybrid1}
\end{figure}

\end{enumerate}

    \item $\adv$ outputs two (input, randomness) pairs $(x_1, r_1), (x_2, r_2)$ such that $x_1 \neq x_2$ or $r_1 \neq r_2$.

    \item Challenger samples two independent, uniform random bits $b_1 \gets \{0,1\}, b_2 \gets \{0,1\}$. 
    
    If $b_1 = 0$, then let $y_1 = F_1.\eval(\sk_1 = \msk, x_1, r_1)$; else let $y_1 \gets \{0,1\}^m$.

     If $b_2 = 0$, then let $y_2 = F_1.\eval(\sk_1 = \msk, x_2, r_2)$; else let $y_2 \gets \{0,1\}^m$.

    Challenger sends $(y_1, y_2)$ to $\adv$.
    
    \item $\adv$ outputs guesses $(b_1', b_2')$ for $(b_1, b_2)$ respectively. $\adv$ wins if and only if $b_1' = b_1$ and $b_2' = b_2$.
    \end{enumerate}

\paragraph{$H_2$:}
In this hybrid, all steps are the same except in step 2, the challenger additionally checks if $F_2.\invert(\sk_2, r_1)$ and $F_2.\invert(\sk_2, r_2)$ are in the subspaces with respect to $x_1, x_2$: if so, abort the game and $\adv$ loses.

\begin{enumerate}
    \item The challenger prepares the program as in Hybrid 1. 

    \item $\adv$ outputs two (input, randomness) pairs $(x_1, r_1), (x_2, r_2)$ such that $x_1 \neq x_2$ or $r_1 \neq r_2$.

    {\color{red}Challenger makes the following check on $(x_1, r_1), (x_2, r_2)$:} 
{\color{red}
 \begin{enumerate}
     \item  Compute $(u_1 = u_{1,1} || u_{1,2} || \cdots || u_{1,n }\in \{0,1\}^{n\cdot \lambda}) \gets F_2.\invert(\sk_2, r_1)$, where each $u_{1,i} \in 
\mathbb{F}_2^\lambda$.
Similarly compute $u_2 \gets F_2.\invert(\sk_2, r_2)$.

If one of the inversion results is $\bot$, continue to step 3. 

\item If for both $j \in \{1,2\}$, for all $i \in [n]$, $\shO_i^{x_{j,i}}(u_{j,i}) = 1$, where $x_{j,i}$ is the $i$-th bit of $x_j$: \emph{abort and output 0.}

\item Else if there exists one of $j \in \{1,2\}$, $i \in [n]$ such that 
$\shO_i^{x_{j,i}}(u_{j,i}) = 0$, then continue to step 3.
 \end{enumerate}
}

    \item Challenger samples two independent, uniform random bits $b_1 \gets \{0,1\}, b_2 \gets \{0,1\}$. 
    
    If $b_1 = 0$, then let $y_1 = F_1.\eval(\sk_1, x_1, r_1)$; else let $y_1 \gets \{0,1\}^m$.

     If $b_2 = 0$, then let $y_2 = F_1.\eval(\sk_1, x_2, r_2)$; else let $y_2 \gets \{0,1\}^m$.

    Challenger sends $(y_1, y_2)$ to $\adv$.
    
    \item $\adv$ outputs guesses $(b_1', b_2')$ for $(b_1, b_2)$ respectively. $\adv$ wins if and only if $b_1' = b_1$ and $b_2' = b_2$.
    \end{enumerate}

\begin{claim}
    \label{claim:h0_h1}
    Assuming the post-quantum security  of $\iO$, The difference between $\adv$'s advantage in $H_0$ and $H_1$ is negligible.
\end{claim}

\begin{proof}
    The program $\cprf_\OTP$ in $H_0$ and $H_1$ have the same functionality: we only change the time of when we check the input vectors $u$ are in the corresponding subspaces. Therefore, by the security of $\iO$, the above claim holds. 
\end{proof}

\begin{claim}
\label{claim:h1_h2}
    By the security  of subspace-hiding obfuscation (\Cref{def:sho_security}), the difference between $\adv$'s advantage in $H_1$ and $H_2$ is negligible.
\end{claim}

\begin{proof}
    We invoke the computational direct product hardness \Cref{thm: direct product comp}: when giving a subspace state, and subspace-hiding obfuscations for the corresponding primal and dual subspaces, it is hard to produce two different vectors in the subspaces. Therefore, the event that challenger aborts on the event in Step 2 of Hybrid 1 is negligible.

    Otherwise, if all the preimages of $r_1, r_2$ are valid subspace vectors, since we require $x_1 \neq x_2$ or $r_1 \neq r_2$, there must exist at least an index $i^* \in [n]$ such that $u_{1,i^*} \neq u_{2,i^*}$ . Therefore, we can build a reduction to break the computational direct product hardness property: the reduction can sample its own $\msk, \sk_2$ and constrain the key $\msk$ on the circuit $C_A$ since it is given the programs $\shO_{A_i}$'s. When receiving the adversary's output of $r_1, r_2$, it can invert them to find the vectors $u_{1,i^*}, u_{2,i^*}$ that help it break \Cref{thm: direct product comp}.
\end{proof}

\begin{claim}    \label{claim:hyb2_negligible}
    Assuming adaptive single-key constrained pseudorandomness of $\cprf F_1$,  then $\calA$'s advantage in $H_2$ is negligible.
\end{claim}

\begin{proof}
    If there exists an $\adv$ that wins the game in $H_2$ with probability non-negligibly larger than $1/2$, then we can build a reduction $\calB$ to break the adaptive single-key constrained pseudorandomness in \Cref{def:constrain_prf_security}.

     Note that in this game, we have ruled out all $\adv$ that outputs $(x_1, r_1), (x_2, r_2)$ that satisfies $C_A(x_1, r_1) = C_A(x_2, r_2) =1$. That is, at least one of the above evaluations is 0. Therefore, the reduction, which makes a single key query on the circuit $C_A$, can use this input as the challenge input to the \Cref{def:constrain_prf_security} security game. If both inputs  satisfy that $C_A(x_1, r_1) = C_A(x_2, r_2) =0$, then we use the variant game in \Cref{remark:double_challenge_cprf}. In either case, the security of the constrained PRF guarantees that
     $\calA$'s winning probability in the game of $H_2$ is $1/2+\negl(\lambda)$.
\end{proof}

\section{Impossibility Results in the Plain Model and the Oracle Model}
\label{sec:impossibility}

In this section, we provide two different infeasibility results for one-time sampling programs that complement our positive results from two perspectives:
\begin{enumerate}
    \item Assuming LWE, there exists a family of one-query unlearnable, high min-entropy output functions where there are no insecure OTP for it \emph{in the plain model}, even with respect to the weak operational definition \Cref{def:weak-operational-security}.

    \item There exists a family of single \emph{physical} query unlearnable, high average min-entropy output, but partially deterministic functions where there are no insecure OTP for it, even in the oracle model, with respect to the weakest operational definition \Cref{def:weak-operational-security}.
\end{enumerate}

The first result in the non-black-box model is inspired by the non-black-box impossibility result of quantum obfuscation and copy protection in \cite{alagic2020impossibility,ananth2020secure}. However, the circuit family they use is a deterministic one. In order to show a nontrivial result for the one-time program, we design a family of randomized circuits which have almost full entropy output, but can nevertheless be "learned" through a single non-black-box evaluation.

\subsection{Preliminaries}

\subsubsection{Quantum Fully Homomorphic Encryption Scheme}
\label{sec:prelim_qfhe}

We give the definition of the type of QFHE we need for the construction in this section.

\begin{definition}
 [Quantum Fully Homomorphic Encryption]
 \label{def:qfhe}
 Let $\mathcal{M}$ be the Hilbert space associated with the message space (plaintexts), $\calC$ be
the Hilbert space associated with the ciphertexts, and $\calR_\ek$ be the Hilbert space associated with the
evaluation key. A quantum fully homomorphic encryption scheme is a tuple of QPT algorithms
$\qhe = (\keygen, \Enc, \Dec, \Eval)$ satisfying:
\begin{description}
    \item $\keygen(1^\lambda):$ a classical probabilistic algorithm that outputs a public key, a secret key as well as an evaluation key, $(\pk,\sk, \ek)$. 

\item $\Enc(\pk, \rho_\calM):$ takes as input a state $\rho_\calM$ in the space $L(\calM)$ and outputs a ciphertext $\sigma$ in $L(\calC)$.

\item $\Dec(\sk, \sigma):$ takes a quantum ciphertext $\sigma$, and outputs a state $\rho_\calM$ in the
message space 
$L(\calM)$.

\item $\Eval(\ek, U, \sigma_1, \cdots, \sigma_k)$ takes input of a quantum citcuit $U$ with $k$-qubits input and $k'$-qubits of outputs.
 Its output is a sequence of $k'$
quantum ciphertexts.
\end{description}
\end{definition}

The semantic security is analogous to the classical semantic security of FHE. We refer to \cite{broadbent2015quantum}.

\paragraph{Classical Ciphertexts for Classical Plaintexts} 
For the impossibility result, we require a QFHE scheme where ciphertexts
of classical plaintexts are also classical. Given any $x \in \{0, 1\}$, we want $\qhe.\Enc(\pk, \ket{x}\bra{x}$ to be
a computational basis state $\ket{z}\bra{z}$ for some $z \in \{0, 1\}^l$
(here, $l$ is the length of ciphertexts for 1-bit
messages). In this case, we write $\qhe.\Enc(pk,x)$. We also want the same to be true for evaluated
ciphertexts: if $U\ket{x}\bra{x} = \ket{y}\bra{y}$ for some basis state $x \in \{0,1\}^n, y \in \{0,1\}^l$, then we have:  $\qhe.\eval(\ek, U, \qhe.\Enc(\pk, \ket{x}\bra{x})) \to \qhe.\Enc(\pk, \ket{y}\bra{y})$ where the result is a classical ciphertext.

The QFHE schemes in \cite{brakerski2018quantum,mahadev2020classical} satisfy the above requirement. 

Note that we also need to evaluate on a possibly arbitrary polynomial depth circuit. The QFHE schemes in \cite{brakerski2018quantum,mahadev2020classical}  still require circular security to go beyond leveled FHE.

\subsubsection{Compute-and-Compare Obfuscation}
\label{sec:prelim_compute_compare}

\begin{definition}[Compute-and-Compare Program]
    Given a function $f:\{0,1\}^{\ell_{\sf in}} \to \{0,1\}^{\ell_{\sf out}}$ along with a target value $y \in \{0,1\}^{\ell_{\sf out}}$ and a message $z \in \{0,1\}^{\ell_{\sf msg}}$, we define the compute-and-compare program: 
    \begin{align*}
        \CC[f, y, z](x) = \begin{cases}
                            z & \text{ if } f(x) = y \\
                            \bot & \text{ otherwise }
                        \end{cases}
    \end{align*}
\end{definition}

We define the following class of \emph{unpredictable distributions} over pairs of the form $(\CC[f, y, z], \aux)$, where $\aux$ is auxiliary quantum information. These distributions are such that $y$ is computationally unpredictable given $f$ and $\aux$.

\begin{definition}[Unpredictable Distributions]
\label{def:cc_unpredictable_dist}
    We say that a family of distributions $D = \{D_\lambda\}$ where $D_{\lambda}$ is a distribution over pairs of the form $(\CC[f, y, z], \aux)$ where $\aux$ is a quantum state, belongs to the class of \emph{unpredictable distributions} if the following holds. There exists a negligible function $\negl$ such that, for all QPT algorithms $\adv$, 
    \begin{align*}
        \Pr_{ (\CC[f, y, z], \aux) \gets D_\lambda } \left[ A(1^\lambda, f, \aux) = y \right] \leq \negl(\lambda). 
    \end{align*}
\end{definition}

We assume that a program $P$ has an associated set of parameters $P.{\sf param}$ (e.g input size, output size, circuit size, etc.), which we are not required to hide.
\begin{definition}[Compute-and-Compare Obfuscation]
\label{def: cc obf}
    A PPT algorithm $\ccobf$ is an obfuscator for the class of unpredictable distributions (or sub-exponentially unpredictable distributions) if for any family of distributions $D = \{ D_{\lambda}\}$ belonging to the class, the following holds:
    \begin{itemize}
        \item Functionality Preserving: there exists a negligible function $\negl$ such that for all $\lambda$, every program $P$ in the support of $D_\lambda$, 
        \begin{align*}
            \Pr[\forall x,\,  \widetilde{P}(x) = P(x),\, \widetilde{P} \gets \ccobf(1^\lambda, P) ] \geq 1 - \negl(\lambda)
        \end{align*}
        \item Distributional Indistinguishability: there exists an efficient simulator $\Sim$ such that:
        \begin{align*}
            (\ccobf(1^\lambda, P), \aux) \approx_c (\Sim(1^\lambda, P.{\sf param}), \aux)
        \end{align*}
        where $(P, \aux) \gets D_\lambda$.
    \end{itemize}
\end{definition}

Combining the results of \cite{wichs2017obfuscating, goyal2017lockable} with those of \cite{C:Zhandry16}, we have the following two theorems. 
\justin{The Zhandry citation previously said zhandry2019magic. I assume this refers to "The Magic of ELFs", which I added as the citation.}
\jiahui{Yes thanks
!}

\begin{theorem}
\label{thm: cc from lwe}
Assuming the existence of the
quantum hardness of LWE, there exist obfuscators for unpredictable distributions, as in \Cref{def: cc obf}.
\end{theorem}

\subsubsection{Quantum Query Lower Bounds}

In the analysis of the unlearnability of circuits (\Cref{sec:impossibility}), we will use the following theorem from \cite{BBBV97} to bound the change in adversary's state when we change the oracle's input-output behavior at places where the adversary hardly ever queries on. 

    Let $\ket{\phi_i}$ be the state of the adversary after the $i$-th query to the oracle $\cO$, i.e. $\ket{\phi_i} = U_i \cO \cdots \cO U_2 \cO U_1 \ket{\phi_0}$, where $\ket{\phi_0}$ is the initial adversary state.

\begin{theorem}[\cite{BBBV97}] 
\label{thm:bbbv97_oraclechange}

We have defined the quantum query model for classical oracles in \Cref{sec:prelim_quantum_query_classical_oracle}. We give some further preliminaries here.

Let $\ket{\phi_i}$ be the superposition of oracle quantum algorithms $\mathcal{M}$ with oracle $\cO$ on input $x$ at time $i$. Define $W_y(\ket{\phi_i})$ to be the sum of squared magnitudes in $\ket{\phi_i}$ of configurations of $\mathcal{M}$ which are querying the oracle on string $y$. For $\epsilon  > 0$, let $F \subseteq [0, T-1] \times \Sigma^*$ be the set of time-string pairs such that 
$\sum_{(i,y) \in F} W_y(\ket{\phi_i}) \leq \epsilon^2/T$.

Now suppose the answer to each query $(i, y) \in F$ is modified to some arbitrary fixed $a_{i,y}$ (these answers need not be consistent with an oracle). Let $\ket{\phi_i'}$ be the superposition of $\mathcal{M}$ on input $x$ at time $i$ with oracle $\cO$ modified as stated above. Then $\left\| \ket{\phi_T} - \ket{\phi_T'} \right\|_{\mathrm{tr}}\leq \epsilon$.
\end{theorem}

\subsection{Impossibility Result for Single-Query Security in the Plain Model: for fully randomized functions}

In this section, we present a lower bound/impossibility result for a generic one-time program in the plain model (i.e. without using black-box oracles). The result states that there is no way to construct a generic one-time sampling program for \emph{all randomized functionalities} if the programs allow \emph{non-black-box} access.

Our result is inspired by the non-black-box impossibility result of quantum obfuscation and copy protection in \cite{alagic2020impossibility,ananth2020secure}. However, the circuit family they use is a deterministic one. In order to show a nontrivial result for the one-time program, we design a family of randomized circuits which have almost full entropy output, but can nevertheless be "learned" through a single non-black-box evaluation.

\begin{theorem}
\label{thm:plain_model_impossibility}

 Assuming the post-quantum security of LWE and QFHE, there exists a family of randomized circuits which are single-query $\negl(\lambda)$-unlearnable \Cref{def:single-query-unlearnable}, but
 not one-time program secure with respect to the the weak operational one-time security \Cref{def:weak-operational-security}.
\end{theorem}

We construct the following circuit which has high-entropy outputs and is 1-query unlearnable with only (quantum) single-query access to the function and access to a piece of classical information, but once put into any OTP in the plain model, is insecure.

First we give a few building blocks for the following circuit $C$.
\begin{enumerate}

\item Let $n = n(\lambda),  \ell = \ell(\lambda)$ be polynomials in the security parameter $\lambda$.

    \item Let $\SKE = (\SKE.\keygen, \SKE.\Enc, \SKE.\Dec)$ be any secret key encryption scheme. 

    The $\SKE$ scheme only needs to satisfy relatively weak security notion to be single-query unlearnable. For the sake of convenience, we use the textbook construction of post-quantum IND-CPA secure SKE from PRFs (\cite{zhandry2012quantumprf}):

    \begin{description}
           \item $\SKE.\keygen(1^\lambda): \sk \gets \prf.\keygen(\lambda)$

    \item  $\SKE.\Enc(\sk, m \in \{0,1\}^n) \to \ct:$ samples $r \gets \{0,1\}^\ell$; output $\ct \gets (r, \prf.\eval(\sk, r) \oplus m)$.

    \item $\SKE.\Dec(\sk, \ct)$: parse $\ct := (r, \ct')$; compute $m  := \prf.\eval(\sk, r) \oplus \ct'$.
    \end{description}

    Let $\SKE.\sk \gets \SKE.\keygen(1^\lambda)$ be the key we use in the following circuit construction.

\item Let $\qhe = (\qhe.\keygen, \qhe.\Enc, \qhe.\Dec, \qhe.\Eval)$ be a quantum fully homomorphic encryption scheme. Let the keys be $(\qhe.\pk, \qhe.\sk) \gets \qhe.\keygen(1^\lambda)$. Without loss of generality, we consider the evaluation key to be part of $\qhe.\pk$.

\item Let $\CC = \CC.\obf$ be a compute-and-compare obfuscation scheme in \Cref{sec:prelim_compute_compare}.

\item Let $a \gets \{0,1\}^n,b\gets \{0,1\}^n$ be two uniformly random strings.
\end{enumerate}
 
We design the following with auxiliary information $\aux = (\ct_a = \qhe.\Enc(\qhe.\pk, a); \tildeP = \\ \CC.\obf(\SKE.\Dec(\SKE.\sk,\qhe.\Dec(\qhe.\sk,\cdot)), b, (\SKE.\sk, \qhe.\sk)), \qhe.\pk)$.

\begin{mdframed}
\begin{description}

    \item Input: $(x \in \{0,1\}^n, r \in \{0,1\}^\ell)$

    \item Hardcoded:  $(a, b, \Enc.\sk)$

    \item if $x = a$:
    \begin{description}
    \item   output $\SKE.\Enc(\Enc.\sk, b; r)$
\end{description}

    \item else:
    \begin{description}
        \item output $\SKE.\Enc(\Enc.\sk, x; r)$. 
    \end{description}

\end{description}
\end{mdframed}

Note that the program $\tildeP$ in the auxiliary information $\aux$ is a compute-and-compare obfuscation program of the following circuit:

    \begin{align*}
      \CC[f, (\SKE.\sk, \qhe.\sk), b] = \begin{cases}
                            (\SKE.\sk, \qhe.\sk) & \text{ if }    f(x) = b \\
                            \bot & \text{ otherwise }
                        \end{cases}
    \end{align*}
where $f(x) = \SKE.\Dec(\SKE.\sk, \qhe.\Dec(\qhe.\sk, x))$.

\begin{claim}
\label{claim:plain_model_learnable} 
The above circuit with auxiliary information $(C, \aux)$  can be perfectly reconstructed by any QPT adversary given any one-time program with correctness of the above circuit.
\end{claim}
\begin{proof}
Given any OTP $\ket{\psi_C}$ of the circuit $C$ together with the auxiliary information $\aux$. A QPT adversary can perform the following attack:
\begin{enumerate}
    \item Encrypt the program: $\ct_{\ket{\psi_C}} \gets \qhe.\Enc(\qhe.\pk, \ket{\psi_C})$ using the QFHE public key $\qhe.\pk$ given in $\aux$.

    \item Homomorphically evaluate the program on the input $\ct_a = \qhe.\Enc(\qhe.\pk, a)$ from $\aux$, with respect to a universal quantum circuit $U$, to obtain an outcome $\ct_{\SKE.\Enc(b)}$:
    $$ \ct_{\SKE.\Enc(b)} :=  \qhe.\Enc(\qhe.\pk, \SKE.\Enc(\SKE.\sk, b; r_b)) \gets  \\
   \qhe.\Eval(\qhe.\pk, U, \ct_{\ket{\psi_C}}, \ct_a). $$
for some random $r_b$. 

The above evaluation holds due to the correctness of OTP scheme and the QFHE scheme: when one evaluates $U(\ket{\psi_C},a)$ honestly, then one obtains $\SKE.\Enc(\SKE.\sk, b; r_b)$ for some random classical string $r_b$.
Therefore, by the correctness of QFHE, we obtain the 
a QFHE ciphertext $\qhe.\Enc(\qhe.\pk, \SKE.\Enc(\SKE.\sk, b; r_b))$. This evaluation procedure is randomized and the original OTP state $\ket{\psi_C}$ gets destroyed during the procedure.

 Since the message $\SKE.\Enc(b;r_b)$ is classical, the ciphertext $\ct_{\SKE.\Enc(b)}$ under QFHE is classical by the property of the QFHE we use. 
 
    \item Evaluate the compute-and-compare obfuscation program $\tildeP$ on input $\ct_{\SKE.\Enc(b)}$. 

    Note that by the correctness of the compute-and-compare obfuscation program, the input $\ct_{\SKE.\Enc(b)}$ satisfies that $f(\ct_{\SKE.\Enc(b)}) = \SKE.\Dec(\SKE.\sk, \qhe.\Dec(\qhe.\sk, \ct_{\SKE.\Enc(b)}))  = b$.

    Therefore, one will obtain the information $(\SKE.\sk, \qhe.\sk)$.
    
    \item Now one can first decrypt the QFHE ciphertext $\ct_a = \qhe.\Enc(\qhe.\pk, a)$ to obtain $a$ and the doubly-encrypted ciphertext $\ct_{\SKE.\Enc(b)}$ to obtain $b$, with keys $\qhe.\sk, \SKE.\sk$.

    Given the above information, one can fully reconstruct the circuit $C$ together with all the auxiliary information in $\aux$ perfectly (note that the information in $\aux$ are classical and can be copied and kept in the first place).
\end{enumerate}

\end{proof}

\begin{remark}
    \label{remark:learnability_plain_model}

    Note that the above construction actually shows a stronger statement than an infeasibility result of OTP: it lets a QPT adversary recover the entire circuit perfectly, which obviously allows it to violate the OTP security.
    But if we only need the adversary to output two input-output pairs, storing $\SKE.\sk$ as the secret message in the compute-and-compare program suffices.
\end{remark}

\begin{claim}
 \label{claim:oracle_unlearnable_impossibility}   
Assuming the post-quantum security of LWE,the above circuit with auxiliary information $(C, \aux)$ satisfies single-query unlearnability  (\Cref{def:single-query-unlearnable}, for both physical and effective queries).  
\end{claim}

\begin{proof}
We will prove through a sequence of hybrids to show that the oracle is indistinguishable from a regular SKE functionality, which is single-query unlearnable. 

    The proof is similar to the proof for Claim 47 in  \cite{ananth2020secure}, with some modifications. We directly use the \cite{BBBV1997} argument instead of adversary method. 

        Let $\ket{\phi_i}$ be the state of the adversary after the $i$-th query to the oracle $\cO_C$ (quantum black-box access to the circuit $C$), i.e. $\ket{\phi_i} = U_i \cO_C \cdots \cO_C U_2 \cO_C U_1 \ket{\phi_0}$, where $\ket{\phi_0}$ is the initial adversary state.
    We first make the following claim:

    \begin{claim}
        Assuming the post-quantum security of LWE, the sum of squared amplitudes of query on strings starting with $a$: $\sum_{i}^T W_a(\ket{\phi_i}) \leq \negl(\lambda)$, where $T$ is the total number of steps.
    \end{claim}

    \begin{proof}
        
        We prove this by induction and the security properties of QFHE and $\ccobf$.

        \textbf{Base case}: before the adversary makes the first query, clearly $W_a(\ket{\phi_0})$ is negligible (in fact 0 here), we consider the following hybrids:

        \begin{enumerate}
            \item $H_0$: this is the original game where we give out auxiliary information $\aux = (\ct_a = \qhe.\Enc(\qhe.\pk, a); \\ \tildeP = \CC.\obf(\SKE.\Dec(\SKE.\sk,\qhe.\Dec(\qhe.\sk,\cdot)), b, (\SKE.\sk, \qhe.\sk)), \qhe.\pk)$.

            \item $H_1$: 
            reprogram the oracle $\cO_C$ to have the  functionality in \Cref{fig:hybrid1_oracle}.

        \begin{figure}[hpt]
    \centering
    \caption{$\cO_C'$ in $H_1$}
    \begin{mdframed}
\begin{description}

    \item Input: $(x \in \{0,1\}^n, r \in \{0,1\}^\ell)$

    \item Hardcoded:  $ \Enc.\sk$

        \item output $\SKE.\Enc(\Enc.\sk, x; r)$.

\end{description}
\end{mdframed}
\label{fig:hybrid1_oracle}
\end{figure}

            \item $H_2$: replace $\tildeP$ in the above $\aux$ with $\Sim(1^\lambda, P.{\sf param})$. 

            \item $H_3$: replace $\ct_a = \qhe.\Enc(\qhe.\pk, a)$ with $\ct_0 = \qhe.\Enc(\qhe.\pk, 0^n)$.
        \end{enumerate}

The adversary's state $\ket{\phi_0}$ should have negligible difference in terms of trace distance by \cite{BBBV97} (by plugging $T = 0$) in $H_0$ and $H_1$.

Let the projection $\Pi_a : =
(\ket{a}\bra{a}_{x} \otimes \mathbf{I}_{r,\adv} )$, where $x, r$ are the registers corresponding to the input $(x,r)$ to $\cO_C$ and $\adv$ represents the rest of registers in the adversary's state.

We can measure the adversary's first query by projecting the state $\ket{\phi_0}$ onto $\cO_c U_1 \Pi_a
 (\cO_c U_1)^\dagger$.
The adversary with state $\ket{\phi_0}$ should have negligible difference in query weight on $a$ for the first query in $H_1$ and $H_2$: since $b$ is sampled uniformly at random and for the adversarial state $U_1\ket{\phi_0}$, $b$ satisfies the unpredictable distribution property in \Cref{def:cc_unpredictable_dist}. 
Therefore, by the property of \Cref{def: cc obf}, any measurement in the adversary's behaviors in $H_1$ and $H_2$ should result in computationally indistinguishable outcomes. 

In more details, the reduction to the compute-and-compare security works as follows: since the oracle $\cO_C$ is now independent of $a,b$, the reduction can sample$\SKE, \qhe$ keys, prepare the oracle $\cO_C$ and sample its own $a$; it then receives the obfuscated compute-and-compare program $\tildeP$ (or the simulated program $\Sim(1^\lambda, 1^{|f|})$) from the challenger, where $b$ is uniformly random. If the measurement  $\cO_c U_1 \Pi_a
 (\cO_c U_1)^\dagger$ on adversary's first query  gives outcome 1, then output "real program", else output "simulated program".

 The query weights $H_2$ and $H_3$ should have negligible difference by the security of the QFHE. Since the program $\tildeP$ has been replaced with a simulated program, the $\qhe$ reduction can prepare the programs as well as the oracle $\cO_C$. It receives $\ct_a$ or $\ct_0$ from the challenger. If the measurement on the adversary's first query returns 1, then guess $\ct_a$, else guess $\ct_0$.

 The adversary's first query weight on $a$ in $H_3$ is negligible since now there is no information about $a$ anywhere in $\aux$ and $a$ is only a uniform random string in $\{0,1\}^n$. By the above arguments, the adversary's first query's weight on $a$ 
is negligible in the original game $H_0$.


        \textbf{Induction}: the above argument applies to the $k$-th query, if the sum of squared amplitudes over the first $(k-1)$ queries,  $\sum_{i}^{k-1} W_a(\ket{\phi_{i}})$, is negligible, then we can invoke the above arguments and show that $W_a(\ket{\phi_{k}})$ is negligible as well. 
    \end{proof}

Since we have shown that the total (squared) query weight on $a$, $\sum_i^T W_a(\ket{\phi_i})$ is negligible, we can replace the oracle $\cO_C$ for the entire game with the oracle $\cO_C'$ in the above $H_1$, i.e. \Cref{fig:hybrid1_oracle} and by \cite{BBBV97}, the trace distance between $\ket{\phi_T}$ using the original oracle $\cO_C$ and $\ket{\phi_T}$ using the oracle $\cO_C'$ in \Cref{fig:hybrid1_oracle} is negligible.  Now it remains to show that 
$\cO_C'$ together with $\aux$ is single-query unlearnable for any QPT adversary.

By similar hybrids as above, we can replace the information in $\aux$ with a dummy program and dummy ciphertext so that $\aux = (\ct_0 = \qhe.\Enc(\qhe.\pk, 0^n), \Sim(1^\lambda, 1^{|f|}), \qhe.\pk)$. The adversary's advantage in the unlearnability game should have negligible difference by similar hybrid arguments. 

Now recall that we instantiate $\SKE$ from $\prf$ using the textbook construction for IND-CPA SKE. We can then show that if there exists an adversary that violates the single-query unlearnability of a 
$\prf$ that maps $2\cdot |R|$-length inputs to $\vert m \vert$-length inputs 
(which is essentially a random oracle when accessed in the oracle model and thus it satisfies single-query unlearnability and high min-entropy outputs): the reduction can simulate the oracle $\cO_C'$ on query $x$ by querying a $\prf$ oracle on some random string $r_1$ of its own choice; the $\prf$ oracle will return $(r_2,\prf(\sk, r_1||r_2))$, where $r_2$ is the randomness chosen by the randomized $\prf$ oracle itself; the reduction then replies the adversary with $(r_1||r_2, \prf(\sk, r_1||r_2) \oplus m)$.
In the end, if the adversary wins by outputting two pairs $(m, r, \prf(\sk, r)\oplus m)$ and $(m', r', \prf(\sk, r')\oplus m')$, then the reduction outputs $(r, \prf(\sk, r))$ and $(r', \prf(\sk, r'))$ and would break the single-query unlearnability of the $\prf$ (see \Cref{remark:negl_unlearnable_example} and \Cref{sec:construction_plain_model}). 
\end{proof}

\subsection{Impossibility Result for Partially Randomized Functions in the Oracle Model}
\label{sec:impossibility_oracle_model}
In this section, we give an example of a function family that cannot be compiled into a one-time program under even the weakest definition of operational security, even in the classical oracle model. It is known that deterministic functions fall into this category, but this counterexample is not only randomized but also has high entropy in a very strong sense. It is however \textit{partially} deterministic: the (high) entropy is restricted to one half of the output and the other half is essentially deterministic, demonstrating that high entropy is not sufficient for a function to be one-time programmable.

Moreover, this example also demonstrates the following
\begin{enumerate}
    \item It separates a single-physical-query unlearnable function from single-effective-query unlearnable

    \item It is a single-physical-query unlearnable function that cannot be securely one-time programmed with respect to the classical-output simulation definition \Cref{def:single-query-classical-output-simulation-security}.
\end{enumerate}

Suppose $\mathsf{PRF}$ is a length-preserving pseudorandom function that is secure against adversaries who are allowed to make quantum superposition queries.\footnote{The GGM construction, for instance is secure against quantum superposition queries provided that the underlying PRG is quantum-secure~\cite{zhandry2012construct}.} 
Let $a$ be a uniformly random string in $\{0,1\}^n$. Let $k$ be a random PRF key.
Consider the function family $\calF_n = \{f_{a, k}\}_{a, k \in\{0,1\}^n}$, where $f_{a, k} : \{0,1\}^n \times \{0,1\}^n \rightarrow \{0,1\}^{2n}$ is defined as
\begin{align}\label{eqn:partially-deterministic}
    f_{a, k}(x; r) = 
    \begin{dcases}
        (a, \mathsf{PRF}_k(0 \| r))&\text{if } x=0,\\
        (k, \mathsf{PRF}_k(a \| r))&\text{if } x=a,\\
        (0, \mathsf{PRF}_k(x \| r))&\text{otherwise.}
    \end{dcases}
\end{align}
Also define an associated distribution $\calD$ over this function family such that $a \leftarrow \{0,1\}^n$ is chosen uniformly at random, and the key $k$ is sampled according to the PRF key-generation procedure.

First, we establish that this function family cannot be compiled into a one-time program even under the weak operational security definition.
\begin{lemma}\label{lem:partially-deterministic-weak-operational-impossibility}
$\calF$ cannot be compiled into a one-time program under the weak operational security definition.
\end{lemma}
\begin{proof}
We will construct an adversary $\calA$ that is able to entirely learn the function given its one-time program. First, the $\calA$ runs the one-time program evaluatoin procedure on input $x = 0$, and measures the first $n$ bits of the output to obtain value $a$. Since the first $n$ bits of $f_{a,k}(0, \cdot)$ will always equal $a$, by gentle measurement lemma, this measurement by the adversary cannot disturb the one-time program state. Now, having learned the value of $a$, the advesrary $\calA$ uncomputes its query on $0$ to restore the intial state of the one-time program. Finally, the adversary makes a second query to the one-time program evaluation procedure on input $a$, and measures the first $n$ bits of the output to get the $\mathsf{PRF}$ key $k$. This reveals the entire function description to the adversary. In particular, the adversary can break the weakest operational security definition by computing two input-output pairs $(x_1, f_{a, k}(x_1))$ and $(x_2, f_{a, k}(x_2))$.
\end{proof}

Now consider the following notion of min-entropy for randomized functions $f$
\begin{align*}
        H_{\min}(f) = \min_{x, y} \log \frac{1}{p(y| x)},
\end{align*}
where $p(y|x) = \Pr_{r \leftarrow \calR}[f(x,r) = y]$.

\begin{claim}[High entropy]
    The randomized function family $\calF$ has high entropy for every input.
\end{claim}
\begin{proof}
    
\end{proof}
If the PRF has output length $m$, this counterexample is insidtinguishable from a function $f^*$ that has high min-entropy for every input $x\in\{0,1\}^n$,
\begin{align*}
    \min_{x} H_{\min}(f^*(x, \cdot)) = m
\end{align*}

\begin{claim}[Single physical query unlearnable]\label{claim:partial-deterministic-spq-unlearnable}
The function family $\calF$ defined in \Cref{eqn:partially-deterministic} is unlearnable given a single physical query under the associated probability distribution $\calD$.
\end{claim}
\begin{proof}
We show that for every non-uniform quantum-polynomial-time $\calA$,
\begin{align*}
    \Pr_{f \leftarrow \calF_n}[\mathsf{LearningGame}_{\calF, \calD}^{\calA} = 1] \leq \negl(n).
\end{align*}
We consider a sequence of hybrids. Let the first hybrid $\calH_1$ be the learning game. In the second hybrid $\calH_2$, the function family is now
\begin{align*}
    f_{a, k}(x; r) = 
    \begin{dcases}
        (a, \mathsf{PRF}_k(0 \| r))&\text{if } x=0,\\
        (0, \mathsf{PRF}_k(a \| r))&\text{if } x=a,\\
        (0, \mathsf{PRF}_k(x \| r))&\text{otherwise.}
    \end{dcases}
\end{align*}
Since the adversary $\calA$ gets to make a single quantum query to the oracle $O_{f(\cdot, \$)^{(1)}}$, and since $a \leftarrow \{0,1\}^n$ is sampled uniformly at random by the challenger, with overwhelming probability, the weight placed by this query on input $x = a$ must be negligible. Therefore, hybrids $\calH_1$ and $\calH_2$ are indistinguishable by \cite{BBBV97}.

Now consider a third hybrid $\calH_3$ where the PRF is replaced by a random oracle $H$.
\begin{align*}
    f_{a, k}(x; r) = 
    \begin{dcases}
        (a, H(0 \| r))&\text{if } x=0,\\
        (0, H(a \| r))&\text{if } x=a,\\
        (0, H(x \| r))&\text{otherwise.}
    \end{dcases}
\end{align*}
By the security of the PRF against quantum superposition query attacks, hybrids $\calH_2$ and $\calH_3$ are indistinguishable. Therefore, since the random oracle family is unlearnable under a single physical query, so is the function family $\calF$.
\end{proof}

Since we know that the function family cannot be compiled into a one-time program that satisfies weak operational security definition, by \Cref{lem:partially-deterministic-weak-operational-impossibility} and \Cref{claim:partial-deterministic-spq-unlearnable}, we now know this function family cannot be compiled into a one-time program that satisfies the single physical query classical output simlution-based definition.
\begin{corollary}
    $\calF$ cannot be compiled into a one-time program under the single physival query classical-output simulation-based definition.
\end{corollary}

\begin{claim}[SEQ learnable]
    There is an adversary that, given single effective query access to $f_{\$, 1}$ succeeds in the learning game $\mathsf{LearningGame}_{\calF, \calD}$ with probability $1$.
\end{claim}
\begin{proof}
    The proof of \Cref{lem:partially-deterministic-weak-operational-impossibility} also shows that the function is learnable in the SEQ model.
\end{proof}
This trivially implies that the function family is one-time programmable in the SEQ model.
\begin{corollary}
    The $\OTP$ construction in \Cref{sec:construction} gives a one-time compiler for $\calF$ in the classical oracle model, under the single effective query simulation-based definition.
\end{corollary}

\section{Applications}
\label{sec:applications}

\subsection{Signature Tokens}
\paragraph{Motivation and Comparison to \cite{ben2023quantum}}

In this section we briefly discuss how to generate one-time tokens for the 
Fiat-Shamir signature schemes, by embedding our construction in \Cref{sec:construction} in to the plain signature scheme. One might wonder what benefit this gives us over the signature token in \cite{ben2023quantum}, where the signatures are simply measured subspace vectors corresponding to the messages.

One unsatisfactory property of \cite{ben2023quantum} is that it doesn't satisfy the regular existential unforgeability of signatures. If we give a "signing oracle" to the adversary, then it can trivially break the one-time security: given sufficiently many subspace vectors, one can recover the signing key. A more idealized notion would be allowing the adversary to query a signing oracle (as well as a signature token), but not enabling it to produce two signatures where neither is in the queried set.

The other unsatisfactory part of the \cite{ben2023quantum} signature scheme is that one has to use subspace vectors as signatures and hard to integrate other properties of a signature scheme we may want (e.g. a short signature scheme). 
More importantly, a corporation may have 
been using a plain signature scheme for a long time but when they occasionally need to delegate a one-time signing key to some external third-party, they have to change their entire cooperation's verification scheme into the subspace signature token scheme, which can result in more cost and inconvenience. 
Therefore, one interesting question is: Can we build a generic way to upgrade an existing signature scheme to be one-time secure such that the verification algorithm?


The advantage of the signature schemes below over \cite{ben2023quantum} is that they preserve the original signature scheme's properties. 
In particular, the verification algorithm is almost identical to the original verification algorithm of the signature scheme being compiled; the signature tokens produce signatures from the original scheme on messages of the form $m\concat r$ for some $r$. 
Thus, the verifier can use the original verification procedure and ignore the latter half of the signed message.



\paragraph{Blind Unforgeability}
We describe here how to compile signature schemes satisfying a certain notion of unforgeability with quantum query access into signature tokens (see \Cref{sec:token-sig-defs} for definitions). The notion we require is a slight variant on blind unforgeability.

\begin{definition}[Quantum Blind Unforgeability~\cite{EC:AMRS20}]\label{def:blind-unforge}
    A signature scheme $(\Gen, \Sign, \Verify)$ for message space $\calM$ is \emph{blind-unforgeable} if for every QPT adversary $\adv$ and blinding set $B\subset \calM$,
    \[
        \Pr\left[
                m\in B
                \land 
                \Verify(\vk, m, \sigma) = \mathsf{Accept} 
            :
            \begin{array}{c}
                (\sk, \vk) \gets \Gen(1^\secpar)
                \\
                (m, \sigma) \gets \adv^{\Sign_B(\sk, \cdot)}(\vk)
            \end{array}
        \right] = \negl(\lambda),
    \]
    where $\Sign_B(\sk, \cdot)$ denotes a (quantumly-accessible) signature oracle that signs messages $m$ using $\sk$ if $m\notin B$, and otherwise outputs $\bot$.
\end{definition}

This definition differs from the original in that the adversary may choose its blinding set $B$. \cite{EC:AMRS20} show that the hardness of this task is polynomially related to their original definition, which samples $B$ uniformly at random.

We note that any sub-exponentially secure signature scheme is blind-unforgeable\footnote{Here we need the scheme to be secure against subexponential time and queries, and we also need that the message space is smaller than this subexponential bound.}, since the adversary could simply query for all signatures in $\calM \backslash B$, then simulate the blind-unforgeability experiment. \cite{EC:AMRS20} also gives several other signature schemes which are blind-unforgeable (under the original definition).

\paragraph{Construction.}
Given a signing key $\sk$, the signer constructs a signature token by outputting a one-time program for the following functionality:

\begin{figure}[hpt]
    \centering
    \begin{mdframed}[
      linecolor=black,
      leftmargin =4em,
      rightmargin=4em,
      usetwoside=false,
    ]

    \textbf{Hardcoded:} A signing key $\sk$, two PRF key $k_1$ and $k_2$.
    \\
    On input $m$:
    \begin{enumerate}
        \item Sample randomness $r\gets \{0,1\}^\secpar$.
        \item Compute $r_1 = \PRF(k_1, m\concat r)$ and $r_2 = \PRF(k_2, m)$.
        \item Output $\Sign(\sk, m\concat r \concat r_1; r_2)$.
    \end{enumerate}
    \end{mdframed}
    \caption{Signature Token Functionality}
    \label{fig:sig_token_func}
\end{figure}

To sign a message $m$ using a signature token $T$, the temporary signer evaluates $T(m)$ and measures the output.

\begin{theorem}
    If $(\Gen, \Sign, \Verify)$ satisfies blind-unforgeability (\Cref{def:blind-unforge}), $\PRF$ is a psuedorandom function secure against quantum queries, and the one-time program satisfies SEQ simulation security (\Cref{def:simulation-style-otp-security}), then the above construction is one-time unforgeable (\Cref{def:token-unforge}).
\end{theorem}
\begin{proof}
    We first show that any QPT adversary can only sign messages of the form $m\concat r\concat \PRF(k_1, r)$ for some $r$. Consider the following hybrid experiments:
    \begin{itemize}
        \item $\Hyb_0$ is the one-time unforgeability game.
        \item $\Hyb_1$ is the same as $\Hyb_0$, except that $r_2$ is replaced by true randomness, instead of being a $\PRF$ evaluation.
        \item $\Hyb_2$ is the same as $\Hyb_1$, except that the SEQ oracle does not have the signing key $\sk$ hard-coded. Instead, it queries $m\concat r \concat r_1$ to an external signing oracle.
    \end{itemize}
    $\Hyb_1$ is computationally indistinguishable from $\Hyb_1$ by the psuedorandomness of $\PRF$. $\Hyb_2$ is identical to $\Hyb_1$ in the view of the adversary. Note that in $\Hyb_2$, the SEQ oracle \emph{only} submits queries of the form $m\concat r\concat \PRF(k_1, r)$ for some $r$. By the blind-unforgeability of the signature scheme, no adversary can produce signatures on messages not of this form in $\Hyb_2$. Since $\Hyb_2 \approx \Hyb_0$, this also holds in the original one-time unforgeability game.

    Finally, consider the hybrid experiment where $r_1$ is replaced by $G(r)$, where $G$ is a random function. This is indistinguishable from the one-time unforgeability game by the pseudorandomness of $\PRF$. By the previous claim, any adversary producing two signatures must have signed $m_1\concat r_1 \concat G(m\concat r_1)$ and $m_2\concat r_1 \concat G(m\concat r_1)$. If $m_1 \neq m_2$, then this contradicts the SEQ-unlearnability of random functions (\Cref{claim:random_function_unlearnable-body}).
\end{proof}

\subsection{One-Time NIZK Proofs}
\label{sec:onetime_proof}

\def\prk{\mathsf{prk}}
\def\vk{\mathsf{vk}}

A second application of one-time programs is to a notion of one-time non-interactive zero-knowledge (ZK) proofs that we define and construct. Here, a proving authority (say, a government) publishes a verification key $\vk$ and delegates to its subsidiaries the ability to certify (prove) a limited number of statements on its behalf by giving them a one-time proving token $\prk$. A prover in possession of $\prk$, an NP statement $x$ and its witness $w$, can generate a proof $\pi$ that anyone can verify against $\vk$. Importantly, he can only do so for a single valid statement-witness pair. Thus, we have the following tuple of algorithms:

\begin{itemize}
\item \textsf{Setup}: produces a master secret key $\msk$ together with a verification key/common reference string $\crs$. 
\item \textsf{Delegate}: on input $\msk$, produces a one-time proving token $\rho$.
\item \textsf{Prove}: on input $\rho$ and a statement-witness pair $(x,w)$, produces a proof $\pi$.
\item \textsf{Verify}: on input $x,\pi$ and $\vk$, outputs accept or reject.
\end{itemize}

Note that all objects here are classical except for the proving token $\rho$ which is a quantum state. In addition to the usual properties of completeness, soundness and zero knowledge, we require that the proving token is one-time use only. 

We construct a one-time proof token, following  constructions of one-time PRFs in the plain model.  The proving token consists of a sequence of $n = |x|$ many subspace states corresponding to subspaces $A_1,\ldots,A_n$ together with the obfuscation of a program that contains a PRF key $K$ together with the $A_i$; takes as input $x,w$ and $n$ vectors $v_1,\ldots,v_n$; checks that $(x,w) \in R_L$, and that each $v_i \in A_i$ if $x_i = 0$ and $v_i \in A_i^{\perp}$ if $x_i = 1$. If all checks pass, output $\mathsf{PRF}_K(x,v_1,\ldots,v_n)$; otherwise output $\bot$ \footnote{We may also obtain a construction from random oracle based NIZK such as Fiat Shamir, but it would require the use of classical oracles while using iO and PRF gives a plain model result.}.

The security of the one-time proof follows from similar arguments of the security for one-time PRF in  the plain model \Cref{sec:construction_io}.
We give a more formal description of the scheme below. 

\paragraph{One-time NIZK Security Definition.}
The one-time NIZK scheme first needs to satisfy the usual NIZK soundness and zero knowledge property. We omit these standard definitions here and refer to \cite{sahai2014use} Section 
5.5 for details.

We then define a very natural one-time security through the following game, as a special case of the \Cref{def:generalized_operational_otp_security_game} for NIZK proofs:
\begin{enumerate}
    \item The challenger samples $\crs \gets \setup(1^\lambda)$ and prepares the program $\OTP$ for the $\prove$ functionality. $\adv$ gets $\crs$ and a copy of the one-time program $\OTP$. 

    \item $\adv$ outputs two instance-proof pairs (or instance-randomness-proof tuple, for a relaxed notion)  $(x_1, \pi_1), (x_2, \pi_2)$ , which satisfies $x_1 \neq x_2$ or $\pi_1 \neq \pi_2$, otherwise $\calA$ loses.

    \item Challenger checks if $\Verify(x_1, \pi_1) = 1$ and
    $\Verify(x_2, \pi_2) = 1$. output 1 if and only both are satisfied.
    \end{enumerate}
We say that a one-time sampling program for NIZK satisfies  security if for any QPT adversary $\adv$, there exists a negligible function $\negl(\cdot)$, such that for all $\lambda \in \N$, the following holds:
\begin{align*}
    \Pr\left[\adv \text{ wins the above game }\right] \leq  \negl(\lambda).
\end{align*}

We defer the regular soundness and zero-knowledge definition of NIZK tp \Cref{sec:appendix_prelims}.

\paragraph{Construction} 
Let $F$ be a 
constrainable PRF that takes inputs of $\ell$ bits and outputs $\lambda$ bits. Let $f(\cdot)$ be a PRG. Let $L$ be a language and $R(\cdot, \cdot)$ be a relation that takes in an instance and a witness. Our system
will allow proofs of instances of $\ell$ bits and witness of $\ell'$ bits. The values of bounds given
by the values $\ell$ and $\ell'$
can be specified at setup, although we suppress that notation here.

For simplicity of presentation, we omit the need of a second PRF used to extract randomness in \Cref{sec:construction_plain_model} since it is only used to extract full entropy and will not affect the security proof.

\begin{description}
    \item $\setup(1^\lambda):$ 
    The setup algorithm in our case generates a common reference string $\crs$ along with a one-time proof token.
    \begin{enumerate}
        \item The setup algorithm first chooses a puncturable PRF key $K$ for $F$. Next, it creates an obfuscation of the $\Verify$ NIZK of \Cref{fig:program_OTP_verify}. 
        The size of the program is padded to be the maximum of itself and the program we define later in the security game. 

        \item It samples $n$ independent subspaces $\{\ket{A_i}\}_{i \in [n]}$. 
        It creates an obfuscation of the program Prove NIZK of \Cref{fig:program_OTP_proof}. 
        The size of the program is padded to be the maximum of itself and the program we define later in the security game.
    \end{enumerate}

The common reference string $\crs$ consists of the two obfuscated programs, and the master secret key is $\msk = \{A_i\}_{i \in [n]}$, the classical descriptions of the subspaces.

\item $\mathsf{Delegate}$: takes in a master secret key $\msk =  \{A_i\}_{i \in [n]}$ and outputs $\{\ket{A_i}\}_{i \in [n]}$ as the one-time proof token.

\item $\prove(\crs, (x, w), \{\ket{A_i}\}_{i \in [n]}):$ 
The NIZK prove algorithm runs the obfuscated program of $\prove$ from
CRS on inputs $(x, w)$ and subspace states $\ket{A_i}, i \in [n]$ in the following way: 
Apply QFT to each $\ket{A_i}$ if $x_i = 1$, else apply identity operator. Run the program in \Cref{fig:program_OTP_proof} on input $(x,w)$ and the modified subspace states.
If $R(x, w)$ holds the program returns a proof $\pi = (r, F(K, x\Vert r))$.

\begin{figure}[hpt]
    \centering
    \begin{mdframed}[
      linecolor=black,
      leftmargin =4em,
      rightmargin=4em,
      usetwoside=false,
    ]
        
        Hardcoded: $K, \{\shO_i^0,\shO_i^1\}_{i \in [n]}$.
        
        On input $((x,w) \in \{0,1\}^{n+m}, u = u_1 || u_2 || \cdots || u_n \in \{0,1\}^{n\cdot \lambda})$ (where each $u_i \in 
        \mathbb{F}_2^\lambda$):
        
        \begin{enumerate}
            \item If $(x,w) \in R_L$ and for all $i \in [n]$, $\shO_i^{x_i}(u_i) = 1$, where $x_i$ is the $i$-th bit of $x$:
            
                
                \quad Output $(r = u, F(k, x || r))$.
            \item Else:
            
                \quad Output $\bot$
        \end{enumerate}
    \end{mdframed}
    \caption{Program NIZK $ \prove_{\OTP}$}
    \label{fig:program_OTP_proof}
\end{figure}

\item $\Verify(x, \pi, \crs)$:
Run the input $(x, \pi)$ into the obfuscated program \Cref{fig:program_OTP_verify} and output the program's output.

\begin{figure}[hpt]
    \centering
    \begin{mdframed}[
      linecolor=black,
      leftmargin =4em,
      rightmargin=4em,
      usetwoside=false,
    ]
    
    Hardcoded: $K, f,  \{\shO_i^0,\shO_i^1\}_{i \in [n]}$.
    
    On input $(x \in \{0,1\}^{n}, \pi)$:
    
    \begin{enumerate}
        \item  Parse $\pi := (r, y)$ 
        
        \item If for all $i \in [n]$, $\shO_i^{x_i}(r_i) = 1$, where $x_i$ is the $i$-th bit of $x$: proceed to step 3; else output 0.
            
        \item Check if $f(y) =  f(F(K, x\Vert r)))$: if yes, output $1$; if no, output 0.
    \end{enumerate}
    \end{mdframed}
    \caption{Program NIZK $\Verify_{\OTP}$}
    \label{fig:program_OTP_verify}
\end{figure}

\end{description}

We then prove the following statement:
\begin{theorem}
\label{thm:one_time_nizk}   
Assuming LWE and subexponentially secure iO, there exists one-time NIZK proofs satisfying the above definition.
\end{theorem}

\paragraph{One-Time Security}
The correctness and one-time security proof follows relatively straightforward in a similar way as the proof for the PRF construction \Cref{sec:construction_plain_model}. 

We first replace the PRF $F$'s kardcoded key $K$ in both programs with a key $K^*$ that is constrained on a circuit to evaluate only on $(x,r)$ such that $r$ are subspace vectors corresponding to $x$. By the property of iO, this change is indistinguishable.  Then by the computational direct product hardness of the subspace states, we can argue that the adversary can only provide two proofs such that one of $r_1$ and $r_2$ are not valid subspace vectors, which will break the constrained pseudorandomness of the constrained PRF.

\paragraph{Soundness} The regular soundness of the construction is inspired by the proof for Theorem 9 in \cite{sahai2014use}.
but we need to use subexponentially secure iO and slightly more complicated hybrids. 
In hybrid 1, for some instance $x^* \notin L$, we constrain the PRF $F$ key $K$ used in the program \Cref{fig:program_OTP_proof} to a constrained key $K^*$ that evaluates on inputs $(x \Vert r)$ such that $x \neq x^*$. Since $x^* \notin L$, the functionality of the program $\prove_{\OTP}$ is unchanged and we can invoke iO. The next few steps deviate from \cite{sahai2014use}.

We design hybrids $(2.j.1)$ for $i = 1, 2\cdots,t, t = 2^{\lambda\cdot n/2}$: for all vectors $u_1 \in A^{x_1}, \cdots, u_n  \in A^{x_n}$, we make a lexigraphical order on them and call the $j$-th vector $\mathbf{u}_j$.  In hybrid $(2.j.1)$, we modify the program $\Verify_{\OTP}$ into the following. 

Note that the key $K_{x^*, r_j}$ is a punctured/constrained key that does not evaluate at $(x^*\vert r_j)$.

\begin{figure}[hpt]
\centering
\begin{mdframed}[
  linecolor=black,
  leftmargin =4em,
  rightmargin=4em,
  usetwoside=false,
]

    Hardcoded: $K_{x^*,r_j}, f,  \{\shO_i^0,\shO_i^1\}_{i \in [n]}, y^* = F(k,x^*\Vert r_j)$.
    
    On input $(x \in \{0,1\}^{n}, \pi)$:
    
    \begin{enumerate}
        \item  Parse $\pi := (r, y)$ 
        
        \item If for all $i \in [n]$, $\shO_i^{x_i}(r_i) = 1$, where $x_i$ is the $i$-th bit of $x$: proceed to step 3; else output 0.
        
        \item If $x = x^*$ and $r < r_j$: output 0.
        
        \item Else if $x = x^*$ and $r = r_j$: 
        
            \quad check if $f(y) = f(y^*)$: if yes, output $1$; if no, output 0.
            
        \item Else if $x \neq x^*$ or $r>r_j$:
        
          \quad Check if $f(y) =  f(F(K_{x^*, r_j}, x\Vert r)))$: if yes, output $1$; if no, output 0.
    \end{enumerate}
    \end{mdframed}
    \caption{Program NIZK $\Verify_{\OTP}$ in hybrid $2.j.1$. Note that the key $K_{x^*, r_j}$ is a punctured/constrained key that does not evaluate at $(x^*\vert r_j)$}
    \label{fig:program_OTP_verify_hybrid_j}
\end{figure}

Note that the functionality of the program is essentially the same as in the original program in \Cref{fig:program_OTP_verify} and therefore we can invoke the security of iO. After hybrid $(2.j.1)$, we introduce the next hybrid $(2.j.2)$.
In the next hybrid $(2.j.2)$, we will replace $y^* = F(K, x^*\Vert r_j)$ hardcoded in the program with a random value. This follows from the pseudorandomness at punctured/constrained values. In the following hybrid $(2.j.3)$, we replace 
the value $f(y^*)$ with a random value from the range of $f$. This is indistinguishable by the property of PRG $f$. Now with overwhelming probability, when the input is $(x^*, \pi = (r_j, y))$, there exists no value $y$ that will make the program output 1.

After hybrid $(2.j.3)$, we move to the next hybrid $(2.j+1.1)$, where 
$\Verify_{\OTP}$ is the same program as in \Cref{fig:program_OTP_verify_hybrid_j} but  we increment the counter $j$ to $j+1$.
By the above observations, the hybrid 
$(2.j.3)$ is statistically indistinguishable from the next hybrid $2.j+1.1$,  and now for all inputs where $x = x^*, r<r_{j+1}$, the program outputs 0.

Finally, after hybrid $2.t.3$, we obtain a $\Verify$ program that outputs all $0$ for inputs $(x, \pi = (r,y))$ where $x = x^*, r \in A^{x_1}, \cdots, A^{x_n}$. The adversary's advantage in getting an acceptance on instance $x^*$ is 0.

\paragraph{Zero Knowledge} The zero-knowledge proof follows exactly from \cite{sahai2014use}. A simulator $S$ that on input $x$, runs the setup algorithm and  outputs the corresponding $\crs$ 
along with $\pi = F(K, x\Vert r)$ for some $r$ of its own choice since it samples the subspaces on its own. The simulator has the exact same distribution as any real prover’s algorithm
for any $x \in L$ and witness $w$ where $R(x, w) = 1$.

\subsection{Future Work: One-Time MPC}
Our one-time sampling program has a definition in between a (deterministic) one-time memory and a fully random one-time memory, where the latter says that the receiver's input is  uniformly random and it does not get to choose which message it receives (similar discussions in \cite{bartusek2023secure}).
In this light, we hope that our one-time program may be extended to a one-time 2PC/MPC where each party has some choice over their inputs, but the rest of the inputs are random. We leave this question for future work.

\subsection{One-time programs for Verifiable Functions Imply Quantum Money}
In this section, we show that our definition of one-time (sampling) program for 
publicly verifiable functions implies public -key quantum money.

\begin{definition}[Public-key quantum money]
    A quantum money scheme consists of a pair of quantum polynomial-time algorithms $(\Gen, \Ver)$ with the following syntax, correctness and security specifications.
    \begin{itemize}
        \item Syntax: The generation procedure $\Gen(1^\lambda)$ takes as input a security parameter $\lambda$ and outputs a classical serial number $\sigma$ and a quantum state $\ket{\psi}$. The verification algorithm $\Ver(\sigma, \ket{\psi})$ outputs an accept/reject bit, $0$ or $1$.
        \item Correctness: There exists a negligible function $\negl$ such that
        \begin{align*}
            \Pr[1 \leftarrow \Ver(\Gen(1^\lambda))] \ge 1 - \negl(\lambda).
        \end{align*}
        \item Security: For all quantum polynomial-time algorithms $A$, there exists a negligible functions $\negl$ such that $A$ wins the following security game with probability $\negl(\lambda)$:
        \begin{itemize}
            \item The challenger runs $(\sigma, \ket{\psi}) \leftarrow \Gen(1^\lambda)$, and give $\sigma, \ket{\psi}$ to $A$.
            \item $A$ produces a (potentially entangled) join state $\rho_{1, 2}$ over two registers $\rho_1$ and $\rho_2$. $A$ sends $\rho_{1,2}$ to the challenger.
            \item The challenger runs $b_1 \leftarrow \Ver(\sigma, \rho_1)$ and $b_2 \leftarrow \Ver(\sigma, \rho_2)$. The adversary $A$ wins if and only if $b_1 = b_2 = 1$.
        \end{itemize}
    \end{itemize}
\end{definition}

\begin{theorem}
    Let $\calF$ be a verifiable randomized function family. Suppose there exists a one-time program scheme that satisfies the operational one-time security notion for verifiable functions (\Cref{def:verifiable-operational-security}). Then, a public-key quantum money scheme exists.
\end{theorem}

\begin{proof}
Let $\calF$ be a verifiable randomized function family, and let $\OTP = (\Generate_{\OTP}, \Evaluate_{\OTP})$ be a one-time program scheme that satisfies the operational one-time security notion for verifiable functions. Then, we can construct a quantum money scheme $\mathsf{QMoney} = (\Gen_{\QM}, \Ver_{\QM})$ as follows:
\begin{itemize}
    \item $\Gen_{\QM}(1^\lambda)$: Sample a function along with its verification key $f, vk_f \leftarrow \calF_\lambda$, and generate its one-time program $\ket{\psi} \leftarrow \OTP(f)$. Output classical serial number $\sigma = vk_f$ and the money state $\ket{\psi} = \OTP(f)$.
    \item $\Ver_{\QM}(\sigma, \rho)$: Sample a random $x \leftarrow \calX$ and initialize an input register $X$ to be $\ket{x}$. Coherently run (without performing any measurement) $\Evaluate_{\OTP}(\rho, x)$ to get the output in register $Y$. Coherently run $\Ver_{\calF}(vk_f, X, Y)$ on the $X$ and $Y$ registers and measure the resulting accept/reject bit. Accept if the verification passes and reject otherwise.
\end{itemize}

The correctness of this quantum money scheme follows from the correctness of the one-time program $\OTP(f)$, the verification procedure $\Ver_{\calF}$ as well as the gentle measurement lemma \Cref{lem:gentle_measure}. 

We now prove security. Suppose for contradiction there exists an adversary $A$ that wins the quantum money security game with non-negligible probability, so that it outputs a joint state $\rho_{1, 2}$ on two registers $\rho_1$ and $\rho_2$ such that both verification checks pass. Then, there exists an adversary $B$ that breaks the one-time program security of $\OTP(f)$: given $\OTP(f)$ and $vk_f$ as auxiliary information, $B$ runs $A$ on the same input, $A(\OTP(f), vk_f)$. Suppose that $A$ outputs (a joint state on) two registers $\rho_1, \rho_2$. Then, since the verification procedure verifies by sampling a random input to test the functionality on each state, running the verification procedure on both these registers results in two samples $(x_1, y_1)$ and $(x_2, y_2)$ with non-negligible probability, breaking the one-time program security of $\OTP$.
\end{proof}

\fi

\bibliographystyle{alpha}
\bibliography{abbrev3,crypto,references}

\newcommand{\etalchar}[1]{$^{#1}$}
\begin{thebibliography}{BBBV97b}

\bibitem[Aar04]{aaronson2004limitations}
Scott Aaronson.
\newblock Limitations of quantum advice and one-way communication.
\newblock In {\em Proceedings. 19th IEEE Annual Conference on Computational
  Complexity, 2004.}, pages 320--332. IEEE, 2004.

\bibitem[ABDS20]{alagic2020impossibility}
Gorjan Alagic, Zvika Brakerski, Yfke Dulek, and Christian Schaffner.
\newblock Impossibility of quantum virtual black-box obfuscation of classical
  circuits, 2020.

\bibitem[AC12]{aaronson2012quantum}
Scott Aaronson and Paul Christiano.
\newblock Quantum money from hidden subspaces.
\newblock In {\em Proceedings of the forty-fourth annual ACM symposium on
  Theory of computing}, pages 41--60. ACM, 2012.

\bibitem[ACE{\etalchar{+}}22]{EC:ACEGMPRT22}
Ghada Almashaqbeh, Ran Canetti, Yaniv Erlich, Jonathan Gershoni, Tal Malkin,
  Itsik Pe'er, Anna Roitburd-Berman, and Eran Tromer.
\newblock Unclonable polymers and their cryptographic applications.
\newblock In Orr Dunkelman and Stefan Dziembowski, editors, {\em
  EUROCRYPT~2022, Part~I}, volume 13275 of {\em {LNCS}}, pages 759--789.
  Springer, Cham, May~/~June 2022.

\bibitem[AMRS20]{EC:AMRS20}
Gorjan Alagic, Christian Majenz, Alexander Russell, and Fang Song.
\newblock Quantum-access-secure message authentication via
  blind-unforgeability.
\newblock In Anne Canteaut and Yuval Ishai, editors, {\em EUROCRYPT~2020,
  Part~III}, volume 12107 of {\em {LNCS}}, pages 788--817. Springer, Cham, May
  2020.

\bibitem[AP21]{ananth2020secure}
Prabhanjan Ananth and Rolando L.~La Placa.
\newblock Secure software leasing.
\newblock Springer-Verlag, 2021.

\bibitem[BBBV97a]{BBBV97}
Charles~H Bennett, Ethan Bernstein, Gilles Brassard, and Umesh Vazirani.
\newblock Strengths and weaknesses of quantum computing.
\newblock {\em SIAM journal on Computing}, 26(5):1510--1523, 1997.

\bibitem[BBBV97b]{BBBV1997}
Charles~H. Bennett, Ethan Bernstein, Gilles Brassard, and Umesh Vazirani.
\newblock Strengths and weaknesses of quantum computing.
\newblock {\em SIAM Journal on Computing}, 26(5):1510–1523, Oct 1997.

\bibitem[BDS23]{ben2023quantum}
Shalev Ben-David and Or~Sattath.
\newblock Quantum tokens for digital signatures.
\newblock {\em Quantum}, 7:901, 2023.

\bibitem[BGI{\etalchar{+}}01]{barak2001possibility}
Boaz Barak, Oded Goldreich, Rusell Impagliazzo, Steven Rudich, Amit Sahai,
  Salil Vadhan, and Ke~Yang.
\newblock On the (im) possibility of obfuscating programs.
\newblock In {\em Annual International Cryptology Conference}, pages 1--18.
  Springer, 2001.

\bibitem[BGS13]{broadbent2013quantum}
Anne Broadbent, Gus Gutoski, and Douglas Stebila.
\newblock Quantum one-time programs.
\newblock In {\em Annual Cryptology Conference}, pages 344--360. Springer,
  2013.

\bibitem[BJ15]{broadbent2015quantum}
Anne Broadbent and Stacey Jeffery.
\newblock Quantum homomorphic encryption for circuits of low t-gate complexity.
\newblock In {\em Annual Cryptology Conference}, pages 609--629. Springer,
  2015.

\bibitem[BKNY23]{bartusek2023obfuscation}
James Bartusek, Fuyuki Kitagawa, Ryo Nishimaki, and Takashi Yamakawa.
\newblock Obfuscation of pseudo-deterministic quantum circuits.
\newblock In {\em Proceedings of the 55th Annual ACM Symposium on Theory of
  Computing}, pages 1567--1578, 2023.

\bibitem[BKS23]{bartusek2023secure}
James Bartusek, Dakshita Khurana, and Akshayaram Srinivasan.
\newblock Secure computation with shared epr pairs (or: How to teleport in
  zero-knowledge).
\newblock In {\em Annual International Cryptology Conference}, pages 224--257.
  Springer, 2023.

\bibitem[BKW17]{boneh2017constrained}
Dan Boneh, Sam Kim, and David~J Wu.
\newblock Constrained keys for invertible pseudorandom functions.
\newblock In {\em Theory of Cryptography Conference}, pages 237--263. Springer,
  2017.

\bibitem[BLW17]{boneh2017constraining}
Dan Boneh, Kevin Lewi, and David~J Wu.
\newblock Constraining pseudorandom functions privately.
\newblock In {\em IACR International Workshop on Public Key Cryptography},
  pages 494--524. Springer, 2017.

\bibitem[Bra18]{brakerski2018quantum}
Zvika Brakerski.
\newblock Quantum fhe (almost) as secure as classical.
\newblock In {\em Annual International Cryptology Conference}, pages 67--95.
  Springer, 2018.

\bibitem[BV15]{brakerski2015constrained}
Zvika Brakerski and Vinod Vaikuntanathan.
\newblock Constrained key-homomorphic prfs from standard lattice assumptions:
  Or: How to secretly embed a circuit in your prf.
\newblock In {\em Theory of Cryptography: 12th Theory of Cryptography
  Conference, TCC 2015, Warsaw, Poland, March 23-25, 2015, Proceedings, Part II
  12}, pages 1--30. Springer, 2015.

\bibitem[CGLZ19]{chung2019cryptography}
Kai-Min Chung, Marios Georgiou, Ching-Yi Lai, and Vassilis Zikas.
\newblock Cryptography with disposable backdoors.
\newblock {\em Cryptography}, 3(3):22, 2019.

\bibitem[CHV23]{chevalier2023semi}
C{\'e}line Chevalier, Paul Hermouet, and Quoc-Huy Vu.
\newblock Semi-quantum copy-protection and more.
\newblock In {\em Theory of Cryptography Conference}, pages 155--182. Springer,
  2023.

\bibitem[CLLZ21a]{coladangelo2021hidden}
Andrea Coladangelo, Jiahui Liu, Qipeng Liu, and Mark Zhandry.
\newblock Hidden cosets and applications to unclonable cryptography.
\newblock In {\em Advances in Cryptology--CRYPTO 2021: 41st Annual
  International Cryptology Conference, CRYPTO 2021, Virtual Event, August
  16--20, 2021, Proceedings, Part I 41}, pages 556--584. Springer, 2021.

\bibitem[CLLZ21b]{C:CLLZ21}
Andrea Coladangelo, Jiahui Liu, Qipeng Liu, and Mark Zhandry.
\newblock Hidden cosets and applications to unclonable cryptography.
\newblock In Tal Malkin and Chris Peikert, editors, {\em CRYPTO~2021, Part~I},
  volume 12825 of {\em {LNCS}}, pages 556--584, Virtual Event, August 2021.
  Springer, Cham.

\bibitem[DFM20]{DFM20}
Jelle Don, Serge Fehr, and Christian Majenz.
\newblock The measure-andreprogram technique 2.0: multi-round fiat-shamir and
  more.
\newblock In {\em Annual International Cryptology Conference}, pages 602--631,
  2020.

\bibitem[DFMS19]{don2019security}
Jelle Don, Serge Fehr, Christian Majenz, and Christian Schaffner.
\newblock Security of the fiat-shamir transformation in the quantum
  random-oracle model.
\newblock In {\em Advances in Cryptology--CRYPTO 2019: 39th Annual
  International Cryptology Conference, Santa Barbara, CA, USA, August 18--22,
  2019, Proceedings, Part II 39}, pages 356--383. Springer, 2019.

\bibitem[GG17]{TCC:GoyGoy17}
Rishab Goyal and Vipul Goyal.
\newblock Overcoming cryptographic impossibility results using blockchains.
\newblock In Yael Kalai and Leonid Reyzin, editors, {\em TCC~2017, Part~I},
  volume 10677 of {\em {LNCS}}, pages 529--561. Springer, Cham, November 2017.

\bibitem[GGH{\etalchar{+}}16]{garg2016candidate}
Sanjam Garg, Craig Gentry, Shai Halevi, Mariana Raykova, Amit Sahai, and Brent
  Waters.
\newblock Candidate indistinguishability obfuscation and functional encryption
  for all circuits.
\newblock {\em SIAM Journal on Computing}, 45(3):882--929, 2016.

\bibitem[GIS{\etalchar{+}}10]{DBLP:conf/tcc/GoyalISVW10}
Vipul Goyal, Yuval Ishai, Amit Sahai, Ramarathnam Venkatesan, and Akshay Wadia.
\newblock Founding cryptography on tamper-proof hardware tokens.
\newblock In Daniele Micciancio, editor, {\em Theory of Cryptography, 7th
  Theory of Cryptography Conference, {TCC} 2010, Zurich, Switzerland, February
  9-11, 2010. Proceedings}, volume 5978 of {\em Lecture Notes in Computer
  Science}, pages 308--326. Springer, 2010.

\bibitem[GKR08a]{goldwasser2008one}
Shafi Goldwasser, Yael~Tauman Kalai, and Guy~N Rothblum.
\newblock One-time programs.
\newblock In {\em Advances in Cryptology--CRYPTO 2008: 28th Annual
  International Cryptology Conference, Santa Barbara, CA, USA, August 17-21,
  2008. Proceedings 28}, pages 39--56. Springer, 2008.

\bibitem[GKR08b]{C:GolKalRot08}
Shafi Goldwasser, Yael~Tauman Kalai, and Guy~N. Rothblum.
\newblock One-time programs.
\newblock In David Wagner, editor, {\em CRYPTO~2008}, volume 5157 of {\em
  {LNCS}}, pages 39--56. Springer, Berlin, Heidelberg, August 2008.

\bibitem[GKW17]{goyal2017lockable}
Rishab Goyal, Venkata Koppula, and Brent Waters.
\newblock Lockable obfuscation.
\newblock In {\em 2017 IEEE 58th Annual Symposium on Foundations of Computer
  Science (FOCS)}, pages 612--621. IEEE, 2017.

\bibitem[GM24]{gunn2024quantum}
Sam Gunn and Ramis Movassagh.
\newblock Quantum one-time protection of any randomized algorithm.
\newblock {\em private communication}, 2024.

\bibitem[Had00]{AC:Hada00}
Satoshi Hada.
\newblock Zero-knowledge and code obfuscation.
\newblock In Tatsuaki Okamoto, editor, {\em ASIACRYPT~2000}, volume 1976 of
  {\em {LNCS}}, pages 443--457. Springer, Berlin, Heidelberg, December 2000.

\bibitem[Liu23]{liu2023depth}
Qipeng Liu.
\newblock Depth-bounded quantum cryptography with applications to one-time
  memory and more.
\newblock In {\em 14th Innovations in Theoretical Computer Science Conference
  (ITCS 2023)}, 2023.

\bibitem[LSZ20]{liu2020quantum}
Qipeng Liu, Amit Sahai, and Mark Zhandry.
\newblock Quantum immune one-time memories.
\newblock {\em Cryptology ePrint Archive}, 2020.

\bibitem[Mah20]{mahadev2020classical}
Urmila Mahadev.
\newblock Classical homomorphic encryption for quantum circuits.
\newblock {\em SIAM Journal on Computing}, 52(6):FOCS18--189, 2020.

\bibitem[NC02]{nc02}
Michael~A Nielsen and Isaac Chuang.
\newblock Quantum computation and quantum information, 2002.

\bibitem[RKB{\etalchar{+}}18]{NATURE:RKBFW2018advantage}
Marie-Christine Roehsner, Joshua~A. Kettlewell, Tiago~B. Batalh{\~a}o,
  Joseph~F. Fitzsimons, and Philip Walther.
\newblock Quantum advantage for probabilistic one-time programs.
\newblock {\em Nature Communications}, 9(1):5225, Dec 2018.

\bibitem[RKFW21]{NATURE:RKFW21probabilistic}
Marie-Christine Roehsner, Joshua~A. Kettlewell, Joseph Fitzsimons, and Philip
  Walther.
\newblock Probabilistic one-time programs using quantum entanglement.
\newblock {\em npj Quantum Information}, 7(1):98, Jun 2021.

\bibitem[SW14]{sahai2014use}
Amit Sahai and Brent Waters.
\newblock How to use indistinguishability obfuscation: deniable encryption, and
  more.
\newblock In {\em Proceedings of the forty-sixth annual ACM symposium on Theory
  of computing}, pages 475--484, 2014.

\bibitem[VZ21]{vidick2021classical}
Thomas Vidick and Tina Zhang.
\newblock Classical proofs of quantum knowledge.
\newblock In {\em Annual International Conference on the Theory and
  Applications of Cryptographic Techniques}, pages 630--660. Springer, 2021.

\bibitem[Win99]{winter1999coding}
Andreas Winter.
\newblock Coding theorem and strong converse for quantum channels.
\newblock {\em IEEE Transactions on Information Theory}, 45(7):2481--2485,
  1999.

\bibitem[WZ82]{WoottersZurek}
W.~K. {Wootters} and W.~H. {Zurek}.
\newblock {A single quantum cannot be cloned}.
\newblock {\em Nature}, 299(5886):802--803, October 1982.

\bibitem[WZ17]{wichs2017obfuscating}
Daniel Wichs and Giorgos Zirdelis.
\newblock Obfuscating compute-and-compare programs under lwe.
\newblock In {\em 2017 IEEE 58th Annual Symposium on Foundations of Computer
  Science (FOCS)}, pages 600--611. IEEE, 2017.

\bibitem[Yue14]{QIC:Yuen14}
Henry Yuen.
\newblock A quantum lower bound for distinguishing random functions from random
  permutations.
\newblock {\em Quantum Inf. Comput.}, 14(13-14):1089--1097, 2014.

\bibitem[YZ21]{yamakawa2021classical}
Takashi Yamakawa and Mark Zhandry.
\newblock Classical vs quantum random oracles.
\newblock In {\em Annual International Conference on the Theory and
  Applications of Cryptographic Techniques}, pages 568--597. Springer, 2021.

\bibitem[Zha12a]{zhandry2012construct}
Mark Zhandry.
\newblock How to construct quantum random functions.
\newblock In {\em 2012 IEEE 53rd Annual Symposium on Foundations of Computer
  Science}, pages 679--687. IEEE, 2012.

\bibitem[Zha12b]{zhandry2012quantumprf}
Mark Zhandry.
\newblock How to construct quantum random functions.
\newblock In {\em Proceedings of the 2012 IEEE 53rd Annual Symposium on
  Foundations of Computer Science}, FOCS '12, page 679–687, USA, 2012. IEEE
  Computer Society.

\bibitem[Zha15]{QIC:Zha15}
Mark Zhandry.
\newblock A note on the quantum collision and set equality problems.
\newblock {\em Quantum Inf. Comput.}, 15(7{\&}8):557--567, 2015.

\bibitem[Zha16]{C:Zhandry16}
Mark Zhandry.
\newblock The magic of {ELFs}.
\newblock In Matthew Robshaw and Jonathan Katz, editors, {\em CRYPTO~2016,
  Part~I}, volume 9814 of {\em {LNCS}}, pages 479--508. Springer, Berlin,
  Heidelberg, August 2016.

\bibitem[Zha19a]{zhandry19compressed}
Mark Zhandry.
\newblock How to record quantum queries, and applications to quantum
  indifferentiability.
\newblock In Alexandra Boldyreva and Daniele Micciancio, editors, {\em Advances
  in Cryptology -- CRYPTO 2019}, pages 239--268, Cham, 2019. Springer
  International Publishing.

\bibitem[Zha19b]{zhandry2017quantum}
Mark Zhandry.
\newblock Quantum lightning never strikes the same state twice.
\newblock In {\em Annual International Conference on the Theory and
  Applications of Cryptographic Techniques}, pages 408--438. Springer, 2019.

\end{thebibliography}

\ifexabs
\else
\appendix
\section{Postponed Proofs for Families of Single-Query Unlearnable Functions}
\label{sec:exampls_seq_unlearnable}

\subsection{Pairwise Independent and Highly Random Functions}\label{sec:pairwise-indep}

In this section, we present the full proof that pairwise independent and highly random functions are single-effective-query $\negl(\lambda)$-unlearnable in the single-effective-query oracle model (\Cref{sec:seq_oracle_definition}).

\begin{lemma}
    Let $\calF$ be a family of functions mapping $\cX \times \cR \to \cY$ that satisfies:
    \begin{enumerate}
        \item \textit{Pairwise independence:} For any $(x, r, y), (x', r', y') \in \cX \times \cR \times \cY$ such that $(x, r) \neq (x', r')$, $\Pr_f[f(x, r) = y \land f(x', r') = y'] = \Pr_f[f(x, r) = y] \cdot \Pr_f[f(x', r') = y']$.
        \item \textit{High Randomness:} There is a negligible function $\nu(\secp)$ such that for any $(x, r, y) \in \cX \times \cR \times \cY$, 
        \[\Pr_f[f(x, r) = y] \leq \nu(n)\]
        \item $\frac{1}{|\cR|} = \negl(\secp)$
    \end{enumerate}
    Then $\cF$ is SEQ-$\negl(\secp)$-unlearnable.
\end{lemma}
\begin{proof}
$ $
\begin{enumerate}
    \item The SEQ oracle implements a compressed oracle for $H$. In particular, it maintains a database register $\cD_H$ that stores $\emptyset$ or a list of values $(x, r) \in \cX \times \cR$. See \Cref{sec:compressedRO} for a formal definition of the compressed oracle representation.
    
    The SEQ oracle also stores the function $f$ on register $\cF$. Let $\cF$ be initialized to the uniform superposition:
    \[\ket{F_\emptyset} := \frac{1}{\sqrt{|\cF|}} \sum_{f \in \cF} \ket{f}_\cF\]
    Then the SEQ oracle will answer queries coherently, without measuring the superposition.

    Also, let us represent $f$ as its truth table, so for every $(x,r) \in \cX \times \cR$, there is a register $\cF_{x,r}$ that holds the value $y = f(x,r)$.
    
    Let $\cO = \cD_H \times \cF$ be the oracle's internal register, let $\cQ$ be the query register submitted to the oracle, and let $\cA$ be the adversary's private register. We can assume that the state of the system is a pure state over $\cA \times \cQ \times \cO$ since the adversary and the SEQ oracle act as unitaries over these registers.

    \item Let us define the states that $\cO$ is allowed to be in, and let $E$ project onto the allowed states.
    \begin{align*}
        \text{Let } \cF_{x, r, y} &= \{f \in \cF : f(x, r) = y\}\\
        \ket{F_{x,r,y}} &= \frac{1}{\sqrt{|\cF_{x,r,y}|}} \cdot \sum_{f \in \cF_{x,r,y}} \ket{f}_{\cF}\\
        E_\cO &= \ketbra{\emptyset}_{\cD_H} \otimes \ketbra{F_\emptyset}_\cF + \sum_{(x,r,y) \in \cX \times \cR \times \cY} \ketbra{(x,r)}_{\cD_H} \otimes \ketbra{F_{x,r,y}}_\cF\\
        E &= \mathbb{I}_{\cA \times \cQ} \otimes E_\cO
    \end{align*}
    $E_\cO$ projects onto all states on $\cO$ in the span of the following basis states:
    \[\ket{\emptyset}_{\cD_H} \otimes \ket{F_\emptyset}_\cF \quad \text{ or } \quad \left(\ket{(x,r)}_{\cD_H} \otimes \ket{F_{x,r,y}}_\cF\right)_{(x,r,y) \in \cX \times \cR \times \cY}\]
    Let us also define $\overline{E} = \mathbb{I} - E$, and $\overline{E}_\cO = \mathbb{I} - E_\cO$.

    \item After any polynomial number of queries to the SEQ oracle, the state of the system $\ket{\psi}$ satisfies $\|E \cdot \ket{\psi}\| \geq 1 - \negl(\secp)$. This is proven in \Cref{thm:SEQ-oracle-keeps-states-in-sigma}.

    \item At the end of the SEQ learning game, the following steps are executed. We have added an additional step (step 2, shown in \textcolor{red}{red}). 
    \begin{enumerate}
        \item The adversary outputs two pairs $(x,r,y)$ and $(x',r',y')$ such that $(x,r)\neq(x',r')$.
        \item {\color{red}The challenger measures $\cD_H$ in the computational basis to obtain $(x'',r'') \in \cX \times \cR$ or $\emptyset$. If the outcome is some $(x'', r'')$ (and not $\emptyset$), then the challenger measures the register $\cF_{x'', r''}$ to obtain $y'' = f(x'',r'')$.}
        \item The challenger measures $\cF$ in the computational basis to get a function $f$.
        \item The challenger checks whether $f(x, r) = y$ and $f(x',r')=y'$. If so, the adversary wins. If not, the adversary loses.
    \end{enumerate}
    Adding step 2 does not affect the probability that the adversary wins the learning game because the measurements made in step 2 commute with the measurements made in step 3. Clearly, the measurement on $\cD_{H}$ commutes with the measurement on $\cF$ in step 3 because they act on disjoint registers. Furthermore, measuring $\cF_{x'', r''}$ is a partial measurement on the $\cF$ register in the computational basis. Therefore, it commutes with step 3's full measurement of $\cF$ in the computational basis.

    \item We will show that the measurement outcome of $f$ in step 3 is highly random, so the adversary wins the learning game with negligible probability. 
    
    Let us assume that at the start of step $1$, the state of the system $\ket{\psi}$ satisfies $E \cdot \ket{\psi} = \ket{\psi}$. Next, there are two cases to consider:
    
    \begin{enumerate}
        \item Case 1: The measurement on $\cD_H$ returns $\emptyset$. Then at the end of step 2, the state on $\cF$ is $\ket{F_\emptyset}$. When we measure $\cF$ in step 3, we obtain a uniformly random $f \getsr \cF$. The adversary will lose the learning game with overwhelming probability because
        \begin{align*}
            \Pr_{f \getsr \cF}[f(x, r)=y \land f(x',r')=y'] &= \Pr_{f \getsr \cF}[f(x, r)=y] \cdot \Pr_{f \getsr \cF}[f(x',r')=y'] \\
            &\leq \nu(\secp)^2\\
            &= \negl(\secp)
        \end{align*}
        \item Case 2: Otherwise, the measurement on $\cD_H$ returns some $(x'', r'')$, and we also measure $y''$. At the end of step 2, the state on the $\cF$ register is $\ket{F_{x'', r'', y''}}$. When we measure $\cF$ in step 3, we obtain a uniformly random $f \getsr \cF_{x'', r'', y''}$. 

        At least one of $(x,r)$ and $(x',r')$ do not equal $(x'',r'')$. Without loss of generality, let us say that $(x,r) \neq (x'',r'')$. Then the probability that the adversary wins the learning game is
        \begin{align*}
            \Pr[\cA \text{ wins}] &\leq \Pr_{f \getsr \cF_{x'',r'',y''}}[f(x,r)=y]\\
            &= \frac{|\cF_{x,r,y} \cap \cF_{x'',r'',y''}|}{|\cF_{x'',r'',y''}|}\\
            &= \frac{|\cF| \cdot \Pr_{f \getsr \cF}[f(x,r)=y \land f(x'',r'')=y'']}{|\cF| \cdot \Pr_{f \getsr \cF}[f(x'',r'')=y'']}\\
            &= \frac{\Pr_{f \getsr \cF}[f(x,r)=y] \cdot \Pr_{f \getsr \cF}[f(x'',r'')=y'']}{\Pr_{f \getsr \cF}[f(x'',r'')=y'']}\\
            &= \Pr_{f \getsr \cF}[f(x,r)=y]\\
            &\leq \nu(\secp) = \negl(\secp)
        \end{align*}
    \end{enumerate}
    
    We've shown that in both cases, the adversary wins the learning game with negligible probability. Therefore $\cF$ is SEQ-$\negl(\secp)$-unlearnable.
\end{enumerate}
\end{proof}

\begin{lemma}\label{thm:SEQ-oracle-keeps-states-in-sigma}
    After any polynomial number of queries to the SEQ oracle, the state of the system $\ket{\psi}$ satisfies $\|E \cdot \ket{\psi}\| \geq 1 - \negl(\secp)$.
\end{lemma}
\begin{proof}
$ $
\begin{enumerate}
    \item Let us define some useful operations. 
    \begin{itemize}
        \item Let $V_x$ be a projector acting on $\cO$ that projects onto all states for which $\cD_H$ contains $\emptyset$ or a single entry of the form $(x,r)$ for some $r \in \cR$. Furthermore, let $V$ be the following projector acting on $\cA \times \cQ \times \cO$:
        \[V = \sum_{x \in \cX} \mathbb{I}_{\cA \times \cQ_u \times \cQ_b} \otimes \ketbra{x}_{\cQ_x} \otimes (V_x)_\cO\]
        In other words, $V$ verifies whether the query $x$ on $\cQ_x$, if answered by the SEQ oracle, would keep the number of queries recorded in $\cD_H$ $\leq 1$. 

        Note that $E$ and $V$ commute with each other and share an eigenbasis.
        \item Let $\Decomp$ and $\Decomp_x$ be defined as they were in \cref{sec:compressedRO}.
        \item The SEQ oracle maintains an internal register $\cR$ that is used to store an intermediate $r$-value. Then $\CO'_H$ copies the output of $H(x)$ into the $\cR$ register. $\CO'_H$ computes the following mapping:
        \[\ket{x}_{\cQ_\cX} \otimes \ket{r'}_\cR \otimes \ket{(x,r)}_{\cD_H} \overset{\CO'_H}{\longrightarrow} \ket{x}_{\cQ_\cX} \otimes \ket{r' \oplus r}_\cR \otimes \ket{(x,r)}_{\cD_H}\]
        \item Let $\CO'_F$ answer the query to $f$ with the following mapping:
        \[\ket{x,u,b}_\cQ \otimes \ket{r}_{\cR} \otimes \ket{f}_\cF \overset{\CO'_F}{\longrightarrow} \ket{x,u \oplus f(x,r), b \oplus 1}_\cQ \otimes \ket{r}_{\cR} \otimes \ket{f}_\cF\]
    \end{itemize}

    \item Let us describe how the SEQ oracle operates.
    
    The SEQ oracle acts as the identity on all states $\ket{\psi}$ for which $V \cdot \ket{\psi} = \vec{0}$. And for all states $\ket{\psi}$ for which $V \cdot \ket{\psi} = \ket{\psi}$, the SEQ oracle operates as follows:
    \begin{enumerate}
        \item Initialize $\cR$ to $\ket{0}$. Then apply $\Decomp \circ \CO'_H \circ \Decomp$ to $\cQ_\cX \times \cR \times \cD_H$.
        \item Apply $\CO'_F$ to $\cQ \times \cR \times \cF$.
        \item Apply $\Decomp \circ \CO'_H \circ \Decomp$ to $\cQ_\cX \times \cR \times \cD_H$.
    \end{enumerate}
    Note that $\Decomp$ commutes with $\CO'_F$. $\Decomp$ depends on the computational basis state of $\cQ_\cX$, and applies $\Decomp_x$ to $\cD_H$. $\CO'_F$ depends on the computational basis state of $\cQ_\cX$ and otherwise acts on registers ($\cQ_u \times \cQ_b \times \cR \times \cF$) that are disjoint from the ones acted on by $\Decomp$.
    
    Furthermore, $\Decomp^{-1} = \Decomp$, so we can write the action of the SEQ oracle as follows:
    \[\Decomp \circ \CO'_H \circ \CO'_F \circ \CO'_H \circ \Decomp\]
    
    \item Let us assume that at the beginning of the query, the state $\ket{\psi}$ is in the span of $E$. This is true at the beginning of the first query because the state of $\cO$ is $\ket{\emptyset}_{\cD_H} \otimes \ket{F_\emptyset}_\cF$, which is in the span of $E$. Then we will show that at the end of the query, the state $\ket{\psi'}$ is negligibly close to a state in the span of $E$. Then by induction, after any polynomial number of queries to the SEQ oracle, the state of $\cO$ will be negligibly close to a state in the span of $E$.

    \item Next, we split $\ket{\psi}$ into the components in the span of $V$ and perpendicular to $V$.
    $\ket{\psi} = (\mathbb{I} - V) \cdot \ket{\psi} + V \cdot \ket{\psi}$
    Note that $\vec{s}_{\overline{V}} := (\mathbb{I} - V) \cdot \ket{\psi}$ and $\vec{s}_V := V \cdot \ket{\psi}$ are still in the span of $E$ because applying $(\mathbb{I} - V)$ or $V$ to $\ket{\psi}$ just checks whether $\cQ_\cX$ records a different $x$-value than $\cD_H$.

    The SEQ oracle applies $\mathbb{I}$ to the first component $\vec{s}_{\overline{V}}$ and applies $\Decomp \circ \CO'_H \circ \CO'_F \circ \CO'_H \circ \Decomp$ to $\vec{s}_V$. We will show that applying $\Decomp \circ \CO'_H \circ \CO'_F \circ \CO'_H \circ \Decomp$ to $\vec{s}_V$ produces a state that gives an overwhelming fraction of its amplitude to a state in the span of $E$.

    \item \Cref{thm:Decomp-does-not-move-states-out-of-E} says that for any state $\ket{\phi}$ on $\cA \times \cQ \times \cO$, such that $V \cdot \ket{\phi} = \ket{\phi}$, 
    \[\|E \cdot \ket{\phi}\| - \negl(\secp) \leq \|E \cdot \Decomp \cdot E \cdot \ket{\phi}\|\]

    This means that if a state starts in the span of $E$ and $V$, then applying $\Decomp$ to it will not move it out of the span of $E$, except by a negligible amount.

    Let us set $\ket{\phi} = \frac{\vec{s}_V}{\|\vec{s}_V\|}$. Note that $E \ket{\phi} = \ket{\phi}$, and $V \ket{\phi} = \ket{\phi}$. Then applying \cref{thm:Decomp-does-not-move-states-out-of-E} shows that 
    \begin{align*}
        1 - \negl(\secp) &\leq \frac{\|E \cdot \Decomp \cdot \vec{s}_V\|}{\|\vec{s}_V\|}\\
        \|\vec{s}_V\| - \negl'(\secp) &\leq \|E \cdot \Decomp \cdot \vec{s}_V\|
    \end{align*}

    \item Next, the operations $\CO'_H$ and $\CO'_F$ commute with $E$. 
    If a state $\ket{\phi}$ is in the span of $E$, then $\CO'_H \circ \CO'_F \circ \CO'_H \ket{\phi}$ will be in the span of $E$ as well. Therefore, 
    \[\|\vec{s}_V\| - \negl'(\secp) \leq \|E \cdot \CO'_H \circ \CO'_F \circ \CO'_H \circ \Decomp \cdot \vec{s}_V\|\]
    Furthermore, $\CO'_H \circ \CO'_F \circ \CO'_H \circ \Decomp \cdot \vec{s}_V$ will be in the span of $V$ since all of the operations $\CO'_H, \CO'_F, \Decomp$ map a state in the span of $V$ to a state in the span of $V$.

    \item After applying $\CO'_H \circ \CO'_F \circ \CO'_H$, the SEQ oracle applies $\Decomp$ again. We can apply \cref{thm:Decomp-does-not-move-states-out-of-E} again to show that 
    \[\|\vec{s}_V\| - \negl(\secp) \leq \|E \cdot \Decomp \circ \CO'_H \circ \CO'_F \circ \CO'_H \circ \Decomp \cdot \vec{s}_V\|\]

    \item In summary, after the query to the SEQ oracle, the state is
    \begin{align*}
        \ket{\psi'} &= (\mathbb{I} - V) \cdot \ket{\psi} + \left(\Decomp \circ \CO'_H \circ \CO'_F \circ \CO'_H \circ \Decomp \cdot V\right) \cdot \ket{\psi}
    \end{align*}
    Then, 
    \begin{align*}
        \|E \cdot \ket{\psi'}\|^2 &= \|E \cdot (\mathbb{I} - V) \cdot \ket{\psi}\|^2 + \|E \cdot \left(\Decomp \circ \CO'_H \circ \CO'_F \circ \CO'_H \circ \Decomp \cdot V\right) \cdot \ket{\psi}\|^2\\
        &\geq \|(\mathbb{I} - V) \cdot \ket{\psi}\|^2 + \|V \cdot \ket{\psi}\|^2 - \negl(\secp)\\
        &= \|\ket{\psi}\|^2 - \negl(\secp)\\
        &= 1 - \negl(\secp)
    \end{align*}

    This shows that after any polynomial number of queries to the SEQ oracle, the state of the system $\ket{\psi}$ satisfies $\|E \cdot \ket{\psi}\| \geq 1 - \negl(\secp)$.
\end{enumerate}
\end{proof}

\begin{lemma}\label{thm:Decomp-does-not-move-states-out-of-E}
    For any state $\ket{\psi}$ on $\cA \times \cQ \times \cO$, such that $V \cdot \ket{\psi} = \ket{\psi}$, $\|E \cdot \ket{\psi}\| - \negl(\secp) \leq \|E \cdot \Decomp \cdot E \cdot \ket{\psi}\|$.
\end{lemma}
\begin{proof}
$ $
\begin{enumerate}
    \item A generic state over $\cA \times \cQ \times \cO$ such that $V \cdot \ket{\psi} = \ket{\psi}$ can be written as follows: 
    \[\ket{\psi} = \sum_{a,x,u,b} \alpha_{a,x,u,b} \cdot \ket{a}_\cA \otimes \ket{x,u,b}_\cQ \otimes \ket{\Psi_{a,x,u,b}}_\cO\]
    where $1 = \sum_{a,x,u,b} |\alpha_{a,x,u,b}|^2$, and for each $(a,x,u,b)$, $V_x \cdot \ket{\Psi_{a,x,u,b}} = \ket{\Psi_{a,x,u,b}}$. This means the $\cD_H$ register of $\ket{\Psi_{a,x,u,b}}$ contains either $\emptyset$ or $(x,r)$ for some $r \in \cR$.

    Next,
    \begin{align*}
        E \cdot \Decomp \cdot E \cdot \ket{\psi} &= \sum_{a,x,u,b} \alpha_{a,x,u,b} \cdot \ket{a}_\cA \otimes \ket{x,u,b}_\cQ \otimes \left(E_\cO \cdot \Decomp_x \cdot E_\cO \cdot \ket{\Psi_{a,x,u,b}}_\cO\right)\\
        \left\|E \cdot \Decomp \cdot E \cdot \ket{\psi}\right\|_2^2 &= \sum_{a,x,u,b} |\alpha_{a,x,u,b}|^2 \cdot \left\|E_\cO \cdot \Decomp_x \cdot E_\cO \cdot \ket{\Psi_{a,x,u,b}}_\cO\right\|_2^2
    \end{align*}
    Additionally,
    \begin{align*}
        E \cdot \ket{\psi} &= \sum_{a,x,u,b} \alpha_{a,x,u,b} \cdot \ket{a}_\cA \otimes \ket{x,u,b}_\cQ \otimes \left(E_\cO \cdot \ket{\Psi_{a,x,u,b}}_\cO\right)\\
        \left\|E \cdot \ket{\psi}\right\|_2^2 &= \sum_{a,x,u,b} |\alpha_{a,x,u,b}|^2 \cdot \left\|E_\cO \cdot \ket{\Psi_{a,x,u,b}}_\cO\right\|_2^2
    \end{align*}

    It suffices to prove that for any $x \in \cX$, and any state $\ket{\Psi}$ on register $\cO$ such that $V_x \cdot \ket{\Psi} = \ket{\Psi}$, 
    \[\left\|E_\cO \cdot \ket{\Psi}_\cO\right\|_2 - \negl(\secp) \leq \left\|E_\cO \cdot \Decomp_x \cdot E_\cO \cdot \ket{\Psi}_\cO\right\|_2\]
    because that would imply that 
    \begin{align*}
        \sum_{a,x,u,b} |\alpha_{a,x,u,b}|^2 \cdot \left(\left\|E_\cO \cdot \ket{\Psi_{a,x,u,b}}_\cO\right\|_2^2 - \negl(\secp)\right) &\leq \left\|E \cdot \Decomp \cdot E \cdot \ket{\psi}\right\|_2^2\\
        \left\|E \cdot \ket{\psi}\right\|_2^2 - \negl(\secp) &= \\
        \left\|E \cdot \ket{\psi}\right\|_2 - \negl'(\secp) &\leq \left\|E \cdot \Decomp \cdot E \cdot \ket{\psi}\right\|_2
    \end{align*}
    which is what we wanted to prove. From now on, we will focus on proving that for any $x \in \cX$, and any state $\ket{\Psi}$ on register $\cO$ such that $V_x \cdot \ket{\Psi} = \ket{\Psi}$, 
    \[\left\|E_\cO \cdot \ket{\Psi}_\cO\right\|_2 - \negl(\secp) \leq \left\|E_\cO \cdot \Decomp_x \cdot E_\cO \cdot \ket{\Psi}_\cO\right\|_2\]
    
    \item 
    \begin{align*}
        \text{Let } E_\cO \cdot \ket{\Psi}_\cO &= v_{\emptyset} \cdot \ket{\emptyset}_{\cD_H} \otimes \ket{F_\emptyset}_\cF + \sum_{(r, y) \in \cR \times \cY} v_{r,y} \cdot \ket{(x, r)}_{\cD_H} \otimes \ket{F_{x,r,y}}_{\cF}
    \end{align*}
    where $\vec{v} = (v_{\emptyset}, (v_{r,y})_{(r,y \in \cR\times\cY)})$ is a vector of norm $\|\vec{v}\| \leq 1$.

    Next, we will show that $\|\overline{E}_\cO \cdot \Decomp_x \cdot E_\cO \cdot \ket{\Psi}\| = \negl(\secp)$.
    \begin{align*}
        &\overline{E}_\cO \cdot \Decomp_x \cdot E_\cO \cdot \ket{\Psi} &= \overline{E}_\cO \cdot \Decomp_x \cdot \left(v_{\emptyset} \cdot \ket{\emptyset}_{\cD_H} \otimes \ket{F_\emptyset}_\cF + \sum_{(r, y) \in \cR \times \cY} v_{r,y} \cdot \ket{(x, r)}_{\cD_H} \otimes \ket{F_{x,r,y}}_{\cF}\right)
    \end{align*}

    The first term -- $\Decomp_x \cdot \left(\ket{\emptyset}_{\cD_H} \otimes \ket{F_\emptyset}_\cF\right)$ -- lies in the span of $E_\cO$:
    \begin{align*}
        \Decomp_x \left(\ket{\emptyset}_{\cD_H} \otimes \ket{F_\emptyset}_\cF\right) &= \frac{1}{\sqrt{|\cR|}} \cdot \sum_{r \in \cR} \ket{(x,r)}_{\cD_H} \otimes \ket{F_\emptyset}_\cF\\
        &= \frac{1}{\sqrt{|\cR| \cdot |\cF|}} \cdot \sum_{(r, f) \in \cR \times \cF} \ket{(x,r)}_{\cD_H} \otimes \ket{f}_\cF\\
        &= \frac{1}{\sqrt{|\cR| \cdot |\cF|}} \cdot \sum_{(r,y) \in \cR \times \cY} \sqrt{|\cF_{x,r,y}|} \cdot \ket{(x,r)}_{\cD_H} \otimes \ket{F_{x,r,y}}_\cF\\
        \overline{E}_\cO \cdot \Decomp_x \left(\ket{\emptyset}_{\cD_H} \otimes \ket{F_\emptyset}_\cF\right) &= \vec{0}
    \end{align*}
    We used the fact that for any $(x,r)$, the sets $(\cF_{x,r,y})_{y \in \cY}$ partition $\cF$.
    
    Therefore, we only need to focus on the remaining terms.
    \begin{align*}
        \overline{E}_\cO \cdot \Decomp_x \cdot E_\cO \cdot \ket{\Psi} &= \overline{E}_\cO \cdot \sum_{(r, y) \in \cR \times \cY} v_{r,y} \cdot (\Decomp_x \cdot \ket{(x, r)}_{\cD_H}) \otimes \ket{F_{x,r,y}}_{\cF}\\
        &= \sum_{(r, y) \in \cR \times \cY} v_{r,y} \cdot \left(1 - \frac{1}{|\cR|}\right) \cdot \overline{E}_\cO \cdot \ket{(x, r)} \otimes \ket{F_{x,r,y}}_{\cF}\\
        &+ \sum_{(r, y) \in \cR \times \cY} v_{r,y} \cdot \frac{-1}{|\cR|} \cdot \sum_{r' \in \cR \backslash \{r\}} \overline{E}_\cO \cdot \ket{(x,r')}\otimes \ket{F_{x,r,y}}_{\cF} \\
        &+ \sum_{(r, y) \in \cR \times \cY} v_{r,y} \cdot \frac{1}{\sqrt{|\cR|}} \cdot \overline{E}_\cO \cdot\ket{\emptyset} \otimes \ket{F_{x,r,y}}_{\cF}
    \end{align*}
    We used the fact that 
    \[\Decomp_x \cdot \ket{(x, r)}_{\cD_H} = \left(1 - \frac{1}{|\cR|}\right) \cdot \ket{(x, r)} - \frac{1}{|\cR|} \cdot \sum_{r' \in \cR \backslash \{r\}} \ket{(x,r')} + \frac{1}{\sqrt{|\cR|}} \cdot \ket{\emptyset}\]
    from \cref{thm:Decomp-applied-to-database}.

    \item Next, 
    \begin{align*}
        \overline{E}_\cO \cdot \Decomp_x \cdot E_\cO \cdot \ket{\Psi} &= \sum_{(r, y) \in \cR \times \cY} v_{r,y} \cdot \left(1 - \frac{1}{|\cR|}\right) \cdot \vec{0}\\
        &+ \sum_{(r, y) \in \cR \times \cY} v_{r,y} \cdot \frac{-1}{|\cR|} \cdot \sum_{r' \in \cR \backslash \{r\}} \ket{(x,r')} \otimes \left(\mathbb{I} - \sum_{y' \in \cY} \ketbra{F_{x,r',y'}}\right) \cdot \ket{F_{x,r,y}}_{\cF}\\
        &+ \sum_{(r, y) \in \cR \times \cY} v_{r,y} \cdot \frac{1}{\sqrt{|\cR|}} \cdot \ket{\emptyset} \otimes (\mathbb{I} - \ketbra{F_\emptyset}) \cdot \ket{F_{x,r,y}}_{\cF}
    \end{align*}
    Next, we change the order of summation to obtain:
    \begin{align*}
        &\overline{E}_\cO \cdot \Decomp_x \cdot E_\cO \cdot \ket{\Psi}\\
        &\quad\quad= \frac{1}{\sqrt{|\cR|}} \cdot \sum_{r' \in \cR} \ket{(x,r')} \otimes \left(\mathbb{I} - \sum_{y' \in \cY} \ketbra{F_{x,r',y'}}\right) \cdot \left(\sum_{(r, y) \in \cR \backslash\{r'\} \times \cY} \frac{-v_{r,y}}{\sqrt{|\cR|}} \cdot \ket{F_{x,r,y}}_{\cF}\right)\\
        &\quad\quad+ \ket{\emptyset} \otimes (\mathbb{I} - \ketbra{F_\emptyset}) \cdot \left(\sum_{(r, y) \in \cR \times \cY} \frac{v_{r,y}}{\sqrt{|\cR|}} \cdot \ket{F_{x,r,y}}_{\cF}\right)
    \end{align*}
    
    \item Let $M_x$ be a matrix whose columns are $\left(\frac{1}{\sqrt{|\cR|}} \cdot \ket{F_{x,r,y}}\right)_{(r,y) \in \cR \times \cY}$.
    
    Also, for a given $r' \in \cR$, let $\vec{v}_{r'}$ be the same as $\vec{v}$ except that for every $y \in \cY$, the $(r',y)$-th entry is set to $0$. Note that $\|\vec{v}_{r'}\| \leq \|\vec{v}\| \leq 1$.
    
    Then,
    \begin{align*}
        \overline{E}_\cO \cdot \Decomp_x \cdot E_\cO \cdot \ket{\Psi} &= \frac{1}{\sqrt{|\cR|}} \cdot \sum_{r' \in \cR} \ket{(x,r')} \otimes \left(\mathbb{I} - \sum_{y' \in \cY} \ketbra{F_{x,r',y'}}\right) \cdot \left(-M_x \cdot \vec{v}_{r'}\right)\\
        &+ \ket{\emptyset} \otimes (\mathbb{I} - \ketbra{F_\emptyset}) \cdot M_x \cdot \vec{v}\\
        \left\|\overline{E}_\cO \cdot \Decomp_x \cdot E_\cO \cdot \ket{\Psi}\right\|_2^2 &= \frac{1}{|\cR|} \cdot \sum_{r' \in \cR} \left\|\left(\mathbb{I} - \sum_{y' \in \cY} \ketbra{F_{x,r',y'}}\right) \cdot M_x \cdot \vec{v}_{r'}\right\|_2^2\\
        &+ \|(\mathbb{I} - \ketbra{F_\emptyset}) \cdot M_x \cdot \vec{v}\|_2^2\\
        &\leq \frac{1}{|\cR|} \cdot \sum_{r' \in \cR} \left\|\left(\mathbb{I} - \sum_{y' \in \cY} \ketbra{F_{x,r',y'}}\right) \cdot M_x\right\|_2^2 \cdot \left\|\vec{v}_{r'}\right\|_2^2\\
        &+ \left\|\left(\mathbb{I} - \ketbra{F_\emptyset}\right) \cdot M_x\right\|_2^2 \cdot \left\|\vec{v}\right\|_2^2
    \end{align*}

    \item We know that 
    \begin{align*}
        \left\|\left(\mathbb{I} - \sum_{y' \in \cY} \ketbra{F_{x,r',y'}}\right) \cdot M_x\right\|_2 \leq \left\|\left(\mathbb{I} - \ketbra{F_\emptyset}\right) \cdot M_x\right\|_2 \leq \frac{1}{\sqrt{|\cR|}} 
    \end{align*}
    by \Cref{thm:norm-of-M-x,thm:norm-of-a-projector-applied-to-M}.
    Therefore, 
    \begin{align*}
        \left\|\overline{E}_\cO \cdot \Decomp_x \cdot E_\cO \cdot \ket{\Psi}\right\|_2^2 &\leq \frac{1}{|\cR|} \cdot \sum_{r' \in \cR} \frac{1}{|\cR|} \cdot \left\|\vec{v}_{r'}\right\|_2^2\\
        &+ \frac{1}{|\cR|} \cdot \left\|\vec{v}\right\|_2^2\\
        &\leq \frac{1}{|\cR|} \cdot \left(\left\|\vec{v}_{r'}\right\|_2^2 + \left\|\vec{v}\right\|_2^2\right) \leq \frac{2}{|\cR|}\\
        \left\|\overline{E}_\cO \cdot \Decomp_x \cdot E_\cO \cdot \ket{\Psi}\right\|_2 &\leq \sqrt{\frac{2}{|\cR|}} = \negl(\secp)
    \end{align*}

    \item Finally,
    \begin{align*}
        \left\|E_\cO \cdot \ket{\Psi}_\cO\right\|_2^2 &= \left\|\Decomp_x \cdot E_\cO \cdot \ket{\Psi}_\cO\right\|_2^2\\
        &= \left\|E_\cO \cdot \Decomp_x \cdot E_\cO \cdot \ket{\Psi}_\cO\right\|_2^2 + \left\|\overline{E}_\cO \cdot \Decomp_x \cdot E_\cO \cdot \ket{\Psi}_\cO\right\|_2^2\\
        &\leq \left\|E_\cO \cdot \Decomp_x \cdot E_\cO \cdot \ket{\Psi}_\cO\right\|_2^2 + \frac{2}{|\cR|}
    \end{align*}
    We used the fact that $\Decomp_x$ is a unitary, so it preserves norms. Then,
    \begin{align*}
        \left\|E_\cO \cdot \ket{\Psi}_\cO\right\|_2^2 - \frac{2}{|\cR|} &\leq \left\|E_\cO \cdot \Decomp_x \cdot E_\cO \cdot \ket{\Psi}_\cO\right\|_2^2\\
        \left\|E_\cO \cdot \ket{\Psi}_\cO\right\|_2 - \negl(\secp) &\leq \left\|E_\cO \cdot \Decomp_x \cdot E_\cO \cdot \ket{\Psi}_\cO\right\|_2
    \end{align*}
    which completes the proof.
\end{enumerate}
\end{proof}

\begin{lemma}\label{thm:Decomp-applied-to-database}
    $\Decomp_x \cdot \ket{(x, r)}_{\cD_H} = \left(1 - \frac{1}{|\cR|}\right) \cdot \ket{(x, r)} - \frac{1}{|\cR|} \cdot \sum_{r' \in \cR \backslash \{r\}} \ket{(x,r')} + \frac{1}{\sqrt{|\cR|}} \cdot \ket{\emptyset}$
\end{lemma}
\begin{proof}
\begin{align*}
    \text{Let } \vec{v_0} &= \frac{1}{|\cR|} \cdot \sum_{r' \in \cR} \ket{(x,r')}\\
    \vec{v_1} &= \ket{(x, r)} - \frac{1}{|\cR|} \cdot \sum_{r' \in \cR} \ket{(x,r')}\\
    \ket{(x, r)} &= \vec{v_0} + \vec{v_1}
\end{align*}

Next, $\vec{v_0}$ and $\vec{v_1}$ are orthogonal:
\begin{align*}
    \braket{\vec{v}_0|\vec{v}_1} &= \left(\frac{1}{|\cR|} \cdot \sum_{r'' \in \cR} \bra{(x,r'')}\right) \cdot \left(\ket{(x, r)} - \frac{1}{|\cR|} \cdot \sum_{r' \in \cR} \ket{(x,r')}\right)\\
    &= \frac{1}{|\cR|} \braket{(x,r)|(x,r)} - \frac{1}{|\cR|^2} \cdot \sum_{r' \in \cR} \braket{(x,r')|(x,r')}\\
    &= \frac{1}{|\cR|} - \frac{1}{|\cR|^2} \cdot |\cR| = 0
\end{align*}

Next, $\Decomp_x$ maps $\vec{v_0}$ to $\frac{1}{\sqrt{|\cR|}} \cdot \ket{\emptyset}$ and acts as the identity on $\vec{v_1}$. Therefore,
\begin{align*}
    \Decomp_x \cdot \ket{(x, r)} &= \frac{1}{\sqrt{|\cR|}} \cdot \ket{\emptyset} + \vec{v_1}\\
    &= \ket{(x, r)} - \frac{1}{|\cR|} \cdot \sum_{r' \in \cR} \ket{(x,r')} + \frac{1}{\sqrt{|\cR|}} \cdot \ket{\emptyset}\\
    &= \left(1 - \frac{1}{|\cR|}\right) \cdot \ket{(x, r)} - \frac{1}{|\cR|} \cdot \sum_{r' \in \cR \backslash \{r\}} \ket{(x,r')} + \frac{1}{\sqrt{|\cR|}} \cdot \ket{\emptyset}
\end{align*}
\end{proof}

\begin{lemma}\label{thm:norm-of-a-projector-applied-to-M}
    For any matrix $M$ of the appropriate dimensions,
    \[\left\|\left(\mathbb{I} - \sum_{y' \in \cY} \ketbra{F_{x,r',y'}}\right) \cdot M\right\|_2 \leq \left\|\left(\mathbb{I} - \ketbra{F_\emptyset}\right) \cdot M\right\|_2\]
\end{lemma}
\begin{proof}
$ $
\begin{enumerate}
    \item Note that $\mathbb{I} - \sum_{y' \in \cY} \ketbra{F_{x,r',y'}}$ is a projector because the states $\left(\ket{F_{x,r',y'}}\right)_{y' \in \cY}$ are mutually orthogonal. Therefore,
    \[\left\|\left(\mathbb{I} - \sum_{y' \in \cY} \ketbra{F_{x,r',y'}}\right)\right\|_2 \leq 1\]

    \item $\ket{F_\emptyset}$ is orthogonal to $\mathbb{I} - \sum_{y' \in \cY} \ketbra{F_{x,r',y'}}$.
    \begin{align*}
        \left(\mathbb{I} - \sum_{y' \in \cY} \cdot \ketbra{F_{x,r',y'}}\right) \cdot \ket{F_{\emptyset}} &= \ket{F_{\emptyset}} - \sum_{y' \in \cY} \ket{F_{x,r',y'}} \cdot \braket{F_{x,r',y'} | F_\emptyset}\\
        &= \ket{F_{\emptyset}} - \sum_{y' \in \cY} \sum_{f \in \cF_{x,r',y'}} \ket{f} \cdot \frac{1}{\sqrt{|\cF_{x,r',y'}|}} \cdot \left(\sum_{f' \in \cF_{x,r',y'}}\frac{1}{\sqrt{|\cF_{x,r',y'}| \cdot |\cF|}}\right)\\
        &= \ket{F_{\emptyset}} - \sum_{y' \in \cY} \sum_{f \in \cF_{x,r',y'}} \ket{f} \cdot \left(\frac{|\cF_{x,r',y'}|}{|\cF_{x,r',y'}| \cdot \sqrt{|\cF|}}\right)\\
        &= \ket{F_{\emptyset}} - \sum_{y' \in \cY} \sum_{f \in \cF_{x,r',y'}} \ket{f} \cdot \frac{1}{\sqrt{|\cF|}}\\
        &= \ket{F_{\emptyset}} - \sum_{f \in \cF} \ket{f} \cdot \frac{1}{\sqrt{|\cF|}} = \ket{F_{\emptyset}} - \ket{F_{\emptyset}}\\
        &= \vec{0}
    \end{align*}

    \item Next,
    \begin{align*}
        \left(\mathbb{I} - \sum_{y' \in \cY} \ketbra{F_{x,r',y'}}\right) \cdot \left(\mathbb{I} - \ketbra{F_\emptyset}\right) &= \left(\mathbb{I} - \sum_{y' \in \cY} \ketbra{F_{x,r',y'}}\right) \\
        &\quad- \left(\mathbb{I} - \sum_{y' \in \cY} \ketbra{F_{x,r',y'}}\right) \cdot \ketbra{F_\emptyset}\\
        &= \left(\mathbb{I} - \sum_{y' \in \cY} \ketbra{F_{x,r',y'}}\right)
    \end{align*}

    \item Finally, 
    \begin{align*}
        \left\|\left(\mathbb{I} - \sum_{y' \in \cY} \ketbra{F_{x,r',y'}}\right) \cdot M\right\|_2 &= \left\|\left(\mathbb{I} - \sum_{y' \in \cY} \ketbra{F_{x,r',y'}}\right) \cdot \left(\mathbb{I} - \ketbra{F_\emptyset}\right) \cdot M\right\|_2\\
        &\leq \left\|\left(\mathbb{I} - \sum_{y' \in \cY} \ketbra{F_{x,r',y'}}\right)\right\|_2 \cdot \left\|\left(\mathbb{I} - \ketbra{F_\emptyset}\right) \cdot M\right\|_2\\
        &\leq \left\|\left(\mathbb{I} - \ketbra{F_\emptyset}\right) \cdot M\right\|_2
    \end{align*}
\end{enumerate}
\end{proof}

\begin{lemma}\label{thm:norm-of-M-x}
        Let $M_x$ be a matrix whose columns are $\left(\frac{1}{\sqrt{|\cR|}} \cdot \ket{F_{x,r,y}}\right)_{(r,y) \in \cR \times \cY}$. Then 
        \[\left\|(\mathbb{I} - \ketbra{F_\emptyset}) \cdot M_x\right\|_2 \leq \frac{1}{\sqrt{|\cR|}}\]
    \end{lemma}
    \begin{proof}
    $ $
    \begin{enumerate}
        \item The $(r,y), (r',y')$-th entry of $M_x^\dag \cdot M_x$ is 
        \begin{align*}
            (M_x^\dag \cdot M_x)_{(r,y), (r',y')} &= \frac{1}{|\cR|} \cdot \braket{F_{x,r,y}|F_{x,r',y'}}\\
            &= \begin{cases}
                \frac{1}{|\cR|},& (r,y) = (r',y')\\
                0,& r = r' \land y \neq y'\\
                \frac{1}{|\cR|} \cdot \sqrt{\Pr_f[f(x,r)=y]} \cdot \sqrt{\Pr_f[f(x,r')=y']},& r \neq r'
            \end{cases}
        \end{align*}
        by \cref{thm:inner-products-of-F-register-states}.
    
        \item Let us write $M_x^\dag \cdot M_x$ as a linear combination of PSD matrices.
        
        First, let $\ket{\Phi_0}$ be a column vector such that the $(r,y)$-th entry is $\sqrt{\frac{\Pr_f[f(x,r)=y]}{|\cR|}}$. Additionally, let $\ket{\Phi_{0, r}}$ be a column vector such that the $(r',y)$-th entry is $0$ if $r \neq r'$, and $\sqrt{\Pr_f[f(x,r)=y]}$ if $r = r'$. 

        Note that these vectors have unit norm:
        \begin{align*}
            \|\ket{\Phi_0}\|_2^2 &= \sum_{r, y} \frac{\Pr_f[f(x,r)=y]}{|\cR|} = \sum_{r} \frac{\sum_y \Pr_f[f(x,r)=y]}{|\cR|} = \sum_{r} \frac{1}{|\cR|} = 1\\
            \|\ket{\Phi_{0, r}}\|_2^2 &= \sum_{y \in \cY} \Pr_f[f(x,r)=y] = 1
        \end{align*}
        

        Second, 
        \begin{align*}
            \text{let } N_x &= \sum_{r \in \cR} \frac{1}{|\cR|} \cdot \ketbra{\Phi_{0, r}}
        \end{align*}
        Note that $\left(\ket{\Phi_{0, r}}\right)_{r \in \cR}$ are mutually orthogonal, so $N_x$ is positive semi-definite.
        
        Third, we claim that $M_x^\dag \cdot M_x$ can be written in the following form:
        \begin{align*}
            M_x^\dag \cdot M_x &= \frac{1}{|\cR|} \cdot \mathbb{I} + \ketbra{\Phi_0} - N_x
        \end{align*}
        The $(r,y), (r',y')$-th entry of $\left(\frac{1}{|\cR|} \cdot \mathbb{I} + \ketbra{\Phi_0} - N_x\right)$ is given below. As shorthand, we let $p_{x,r,y} = \Pr_f[f(x,r)=y]$.
        \begin{align*}
            &\left(\frac{1}{|\cR|} \cdot \mathbb{I} + \ketbra{\Phi_0} - N_x\right)_{(r,y),(r',y')}\\
            &\quad\quad= \begin{cases}
                \frac{1}{|\cR|}\left(1 + p_{x,r,y} - p_{x,r,y}\right),& (r,y) = (r',y')\\
                \frac{1}{|\cR|}\left(0 + \sqrt{p_{x,r,y}} \cdot \sqrt{p_{x,r',y'}} - \sqrt{p_{x,r,y}} \cdot \sqrt{p_{x,r',y'}}\right),& r = r' \land y \neq y'\\
                \frac{1}{|\cR|}\left(0 + \sqrt{p_{x,r,y}} \cdot \sqrt{p_{x,r',y'}} - 0\right),& r \neq r'
            \end{cases}\\
            &\quad\quad= \begin{cases}
                \frac{1}{|\cR|},& (r,y) = (r',y')\\
                0,& r = r' \land y \neq y'\\
                \frac{1}{|\cR|} \cdot \sqrt{p_{x,r,y}} \cdot \sqrt{p_{x,r',y'}},& r \neq r'
            \end{cases}\\
            &\quad\quad= (M_x^\dag \cdot M_x)_{(r,y), (r',y')}
        \end{align*}
        Therefore, $M_x^\dag \cdot M_x = \frac{1}{|\cR|} \cdot \mathbb{I} + \ketbra{\Phi_0} - N_x$
        
        \item For any $\ket{\Phi_1}$ of unit norm that is orthogonal to $\ket{\Phi_0}$,
        \begin{align*}
            \|M_x \cdot \ket{\Phi_1}\|_2^2 &= \bra{\Phi_1} \cdot M_x^\dag \cdot M_x \cdot \ket{\Phi_1}\\
            &= \bra{\Phi_1} \cdot \frac{1}{|\cR|} \cdot \mathbb{I} \cdot \ket{\Phi_1} + \braket{\Phi_1|\Phi_0} \cdot \braket{\Phi_0|\Phi_1} - \bra{\Phi_1} \cdot N_x \cdot \ket{\Phi_1} = \frac{1}{|\cR|} - \bra{\Phi_1} \cdot N_x \cdot \ket{\Phi_1}\\
            &\leq \frac{1}{|\cR|}\\
            \|M_x \cdot \ket{\Phi_1}\|_2 &\leq \frac{1}{\sqrt{|\cR|}}
        \end{align*}
        We used the fact that $\bra{\Phi_1} \cdot N_x \cdot \ket{\Phi_1} \geq 0$ because $N_x$ is PSD.

        \item We will show that $M_x \cdot \ket{\Phi_0} = \ket{F_{\emptyset}}$.
            \begin{align*}
            M_x \cdot \ket{\Phi_0} &= \frac{1}{|\cR|} \cdot \sum_{(r,y) \in \cR \times \cY} \ket{F_{x,r,y}} \cdot \sqrt{\Pr_f[f(x,r)=y]}\\
            &= \frac{1}{|\cR|} \cdot \sum_{(r,y) \in \cR \times \cY} \sum_{f \in \cF_{x,r,y}} \ket{f} \cdot \frac{1}{\sqrt{|\cF_{x,r,y}|}} \cdot \sqrt{\frac{|\cF_{x,r,y}|}{|\cF|}} = \frac{1}{|\cR|} \cdot \sum_{(r,y) \in \cR \times \cY} \sum_{f \in \cF_{x,r,y}} \ket{f} \cdot \frac{1}{\sqrt{|\cF|}}\\
            &= \frac{1}{|\cR|} \cdot \sum_{r \in \cR} \sum_{f \in \cF} \ket{f} \cdot \frac{1}{\sqrt{|\cF|}} = \frac{1}{|\cR|} \cdot \sum_{r \in \cR} \ket{F_{\emptyset}}\\
            &= \ket{F_{\emptyset}}
        \end{align*}
    We used the fact that for any $(x,r)$, the sets $(\cF_{x,r,y})_{y \in \cY}$ partition $\cF$.
    \item Any vector $\ket{\Phi}$ of unit norm can be written as $\ket{\Phi} = \alpha \cdot \ket{\Phi_0} + \beta \cdot \ket{\Phi_1}$ for some vector $\ket{\Phi_1}$ of unit norm that is orthogonal to $\ket{\Phi_0}$ and for some $\alpha, \beta$ for which $|\alpha|^2 + |\beta|^2 = 1$. Next,
    \begin{align*}
        (\mathbb{I} - \ketbra{F_\emptyset}) \cdot M_x \cdot \ket{\Phi} &= (\mathbb{I} - \ketbra{F_\emptyset}) \cdot M_x \cdot \left(\alpha \cdot \ket{\Phi_0} + \beta \cdot \ket{\Phi_1}\right)\\
        &= \alpha \cdot (\mathbb{I} - \ketbra{F_\emptyset}) \cdot \ket{F_\emptyset} + \beta \cdot (\mathbb{I} - \ketbra{F_\emptyset}) \cdot M_x \cdot \ket{\Phi_1}\\
        &= \beta \cdot (\mathbb{I} - \ketbra{F_\emptyset}) \cdot M_x \cdot \ket{\Phi_1}\\
        \|(\mathbb{I} - \ketbra{F_\emptyset}) \cdot M_x \cdot \ket{\Phi}\|_2 &\leq |\beta| \cdot \|(\mathbb{I} - \ketbra{F_\emptyset})\|_2 \cdot \|M_x \cdot \ket{\Phi_1}\|_2\\
        &\leq 1 \cdot 1 \cdot \frac{1}{\sqrt{|\cR|}} = \frac{1}{\sqrt{|\cR|}}
    \end{align*}
    We've shown that for any vector $\ket{\Phi} \neq \vec{0}$, $\frac{\|(\mathbb{I} - \ketbra{F_\emptyset}) \cdot M_x \cdot \ket{\Phi}\|_2}{\|\ket{\Phi}\|_2} \leq \frac{1}{\sqrt{|\cR|}}$, so 
    \[\left\|(\mathbb{I} - \ketbra{F_\emptyset}) \cdot M_x\right\|_2 \leq \frac{1}{\sqrt{|\cR|}}\]
\end{enumerate}
\end{proof}

\begin{lemma}\label{thm:inner-products-of-F-register-states}
For any $(x,r,y), (x',r',y') \in \cX \times \cR \times \cY$,
\begin{itemize}
    \item $\braket{F_{\emptyset}|F_{x,r,y}} = \sqrt{\Pr_f[f(x,r)=y]}$
    \item If $(x,r) = (x',r')$ and $y \neq y'$, then $\braket{F_{x,r,y}|F_{x',r',y'}} = 0$.
    \item If $(x,r) \neq (x',r')$, then $\braket{F_{x,r,y}|F_{x',r',y'}} = \sqrt{\Pr_f[f(x,r)=y] \cdot \Pr_f[f(x',r')=y']}$.
\end{itemize}
\end{lemma}
\begin{proof}
\begin{align*}
    \braket{F_{\emptyset}|F_{x,r,y}} &= \left(\frac{1}{\sqrt{|\cF|}} \cdot \sum_{f \in \cF} \bra{f}\right) \cdot \left(\frac{1}{\sqrt{|\cF_{x,r,y}|}} \cdot \sum_{f' \in \cF_{x,r,y}} \ket{f'}\right)\\
    &= \frac{1}{\sqrt{|\cF| \cdot |\cF_{x,r,y}|}} \cdot \sum_{f \in \cF_{x,r,y}} \braket{f|f}\\
    &= \frac{|\cF_{x,r,y}|}{\sqrt{|\cF| \cdot |\cF_{x,r,y}|}} = \sqrt{\frac{|\cF_{x,r,y}|}{|\cF|}}\\
    &= \sqrt{\Pr_f[f(x,r)=y]}
\end{align*}

Next, if $(x,r) = (x',r')$, but $y \neq y'$, then $\cF_{x,r,y} \cap \cF_{x',r',y'} = \{\}$, so $\braket{F_{x,r,y}|F_{x',r',y'}} = 0$.

Finally, if $(x,r) \neq (x',r')$, then
\begin{align*}
    \braket{F_{x,r,y}|F_{x',r',y'}} &= \left(\frac{1}{\sqrt{|\cF_{x,r,y}|}} \cdot \sum_{f \in \cF_{x,r,y}} \bra{f}\right) \cdot \left(\frac{1}{\sqrt{|\cF_{x',r',y'}|}} \cdot \sum_{f' \in \cF_{x',r',y'}} \ket{f'}\right)\\
    &= \frac{1}{\sqrt{|\cF_{x,r,y}| \cdot |\cF_{x',r',y'}|}} \cdot \sum_{f \in \cF_{x,r,y} \cap \cF_{x',r',y'}} \braket{f|f}\\
    &= \frac{|\cF_{x,r,y} \cap \cF_{x',r',y'}|}{\sqrt{|\cF_{x,r,y}| \cdot |\cF_{x',r',y'}|}} = \frac{|\cF| \cdot \Pr_f[f(x,r)=y \land f(x',r')=y']}{\sqrt{|\cF| \cdot \Pr_f[f(x,r)=y] \cdot |\cF| \cdot \Pr_f[f(x',r')=y']}}\\
    &= \frac{\Pr_f[f(x,r)=y] \cdot \Pr_f[f(x',r')=y']}{\sqrt{\Pr_f[f(x,r)=y] \cdot \Pr_f[f(x',r')=y']}} = \sqrt{\Pr_f[f(x,r)=y] \cdot \Pr_f[f(x',r')=y']}
\end{align*}
\end{proof}


\section{Security Proof for \Cref{def:weak-operational-security} with Measure-and-Reprogram}
\label{sec:proof_measure_and_reprogram}

In this section, we show that if we aim at proving the weak operational one-time security definition \Cref{def:weak-operational-security} for random functions with superpolynomial range size, for the construction in \Cref{sec:construction}, we can have a simple proof using a technique called measure-and-reprogram lemma developed in \cite{don2019security}. The following proof actually is applicable to a class of functions we call $2$-replaceable below, which is more generic than truly random functions. 

\subsection{Preliminaries}
We review the measure-and-reprogram lemma developed in \cite{don2019security} and \cite{DFM20}. We adopt the formulation presented in in \cite[Section~4.2]{yamakawa2021classical}.

\begin{definition}[Reprogramming Oracle]
  Let $\mathcal{A}$ be a quantum algorithm that is given quantum oracle access to an oracle $\mathcal{O}$, where $\mathcal{O}$ is an oracle that is intialized to compute a classical function $f : \mathcal{X} \rightarrow \mathcal{Y}$ such that $\mathcal{A}$. At some point in an execution of $\mathcal{A}^\mathcal{O}$, we say that we reprogram $\mathcal{O}$ to output $g(x)$ on $x \in \mathcal{X}$ if we update the oracle to compute the function $f_{x, g}$ defined by
  \begin{align*}
    f_{x, g}(x') = \begin{dcases}
      g(x') &\text{if } x = x',\\
      f(x') &\text{otherwise}.\end{dcases}
  \end{align*}
  This updated oracle is used in the rest of execution of $\mathcal{A}^\mathcal{O}$. We denote the above reprogramming procedure as $\mathcal{O} \leftarrow \mathsf{Reprogram}(\mathcal{O}, x', g)$.
\end{definition}

\begin{definition}[Measure-and-Reprogram Algorithm]\label{definition:measure-and-reprogram}
    Let $\mathcal{X}, \mathcal{Y}, \mathcal{Z}$ be a set of classical strings and $k$ be a positive integer. Let $\mathcal{A}$ be a $q$-quantum-query algorithm that is given quantum oracle access to an oracle that computes a classical function $f: \mathcal{X} \rightarrow \mathcal{Y}$. The algorithm $\mathcal{A}$, when given a (possibly quantum) input $\mathsf{input}$, outputs $\vecx \in \mathcal{X}^k$ and $z \in \mathcal{Z}$. For a function $g : \mathcal{X} \rightarrow \mathcal{Y}$, we define a measure-and-reprogram algorithm $\tilde{\mathcal{A}}[f, g]$ as follows:\\

    \noindent
    $\underline{\tilde{\mathcal{A}}[f, g](\mathsf{input})}$
  \begin{enumerate}
    \item For each $j \in [k]$, uniformly pick $(i_j, b_j) \in ([q] \times \{0,1\}) \cup \{(\bot, \bot)\}$ such that there does not exist $j \neq j'$ such that $i_j = i_{j'} \neq \bot$.
    \item Run $\mathcal{A}^{\mathcal{O}}(\mathsf{input})$ where the oracle $\mathcal{O}$ is initialized to be a quantumly-accessible classical oracle that computes the classical function $f$. When $\mathcal{A}$ makes its $i$th query, the oracle is simulated as follows:
    \begin{enumerate}
        \item If $i = i_j$ for some $j \in [k]$, measure $\mathcal{A}$'s query register to obtain $x'_j$, and do either of the following:
        \begin{enumerate}
            \item If $b_j = 0$, reprogram $\mathcal{O} \leftarrow \mathsf{Reprogram}(\mathcal{O}, x'_j, g)$ and answer $\mathcal{A}$'s $i$th query using the reprogrammed oracle.
            \item If $b_j = 1$, answer $\mathcal{A}$'s $i$th query by using the oracle before the reprogramming and then reprogram $\mathcal{O} \leftarrow \mathsf{Reprogram}(\mathcal{O}, x'_j, g)$.
        \end{enumerate}
        \item Otherise, answer $\mathcal{A}$'s $i$the query by just using the oracle $\mathcal{O}$ without any measurement or reprogramming.
    \end{enumerate}
    \item Let $(\vecx = (x_1, \ldots, x_k), z)$ be $\mathcal{A}$'s output.
    \item For all $j \in [k]$ such that $i_j = \bot$, set $x'_j := x_j$.
    \item Output $(\vecx', z)$, where $\vecx' := (x'_1, \ldots, x'_k)$.
  \end{enumerate}
\end{definition}

\begin{lemma}\label{lemma:measure-and-reprogram}
    Let $\mathcal{X}, \mathcal{Y}, \mathcal{Z}$, and $\mathcal{A}$ be as in Definition~\ref{definition:measure-and-reprogram}. Then, for any input $\mathsf{input}$, functions $f, g: \mathcal{X} \rightarrow \mathcal{Y}$, $\vecx^* \in \mathcal{X}^k$ such that $x^*_j \neq x^*_{j'}$ for all $j \neq j'$, and relation $R \subseteq \mathcal{X}^k \times \mathcal{Y}^k \times \mathcal{Z}$, we have
    \begin{align*}
        \Pr \left[
        \begin{array}{c}
        \vecx' = \vecx^*\ \land\\
        (\vecx', g(\vecx'), z) \in R
        \end{array}  :
        (\vecx', z) \leftarrow \tilde{\mathcal{A}}[f, g](\mathsf{input})
        \right] \ge \frac{1}{2q + 1}^{2k} \Pr \left[\begin{array}{c}
        \vecx = \vecx^*\ \land\\
        (\vecx, g(\vecx'), z) \in R
        \end{array} :
        (\vecx, z) \leftarrow \mathcal{A}^{|f_{\vecx^*, g}\rangle}(\mathsf{input}) \right],
    \end{align*}
where $\tilde{\mathcal{A}}[f, g]$ is the measure-and-reprogram algorithm as defined in Definition~\ref{definition:measure-and-reprogram}, and $f_{\vecx^*, g}$ is defined as
\begin{align*}
    f_{\vecx^*, g}(x') :=
    \begin{dcases}
        g(x') &\text{if } \exists j \in [k] \text{ s.t. } x' = x^*_j,\\
        f(x') &\text{otherwise.}
    \end{dcases}
\end{align*}
\end{lemma}


\begin{definition}[$k$-wise replaceable]
    A function family $\mathcal{F} = \{f: \mathcal{X} \rightarrow \mathcal{Y}\}$ is $k$-wise replaceable if the following holds.
    For all $\vecx \in \mathcal{X}^k$, the distribution over functions $f'_{\vecx, f(\vecx)}$ obtained by sampling $f, f' \leftarrow \mathcal{F}$ is the same as the distribution of functions obtained by sampling $f \leftarrow \mathcal{F}$.
\end{definition}

A random function with superpolynomial range size is $k$-wise replaceable as shown in \cite{yamakawa2021classical}.

\subsection{Security Proof}

\begin{theorem}
\label{thm:weak_op_security_proof}
Let $\mathcal{F}$ be a function family that is $1$-query unlearnable \Cref{def:single-query-unlearnable} and $2$-wise replaceable. Then the scheme in \Cref{fig:otp-construction} satisfies the one-time security of \Cref{def:weak-operational-security}.
\end{theorem}

\begin{proof}
We first prove one-time secrecy for what we call a "mini-scheme" where the user's input is one bit, i.e. $\mathcal{X} = \{0,1\}$. We will call this user input $b \in \{0,1\}$ for "bit".

Suppose for contradiction that there exists a (quantum) polynomial time algorithm $\calA$ that takes as input one-time sampling program $\mathcal{P}_f = (\ket{A}, \mathcal{O}_A, \mathcal{O}_{A^\perp}, \mathcal{O}_{f, A})$ and breaks one-time secrecy with non-negligible probability. Note that we can consider $\calA$ as having access to oracles $f_A$, which compute $f$ restricted on valid inputs:
\begin{align*}
  f_A(b, u) = \begin{dcases}
    f(b, H(u)) &\text{if }u \in A^b,\\
    \bot &\text{otherwise}.\end{dcases}
\end{align*}
We will use $\calA$ to construct a (quantum) polynomial time algorithm $\calB$ that either (1) breaks the one-query unlearnability of $f$, or (2) violates the direct product hardness for random subspaces.\\

\noindent
$\underline{\calB(1^\lambda, \calP_f)}$
\begin{enumerate}
    \item Sample a random function $g \leftarrow \calF$.
    \item Run the measure-and-reprogram algorithm $\tilde{\calA}[g, f_A]$ as defined in Definition~\ref{definition:measure-and-reprogram} with $k = 2$.
    \item Output the output of $\tilde{\calA}[g, f_A]$.
\end{enumerate}
Using the shorthand $\vecx = (x_1, x_2) \in (\{0,1\} \times \calR)^2$ where $x_1 = (b_1, r_1)$ and $x_2 = (b_2, r_2)$ and $\vecy = (y_1, y_2) \in \calY^2$, define the relation $R \subseteq (\{0,1\} \times \calR)^2 \times \calY^2$ as
\begin{align*}
    R = \left\{ (\vecx, \vecy) \mid f(x_1) = y_1, f(x_2) = y_2 \right\}.
\end{align*}
By Lemma~\ref{lemma:measure-and-reprogram}, we have that for all $\vecx^* \in (\{0,1\} \times \calR)^2$ such that $x^*_1 \neq x^*_2$,
\begin{gather*}
    \Pr_{f \leftarrow \calF}\left[
        (\vecx, \vecy) \in R \ \land
        \vecx = \vecx^* :
    (\vecx, \vecy) \leftarrow \calB(1^\lambda, \calP_f)
    \right] \\
    \ge \frac{1}{(2q+1)^4} \Pr_{f, g \leftarrow \calF}\left[
        (\vecx, \vecy) \in R \ \land
        \vecx = \vecx^* :
    (\vecx, \vecy) \leftarrow \calA^{|g_{\vecx^*, f}\rangle}(1^\lambda, \ket{A}, \calO_A, \calO_{A^\perp})
    \right].
\end{gather*}
Since $\calF$ is $2$-wise replaceable, the above probability expression is
\begin{gather*}
    = \frac{1}{(2q+1)^4} \Pr_{f \leftarrow \calF}\left[
        (\vecx, \vecy) \in R \ \land
        \vecx = \vecx^* :
    (\vecx, \vecy) \leftarrow \calA^{|f\rangle}(1^\lambda, \ket{A}, \calO_A, \calO_{A^\perp})
    \right] \\
    = \frac{1}{(2q+1)^4} \Pr_{f \leftarrow \calF}\left[
        (\vecx, \vecy) \in R \ \land
        \vecx = \vecx^* :
    (\vecx, \vecy) \leftarrow \calA(1^\lambda, \calP_f)
    \right] \ge \mathsf{non\text{-}negl}(\lambda),
\end{gather*}
where the last inequality holds because $\calA$ breaks the one-time security of $\calP_f$, by assumption. Note that $\calB$ makes at most $2$ queries, both of which are classical, to $\calO_{f, A}$. Consider the following two cases:

\paragraph{Case 1.} $\calB$ makes two valid and distinct classical queries $x_1 = (b_1, r_1), x_2 = (b_2, r_2) \in \{0,1\} \times \calR$ to $\calO_{f, A}$ such that $\calO_{f, A}(x_1) \neq \bot$ and $\calO_{f, A}(x_2) \neq \bot$. This means that $r_1 \in A^{b_1}$ and $r_2 \in A^{b_2}$. Since these are classical queries, they can be recorded, and this immediately gives us an adversary that breaks the direct product hardness of subspace states (\Cref{thm: direct product oracle}).

\paragraph{Case 2.} $\calB$ makes at most one valid classical queries $x = (b, r) \in \{0,1\} \times \calR$ to $\calO_{f, A}$ such that $\calO_{f, A}(x) \neq \bot$. This is impossible because it would break the $1$-query unlearnability of $\calF$.
\end{proof}

\begin{remark}
    We conjecture that the above proof extends to function families larger than uniformly random functions (which are shown to be $k$-wise replaceable in \cite{yamakawa2021classical}. We will leave to future works on how to characterize the $k$-wise replaceable property. 
\end{remark}

\section{More on the Compressed Oracle}

\subsection{Alternative Representations of the Decompression Function}\label{sec:decompression-alternative-form}

Let $D$ be a database for a random function $H$ such that $D(x) = \bot$. If one were to query $H$ and receive result $H(x) = y$, the state of the oracle before the final decompression operation is $\ket{D \cup (x, y)}$. The effect of the $x$ decompression operation $\mathsf{Decomp}_x$ on a state $\ket{D \cup (x, y)}$ is
\begin{align*}
    &\mathsf{Decomp}_x \ket{D\cup (x, y)}
    \\
    &= \mathsf{Decomp}_x \frac{1}{|\calY|} \sum_{y'} \ket{D\cup (x, y')} \sum_{r\in \calY}(-1)^{r\cdot (y'-y)} 
    \\
    &= \mathsf{Decomp}_x  \frac{1}{|\calY|}\left(\sum_{r\neq 0} (-1)^{r\cdot -y} \sum_{y'}(-1)^{r\cdot y'} \ket{D\cup (x, y')} + \sum_{y'} \ket{D\cup (x, y')} \right)
    \\
    &=  \frac{1}{|\calY|} \left( \sum_{r\neq 0} (-1)^{r\cdot -y} \sum_{y'}(-1)^{r\cdot y'} \ket{D\cup (x, y')} + \sqrt{|\calY|} \ket{D} \right)
    \\
    &=  \left(\ket{D\cup (x, y)} - \frac{1}{|\calY|} \sum_{y'} \ket{D\cup (x, y')} + \frac{1}{\sqrt{|\calY|}} \ket{D} \right)
    \\
    &=  \left(\left(1-\frac{1}{|\calY|} \right)\ket{D\cup (x, y)} - \frac{1}{|\calY|} \sum_{y'\neq y} \ket{D\cup (x, y')} + \frac{1}{\sqrt{|\calY|}} \ket{D} \right)
\end{align*}

We note that this only describes the state \emph{in between} queries to $x$. A second query to $x$ will apply $\mathsf{Decomp}_x$ to the database register before determining the response, which maps the database register back to $\ket{D\cup (x,y)}$. Thus, despite the support of the state on other databases, $G(x)$ cannot actually change in between queries to $x$. However, this form is relevant when looking at the database register without knowledge of $x$.


\subsection{Compressed Oracle Chaining}\label{sec:compressed-chaining}

In this section, we show that if an adversary has access to the composition of compressed oracles $H\circ G$ and $H$ records some input/output pair $H(y) = z$, then with overwhelming probability $G$ also records a matching input/output pair $G(x) = y$.

Following almost the same analysis, we will also show that if an adversary has access to a function $F(x) = H(x, G(x))$ and $H$ records $H(x\concat y) = z$, then with overwhelming probability $G$ contains a matching entry $G(x) = y$. See \Cref{lemma:compressed-chaining-carry} at the end of this subsection for more details.

\begin{lemma}[Compressed Oracle Chaining]\label{lemma:compressed-chaining}
    Let $G:\calX \rightarrow \calY$ and $H:\calY \rightarrow \calZ$ be random oracles implemented by the compressed oracle technique. Consider running an interaction of an oracle algorithm with their composition $H\circ G$ until query $t$, then measuring the internal state of $G$ and $H$ to obtain $D_G$ and $D_H$. 

    Let $E_t$ be the event that after the measurement at time $t$, for all $(y, z)\in D_H$, there exists an $x\in \calX$ such that $(x, y) \in D_G$.
    Then
    \[
        \Pr[E_t]
        \geq 
        1 - 4t^2\left(\frac{2}{|\calY|} - \frac{1}{|\calY|^2}\right)
    \]
\end{lemma}
\begin{proof}
    We first explicitly describe how $H\circ G$ works. The internal states of $H$ and $G$ are stored in registers $\calD_H$ and $\calD_G$, respectively. 
    In general, we can consider a basis query $\ket{x, u}_{\calQ}$ in register $\calQ = \calQ_\calX, \calQ_\calZ$. For convenience, we insert into $\calQ$ an additional work register $\calQ_\calY$ which is initialized to $\ket{0}$, resulting in $\ket{x, 0, u}_{\calQ}$. 
    To compute $H\circ G$ on this query, first query $(\calQ_\calX, \calQ_\calY)$ to $G$ to obtain $\ket{x, G(x), u}_{\calQ}$, i.e. apply $\Decomp_G\circ \CO'_G \circ \Decomp_G$ on $(\calQ_\calX, \calQ_\calY, \calD_G)$.\footnote{Since $G$ and $H$ have different domains and ranges, we differentiate their decompression operations, which depend on the domain/range. We also differentiate the $\CO'_G$ operation to clarify that it acts on registers corresponding to a query to $G$.} 
    Then query $(\calQ_\calY, \calQ_\calZ)$ to $H$ to obtain $\ket{x, G(x), u\oplus H(G(x))}$. Finally, query $(\calQ_\calX, \calQ_\calY)$ to $G$ again to obtain $\ket{x, 0, u\oplus H(G(x))}$ and return registers $(\calQ_\calX, \calQ_{\calZ})$.

    Observe that $\Decomp_G$ commutes with the query to $H$, since they operate on disjoint registers $(\calQ_\calX, \calD_G)$ and $(\calQ_\calY, \calQ_\calZ, \calD_G)$, respectively. Furthermore, $\Decomp_G \circ \Decomp_G = I$. Thus, if we write the operation of $H$ as $U_H$, we may write the implementation of a query to $H\circ G$ as
    \[
        U_{H\circ G} 
        \coloneqq
        (\Decomp_G \circ \CO'_G) \circ (U_H) \circ (\CO'_G \circ \Decomp_G)
    \]
    where $\Decomp_G$ and $\CO'_G$ act on registers $(\calQ_\calX, \calQ_\calY, \calD_G)$, while $U_H$ acts on registers $(\calQ_\calY, \calQ_\calZ, \calD_H)$. 

    Now consider the interaction of the algorithm with $H\circ G$, where the algorithm maintains an additional internal register $\calA$.
    Define the projector $E$ onto states $\ket{a}_{\calA}\otimes \ket{x,0, u}_{\calQ}\otimes \ket{D_G, D_H}_{\calD}$ where $D_G$ and $D_H$ satisfy the requirements of event $E_t$ and define $\overline{E} = I - E$. 
    We will upper bound the norm $\norm{\overline{E} U_{H\circ G} \ket{\psi}}$ from after the query in terms of the norm $\norm{\overline{E}\ket{\psi}}$ from before the query. To do this, we will individually bounding the norm after each step in terms of the norm before that step, e.g. bound $\norm{\overline{E} \Decomp_G \ket{\psi'}}$ from after the query in terms of the norm $\norm{\overline{E} \ket{\psi'}}$. 

    We first bound the intermediate operations, in between the two $\Decomp_G$ operations.


    \begin{claim}\label{claim:chaining-co-H}
        \[
            \norm{\overline{E} (\CO'_G) \cdot (U_H) \cdot (\CO'_G \cdot \Decomp_G) \ket{\psi}} 
            = 
            \norm{\overline{E} (\Decomp_G) \ket{\psi}}
        \]
    \end{claim}
    \begin{proof}
        First, observe that for all states $\ket{\psi'}$, 
        \[
            \norm{\overline{E} \cdot \CO'_G \ket{\psi'}} 
            = 
            \norm{\overline{E}  \ket{\psi'}}
        \]
        since $\CO'_G$ does not modify the database registers.
    
        Now consider the operation of $U_H$ on $(\CO' \cdot \Decomp_G) \ket{\psi}$. Observe that by the definition of $\CO'_G$, the state of the query register and $H\circ G$ when $H$ is queried is always supported on basis states
        \[
            \ket{x, y, u}_{\calQ} \otimes \ket{D_G\cup (x, y)}_{\calD_G} \otimes \ket{D_H}_{\calD_H}
        \]
        A query of $y$ to $H$ can only result in a new database $D_H'$ where the only potential difference between $D_H$ and $D_H'$ is $D_H(y) \neq D_H'(y)$. Since $(x, y)$ is already recorded in the contents of register $\calD_G$, applying $H$ always results in a valid database state with respect to the event $E$. 
        In particular, if $D_H$ and $D_G$ already satisfied $E$ before the query to $H$, they continue to do so afterwards. Thus
        \[
            \norm{\overline{E} (U_H) \cdot (\CO'_G \cdot \Decomp_G) \ket{\psi}} 
            = 
            \norm{\overline{E} (\CO'_G \cdot \Decomp_G) \ket{\psi}}
        \]
        Putting this together with the prior bound on the effect of $\CO'_G$ yields the claim.
    \end{proof}

    Next, we bound the effect of $\Decomp_G$ on a general state $\ket{\psi'}$. In general, both the first and the last $\Decomp_G$ operation will operate on a state of the form
    \[
        \sum_{a,x,u,D_G,D_H} \alpha_{a,x,u,D_G,D_H} \ket{a}_{\calA} \otimes \ket{x, 0, u}_{\calQ} \otimes \ket{D_G, D_H}_{\calD}
    \]
    where $\calA$ is the adversary's internal register, $\calQ$ contains the (expanded) query, and $\calD$ contains the oracle databases. This is clearly true for the first application; it holds for the second because $\CO_G'$ is applied twice and $U_H$ does not modify the registers that $\CO_G'$ acts on.\footnote{We note that $\CO_G'$ and $U_H$ do not commute, despite this, since $U_H$ does certain actions controlled on the registers that $\CO_G'$ acts on.}

    \begin{claim}\label{claim:chaining-decomp}
        \[
            \norm{\overline{E} \Decomp_G \ket{\psi'}} 
            \leq  
            \norm{\overline{E} \ket{\psi}} + \sqrt{\frac{2}{|\calY|} - \frac{1}{|\calY|^2}}
        \]
    \end{claim}
    \begin{proof}
        Define the following projectors on states $\ket{a}\otimes \ket{x, 0, u} \otimes \ket{D_G, D_H}$:
        \begin{itemize}
            \item $E_{Y1, Z}$ projects onto states where (1) $D_G$ and $D_H$ satisfy $E$, (2) $D_G(x) = y$ for some $y$, (3) $D_H(y) = z$ for some $z$, and (4) $x$ is the unique preimage of $y$ under $G$, i.e. $D_G(x') \neq y$ for all $x\neq x'$.
            
            \item $E_{Y+, Z}$ projects onto states where (1) $D_G$ and $D_H$ satisfy $E$, (2) $D_G(x) = y$ for some $y$, (3) $D_H(y) = z$ for some $z$, and (4) there exists at least one $x'\neq x$ such that $D_G(x') = y$.
            
            \item $E_{Y, \bot}$ projects onto states where (1) $D_G$ and $D_H$ satisfy $E$, (2) $D_G(x) = y$ for some $y$, and (3) $D_H(y) = \bot$.
            \item $E_{\bot}$ projects onto states where (1) $D_G$ and $D_H$ satisfy $E$, (2) $D_G(x) = \bot$.
        \end{itemize}
        Observe that $E = E_{Y1,Z} + E_{Y+,Z} + E_{Y,\bot} + E_{\bot}$, so $I = \overline{E} + E_{Y1,Z} + E_{Y+,Z} + E_{Y,\bot} + E_{\bot}$. Furthermore, $\Decomp_G$ only modifies register $\calD_G$, where it maps $D_G$ to $D_G'$, with the only potential difference that $D_G(x) \neq D_G'(x)$.
        In the case where $x$ is not part of an $(x, y)\in D_G$ and $(y, z)\in D_H$ pair mandated by $E$, this modification does not affect the containment of the state in space $E$. Thus 
        \[
            \norm{\overline{E} \cdot \Decomp_G \cdot E_{Y, \bot}\ket{\psi'}}
            = 
            \norm{\overline{E} \cdot \Decomp_G \cdot E_{\bot}\ket{\psi'}}
            = 
            0
        \]
        Similarly, because there is a ``backup'' option $D_G(x') = y$ for $y$ in the case $E_{Y+, Z}$,
        \[
            \norm{\overline{E} \cdot \Decomp_G \cdot E_{Y+,Z}\ket{\psi'}}
            = 
            0
        \]

        For basis vectors in the support of $E_{Y1,Z}$, we can write $D_G = D_G' \cup (x,y)$ and $D_H = D_H' \cup (y,z)$, where $D_G'(x) = \bot$ and $D_H'(y) = \bot$. $D_G'$ and $D_H'$ represent $D_G$ and $D_H$ with $x$ and $y$ removed, respectively. Furthermore, there does not exist an $x'\in \calX$ such that $D_G'(x') = y$, since $x$ is the unique preimage of $y$ under $G$.
        Recall from \Cref{sec:decompression-alternative-form} that the effect of $\Decomp_G$ on $\ket{x,0, D\cup(x,y)}$ is
        \[
            \ket{x, 0} \otimes \left(\left(1-\frac{1}{|\calY|} \right)\ket{D\cup (x, y)} - \frac{1}{|\calY|} \sum_{y'\neq y} \ket{D\cup (x, y')} + \frac{1}{\sqrt{|\calY|}} \ket{D} \right)
        \]
        Since $D_H(y) \neq \bot$, the projection $\overline{E} \cdot \Decomp_G \ket{a}_{\calA}\otimes\ket{x,0,u}_{\calQ} \otimes \ket{D_G'\cup(x,y), D_H'\cup(y,z)}_{\calD}$ is
        \[
            \ket{a, x,0, D_H'\cup(y,z)}_{\calA,\calQ,\calD_H} \otimes \left(-\frac{1}{|\calY|} \sum_{y'\neq y} \ket{D_G'\cup (x, y')} + \frac{1}{\sqrt{|\calY|}} \ket{D_G'}\right)_{\calD_G}
        \]
        We can write $E_{Y1,Z} \ket{\psi'}$ in general as
        \[
        E_{Y1,Z} \ket{\psi'} = \sum_{a,x,u,y}\sum_{D'_G\not\ni y, D_H\ni y} \alpha_{a,x,u,y,D_G',D_H} \ket{a}\otimes \ket{x,0,u}\otimes \ket{D_G'\cup (x,y), D_H}
        \]
        Then we can compute
        \begin{align*}
            \norm{\overline{E}\cdot \Decomp_G \cdot E_{Y1,Z} \ket{\psi'}}
            &= 
            \sum_{\substack{a,x,u,y\\D_G'\not\ni y,D_H \ni y}} \norm{\alpha_{a,x,u,D_G',D_H} \left(-\frac{1}{|\calY|} \sum_{y'\neq y} \ket{D_G'\cup (x, y')} + \frac{1}{\sqrt{|\calY|}} \ket{D_G'}\right)}
            \\
            &=\sum_{\substack{a,x,u,y\\D_G'\not\ni y,D_H \ni y}} \norm{\alpha_{a,x,u,D_G',D_H}}\sqrt{\frac{|\calY|- 1}{|\calY|^2}  + \frac{1}{|\calY|}}
            \\
            &= \norm{E_{Y1,Z} \ket{\psi'}} \sqrt{\frac{2}{|\calY|} - \frac{1}{|\calY|^2}}
        \end{align*}

        Finally, $\norm{\overline{E} \cdot \Decomp_G \cdot \overline{E} \ket{\psi}} \leq \norm{\overline{E} \ket{\psi}}$. Putting these bounds together with the decomposition $\ket{\psi'} = \overline{E}\ket{\psi'} + E_{Y1,Z}\ket{\psi'} + E_{Y+,Z}\ket{\psi'} + E_{Y,\bot}\ket{\psi'} + E_{\bot}\ket{\psi'}$, we have
        \begin{align*}
            \norm{\overline{E} \cdot \Decomp_G \ket{\psi}} 
            &\leq \norm{\overline{E} \ket{\psi}} + \norm{E_{Y1,Z} \ket{\psi'}} \sqrt{\frac{2}{|\calY|} - \frac{1}{|\calY|^2}} + 0 + 0 + 0
            \\
            &\leq \norm{\overline{E} \ket{\psi}} + \sqrt{\frac{2}{|\calY|} - \frac{1}{|\calY|^2}}
        \end{align*}
    \end{proof}

    Putting together \Cref{claim:chaining-co-H} and \Cref{claim:chaining-decomp}, the norm after any single query to $H\circ G$ is bounded as
    \begin{align*}
        \norm{\overline{E} \cdot U_{H\circ G}\ket{\psi}}
        &\leq \norm{\overline{E} \cdot (\CO'_G) \cdot (U_H) \cdot (\CO' \cdot \Decomp_G) \ket{\psi}} + \sqrt{\frac{2}{|\calY|} - \frac{1}{|\calY|^2}}
        \\
        &= \norm{\overline{E} \Decomp_G \ket{\psi}} + \sqrt{\frac{2}{|\calY|} - \frac{1}{|\calY|^2}}
        \\
        &\leq \norm{\overline{E} \ket{\psi}} + 2\sqrt{\frac{2}{|\calY|} - \frac{1}{|\calY|^2}}
    \end{align*}
    The norm starts at $0$, so after $t$ queries to $H\circ G$ it is at most $2t\sqrt{\frac{2}{|\calY|} - \frac{1}{|\calY|^2}}$. The probability of seeing the event corresponding to $\overline{E}$ when we measure $\calD$ is the square of the norm, which is at most $4t^2\left(\frac{2}{|\calY|} - \frac{1}{|\calY|^2}\right)$. Therefore the probability of seeing the complementary event $E$ is at least $1-4t^2\left(\frac{2}{|\calY|} - \frac{1}{|\calY|^2}\right)$, as claimed.
\end{proof}

\begin{lemma}\label{lemma:compressed-chaining-carry}


    Let $G:\calX_G \rightarrow \calY$ and $H:\calX_H \times \calY \rightarrow \calZ$ be random oracles implemented by the compressed oracle technique. Let $\calX\subset \calX_G \times \calX_H$. Define the function $F: \calX \rightarrow \calZ$ by $F(x_g, x_h) = H(x_g, G(x_h))$.
    Consider running an interaction of an oracle algorithm with $F$ until query $t$, then measuring the internal state of $G$ and $H$ to obtain $D_G$ and $D_H$. 

    Let $E_t$ be the event that after the measurement at time $t$, for all $(x_G\concat y, z)\in D_H$, there exists a entry $(x_H, y) \in D_G$.
    Then
    \[
        \Pr[E_t]
        \geq 
        1 - 4t^2\left(\frac{2}{|\calY|} - \frac{1}{|\calY|^2}\right)
    \]
\end{lemma}
\begin{proof}


    The proof requires changing just a few lines of the proof of \Cref{lemma:compressed-chaining}, mostly for syntactic reasons. The first modification is to expand the query register to $\calQ = (\calQ_{\calX_G}, \calQ_{\calX_H}, \calQ_\calY, \calQ_\calZ)$. $U_H$ acts on registers $\calQ = (\calQ_{\calX_H}, \calQ_\calY, \calQ_\calZ)$, although it only modifies $\calQ_{\calZ}$. This change does not affect the proof of \Cref{claim:chaining-co-H}.

    Second, in the proof of \Cref{claim:chaining-decomp}, the projectors are slightly modified. $E_{Y, \bot}$ now requires $D_H(x_h\concat y) = \bot$, instead of $D_H(y) = \bot$. Additionally, $E_{Y1, Z}$ and $E_{Y+, Z}$ are similarly syntactically modified to consider $D_H(x_h\concat y)$ instead of $D_H(y)$ in condition 3. The analysis of each projector is the same, except for additional syntactic changes.
    
\end{proof}


\subsection{Knowledge of Preimage for Expanding Random Oracles}\label{sec:expanding-preimage-knowledge}

Here we show that if a random oracle $G$ is sufficiently expanding and an adversary making polynomially many queries to $G$ knows elements $y$ that appear in the image of $G$, then $G$'s compressed database should also contain a corresponding entry. In fact, we will show a more general statement which takes into account an adversary which knows partial preimages as well. In the case where $\calX_1 = \emptyset$ in the following lemma, this corresponds to an adversary simply finding elements from the range.

\begin{lemma}\label{lemma:expanding-preimage-knowledge}
    Let $G: \calX_1 \times \calX_2 \rightarrow \calY$ be a random function where $|\calX_2| < |\calY|$. Consider an oracle algorithm $A$ makes $q$ of queries to $G$, then outputs two vector of $k$ values $\vec{x^{(1)}} = \left(x_1^{(1)}\dots, x_k^{(1)}\right)$ and $\vec{y} = (y_1, \dots, y_k)$.
    Let $p$ be the probability that for every $i$, there exists an $x_i^{(2)}\in \calX$ such that $G\left(x_i^{(1)}, x_i^{(2)}\right) = y_i$.

    Now consider running the same experiment where $G$ is instead implemented as a compressed oracle, and measuring its database register after $A$ outputs to obtain $D$. Let $p'$ be the probability that for every $i$, there exists an $x_i^{(2)}\in \calX_2$ such that $D\left(x_i^{(1)}\concat x_i^{(2)} \right) = y_i$. If $k$ and $q$ are $\poly(\secpar)$ and $|\calX_2|^k/|\calY| = \negl(\secpar)$, then\footnote{We remark that the reliance on the number of queries is unlikely to be tight. A tighter bound might be achieved by performing a direct computation of the effects of querying $G$ on every $x\in \calX$ at the end of the experiment.}
    \[
        p \leq p' + \negl(\secpar)
    \]
\end{lemma}
\begin{proof}
    Consider the adversary $B$ which attempts to find $k$ input/output pairs by running $A$ to obtain $\vec{x^{(1)}}$ and $\vec{y}$, then guessing a uniform $\vec{x^{(2)}}\in \calX_2^k$ to construct the vector $\vec{x} = \left(x_1^{(1)}\concat x_1^{(2)}, \dots, x_k^{(1)}\concat x_k^{(2)}\right)$ and outputting $(\vec{x}, \vec{y})$. Denote $p_B$ as the probability that it outputs $(\vec{x}, \vec{y})$ such that $G(x_i^{(1)}, x_i^{(2)}) = y_i$ for all $i$. Since $\vec{x^{(2)}}$ is independent of $A$, we have $p_B \geq p/|\calX_2|^k$. 

    Now consider running $B$ with a compressed oracle, then measuring the compressed oracle to obtain a database $D$. Note that this is the same distribution over databases as running $A$. Denote $p'_B$ as the probability that it successfully outputs $(\vec{x}, \vec{y})$ such that $(x_i^{(1)}\concat x_i^{(2)}, y_i)\in D$. We may decompose the corresponding event into two mutually exclusive components: 
    \begin{itemize}
        \item Let $E'_{B,+}$ be the event that $(x_i^{(1)}\concat x_i^{(2)}, y_i)\in D$ for all $i$ and $D$ contains a collision $(x^*_0, y^*)$ and $(x^*_1, y^*)$.
        \item Let $E'_{B,1}$ be the event that $(x_i^{(1)}\concat x_i^{(2)}, y_i)\in D$ for all $i$ and $D$ does not contain a collision.
    \end{itemize}
    By definition, $p'_B = \Pr[E'_{B,+}] + \Pr[E'_{B,1}]$. We may define analogous events $E'_{A,+}$ and $E'_{A,1}$ when $A$ is run and the database is measured, where the first condition is changed to the existence of $x_i$, rather than requiring $A$ to find it. 
    $\Pr[E'_{B,+}]$ and $\Pr[E'_{A,+}]$ are both bounded by the probability of $D$ containing a collision, which \Cref{lem:compressed-collision} bounds by $O(q^3/|\calY|)$. Furthermore, since there is a unique ``solution'' to $\vec{y}$ in the event $E'_{A,1}$, we have $\Pr[E'_{B,1}] = \Pr[E'_{A,1}]/|\calX_2|^k$. Thus we can relate $p'_A$ and $p'_B$ by $p'_B \leq O(q^3/|\calY|) + \Pr[E'_{A,1}]/|\calX_2|^k \leq O(q^3/|\calY|) + p'_A/|\calX_2|^k$.

    By \Cref{lem:zhandry_lemma5},
    \[
        \sqrt{p_B} \leq \sqrt{p'_B} + \sqrt{k/|\calY|}
    \]
    Combining these with our bounds on $p_B$ and $p'_B$, we have
    \begin{align*}
        \sqrt{p_A/|\calX|^k} &\leq \sqrt{O(q^3/|\calY|) + p'_A/|\calX|^k} + \sqrt{k/|\calY|}
        \\
        \sqrt{p_A} &\leq \sqrt{O(q^3\cdot |\calX|^k/|\calY|) + p'_A} + \sqrt{k\cdot |\calX|^k/|\calY|}
    \end{align*}
    Since $k$ and $q$ are polynomial in $\secpar$ and $|\calX|^k/|\calY| = \negl$, we obtain the desired result by squaring both sides of the inequality.
\end{proof}

\section{Additional Prelims}
\label{sec:appendix_prelims}

We give some additional preliminaries in this section.

\subsection{NIZK}
\label{sec:nizk_def}

A NIZK for NP scheme should satisfy the following properties:

\paragraph{Correctness}
A NIZK proof $(\setup, \delegate, \prove, \Verify)$ is correct if there exists
a negligible function $\negl(\cdot)$ such that for all $\lambda \in \N$, all $x \in L$, and all $w \in \mathcal{R}_L(x)$ it holds that
$$\Pr[\Verify(\crs, \prove(\crs, \mathsf{token},\ w, x), x) = 1 ] = 1 - \negl(\lambda)$$
where $(\crs, \msk) \gets \setup(1^\lambda), \mathsf{token} \gets \delegate(\msk)$.

\paragraph{Computational Soundness}
A one-time NIZK proof $(\setup, \prove, \Verify)$ is computationally sound if there exist a negligible function $\negl(\cdot)$ such that for all unbounded adversaries $\calA$ and all $x \notin L$, it holds that:
$$\Pr[\Verify(\crs, \pi \gets \calA(\crs, x, \mathsf{token}), x) = 1] = \negl(\lambda)$$

where $(\crs, \msk) \gets \setup(1^\lambda), \mathsf{token} \gets \delegate(\msk)$.

\paragraph{Computational Zero Knowledge}
A one-time NIZK proof $(\setup, \delegate, \prove, \Verify)$ is computationally zero-knowledge if there exists a simulator $S$ such that for all non-uniform QPT adversaries with quantum advice $\{\rho_\lambda\}_{\lambda \in \N}$, all statements $x \in L$ and all witnesses $w \in \mathcal{R}_L(x)$, it holds that
$S(1^\lambda
, x) \approx_c  \prove(\crs, \mathsf{token}, w, x)$
where $\crs \gets \setup(1^\lambda), \mathsf{token} \gets \delegate(\msk)$.

\fi

\end{document}